\documentclass[english,reqno]{article}
\usepackage{fullpage}
\usepackage{appendix}

\usepackage{amssymb}
\usepackage{amsmath}
\usepackage{amsthm}
\usepackage{dsfont}

\allowdisplaybreaks

\usepackage{xcolor}

\usepackage{mathrsfs}
\usepackage{enumerate}
\usepackage{comment}

\usepackage{physics}

\theoremstyle{remark}
\newtheorem{Remark}{Remark}[section]
\newtheorem{Remarks}{Remarks}

\theoremstyle{plain}
\newtheorem{Theorem}{Theorem}[section]
\newtheorem{Lemma}{Lemma}[section]

\newtheorem{Proposition}{Proposition}[section]
\newtheorem{Assumption}{Assumption}

\DeclareMathOperator{\supp}{supp}

\renewcommand{\tilde}{\widetilde}

\newcommand{\cF}{\mathcal{F}}

\newcommand{\tp}{\widetilde{p}}
\newcommand{\tq}{\widetilde{q}}
\newcommand{\tu}{\widetilde{u}}
\newcommand{\tv}{\widetilde{v}}
\newcommand{\tw}{\widetilde{w}}
\newcommand{\tx}{\widetilde{x}}
\newcommand{\ty}{\widetilde{y}}

\newcommand{\txi}{\widetilde{\xi}}
\newcommand{\teta}{\widetilde{\eta}}
\newcommand{\dx}{\mathrm{d}x}
\newcommand{\dy}{\mathrm{d}y}
\newcommand{\dtx}{\mathrm{d}\tx}
\newcommand{\dty}{\mathrm{d}\ty}
\newcommand{\ddp}{\mathrm{d}p}
\newcommand{\dtp}{\mathrm{d}\tp}
\newcommand{\dq}{\mathrm{d}q}
\newcommand{\dtq}{\mathrm{d}\tq}
\newcommand{\dw}{\mathrm{d}w}
\newcommand{\du}{\mathrm{d}u}
\newcommand{\ddv}{\mathrm{d}v}
\newcommand{\dt}{\mathrm{d}t}
\newcommand{\ds}{\mathrm{d}s}
\newcommand{\dtw}{\mathrm{d}\tw}
\newcommand{\dtu}{\mathrm{d}\tu}

\newcommand{\ii}{\mathrm{i}}

\newcommand{\fkk}{\mathfrak{k}}
\newcommand{\fkl}{\mathfrak{l}}
\newcommand{\fks}{\mathfrak{s}}
\newcommand{\fkK}{\mathfrak{K}}
\newcommand{\fkL}{\mathfrak{L}}

\newcommand{\cR}{\mathcal{R}}

\newcommand{\cN}{\mathcal{N}}

\DeclareRobustCommand{\rchi}{{\mathpalette\irchi\relax}}
\newcommand{\irchi}[2]{\raisebox{\depth}{$#1\chi$}}

\newcommand{\bR}{\mathbb{R}}
\newcommand{\bN}{\mathbb{N}}

\newcommand{\p}{\partial}

\numberwithin{equation}{section} \numberwithin{Lemma}{section}
\numberwithin{Theorem}{section} \numberwithin{Proposition}{section}
\numberwithin{Corollary}{section} \numberwithin{Remark}{section}

\usepackage{enumerate}
\usepackage{hyperref}

\allowdisplaybreaks
\usepackage{mathrsfs}  

\title{Derivation of Vlasov Equations for Multi-Species Fermions}

\date{}

\author{
	Hongshuo Chen\thanks{College of Mathematics and Statistics, Chongqing University, Chongqing 401331, China, \href{mailto:hongshuo.chen@gmail.com}{hongshuo.chen@gmail.com}} 
	\and 
	Li Chen\thanks{Institut f\"ur Mathematik, Universit\"at Mannheim, B6, 68159 Mannheim, Germany, \href{mailto:chen@math.uni-mannheim.de}{chen@math.uni-mannheim.de}} 
	\and 
	Jinyeop Lee\thanks{Department of Mathematics, LMU Munich, Theresienstrasse 39, 80333 Munich, Germany, \href{mailto:lee@math.lmu.de}{lee@math.lmu.de}}%
	\thanks{Department of Mathematics and Computer Science, University of Basel, Spiegelgasse 1, 4051 Basel, Switzerland, \href{mailto:jinyeop.lee@unibas.ch}{jinyeop.lee@unibas.ch}} 
	\and 
	Zhiwei Sun\thanks{School of Mathematical Sciences, Soochow University, 1 Shi-Zi Street, Suzhou 215006, Jiangsu Province, China, \href{mailto:20194007008@stu.suda.edu.cn}{20194007008@stu.suda.edu.cn}}
}

\begin{document}
	
	\maketitle
	
	\begin{abstract}
		We explore a system of two-species fermions with $N$ particles. 
		Through this exploration, we establish a two-species Husimi measure and construct a two-species Bogoliubov–Born–Green–Kirkwood–Yvon (BBGKY) hierarchy hierarchy with error terms.
		We provide proofs of the smallness of these errors in the semiclassical and mean-field limits, fixing $\hbar=N^{-1/3}$. 
		After that, we rigorously establish the uniqueness of this hierarchy. 
		This allows us to conclude that the two-species Husimi measure from the solution of $N$-particle Schr\"odinger equation approximates the solution of the Vlasov equation for two-species systems.
	\end{abstract}
	
	\section{Introduction}
	
	The study of multi-species fermions has been extensive, with numerous applications spanning across plasma physics, solid-state physics, and astrophysics, as documented in various literary works \cite{PhysRevA.78.013613,buchsbaum1960resonance,jungel2009transport,PhysRevLett.104.053202,wang2013decay,PhysRevLett.100.053201}. 
	We consider a system composed of two distinct species. 
	The phase space (normalised) densities for the first and second species are represented by $ \mathrm{f}_{1,t}(q,p) $ and $ \mathrm{f}_{2,t}(\tq,\tp) $, respectively, where these densities are defined at a specific point $ (q,p) $ or $ (\tq,\tp) $ within the phase space $ \mathbb{R}^3 \times \mathbb{R}^3 $ at a given time $ t $ in $ \mathbb{R} $. 
	It is known that the dynamics of $ \mathrm{f}_{1,t}(q,p), \mathrm{f}_{2,t}(\tq,\tp) \in L^1(\mathbb{R}^3\times\mathbb{R}^3) $ are governed by the Vlasov equations:
	\[
	\begin{split}
		\partial_{t}\mathrm{f}_{1,t}(q,p)+p\cdot\nabla_{q}\mathrm{f}_{1,t}(q,p)
		&=\left(n_1(\nabla V_{11}*\varrho_{1,t})(q)+n_2(\nabla V_{12}*\varrho_{2,t})(q)\right)\cdot\nabla_{p}\mathrm{f}_{1,t}(q,p)\\
		\partial_{t}\mathrm{f}_{2,t}(\tq,\tp)+\tp\cdot\nabla_{\tq}\mathrm{f}_{2,t}(\tq,\tp)
		&=\left(n_2(\nabla V_{22}*\varrho_{2,t})(\tq)+n_1(\nabla V_{21}*\varrho_{1,t})(\tq)\right)\cdot\nabla_{\tp}\mathrm{f}_{2,t}(\tq,\tp),
	\end{split}
	\]
	where $ V_{11} $, $ V_{22} $, and $ V_{12} = V_{21} $ describe the interactions within and between the species, while $ \varrho_{1,t} $ and $ \varrho_{2,t} $ are the spatial density distributions of the species, respectively, i.e., $ \varrho_{\alpha,t}(q) =(2\pi)^{-3} \int \mathrm{d}p\, \mathrm{f}_{\alpha,t}(q,p) $ for $ \alpha = 1,2 $.
	Moreover, $ n_1 $ and $ n_2 $ denote the fraction of the first and the second species, respectively.
	
	For notational simplicity, we define $ m_{1,t} := n_1 \mathrm{f}_{1,t} $, $ m_{2,t} := n_2 \mathrm{f}_{2,t} $, $ \rho_{1,t}(q) := n_1\varrho_{1,t}(q) $, and $ \rho_{2,t}(q) := n_2\varrho_{2,t}(q) $. Then we have
	\begin{equation}\label{eq:Vlasov}
		\begin{split}
			\partial_{t}m_{1,t}(q,p)+p\cdot\nabla_{q}m_{1,t}(q,p)&=\left((\nabla V_{11}*\rho_{1,t})(q)+(\nabla V_{12}*\rho_{2,t})(q)\right)\cdot\nabla_{p}m_{1,t}(q,p)\\
			\partial_{t}m_{2,t}(\tq,\tp)+\tp\cdot\nabla_{\tq}m_{2,t}(\tq,\tp)&=\left((\nabla V_{22}*\rho_{2,t})(\tq)+(\nabla V_{21}*\rho_{1,t})(\tq)\right)\cdot\nabla_{\tp}m_{2,t}(\tq,\tp).
		\end{split}
	\end{equation}
	Note that, in this form, the information of the fraction of each species $ n_1 $ and $ n_2 $ is embedded in \eqref{eq:Vlasov} such that $ \|\rho_{1,t}\| = n_1 $ and $ \|\rho_{2,t}\| = n_2 $.
	
	Our goal is to rigorously derive the above Vlasov equations \eqref{eq:Vlasov} from a high-dimensional (many-body) Schrödinger equation, establishing a clear and direct link between microscopic quantum mechanics and the macroscopic behaviors observed in plasma (as well as in stellar systems). 
	
	Now, we consider the following system with two-species fermions:
	\begin{equation}\label{eq:Schrodinger}
		\begin{cases}
			\displaystyle \mathrm{i} \hbar \partial_t \Psi_{N_1, N_2, t} = H_{N_1,N_2} \Psi_{N_1, N_2, t}, \\
			\Psi_{N_1, N_2, 0} = \Psi_{N_1,N_2},
		\end{cases}
	\end{equation}
	where the mean-field Hamiltonian of the system has the form
	
	\begin{equation}\begin{split}
			H_{N_1,N_2} 
			&=
			\frac{\hbar^2}{2}\bigg(\sum_{j=1}^{N_{1}}(-\Delta_{x_{1j}})
			+\sum_{j=1}^{N_{2}}(-\Delta_{x_{2j}})\bigg)\\
			&\qquad+\frac{1}{N
			}\bigg(\sum_{1\leq j<k\leq N_{1}}V_{11}(x_{1 j}-x_{1 k})
			+\sum_{1\leq j<k\leq N_{2}}V_{22}(x_{2 j}-x_{2 k})
			+\sum_{j=1}^{N_{1}}\sum_{k=1}^{N_{2}}V_{12}(x_{1 j}-x_{2 k})
			\bigg)
		\end{split}
	\end{equation}
	where $\hbar=N^{-1/3}$ with $N=N_1+N_2$.
	Here, we have used $x_{\alpha k}$ to denote the position variable of the $k$-th particle of the $\alpha$-th species.
	
	Note that, since we are considering a fermionic system, $\Psi_{N_1, N_2, t}$ is in the partially anti-symmetric subspace $L^2_a(\bR^{3N_1})\otimes L^2_a(\bR^{3N_2})$ of $L^2(\bR^{3(N_1+N_2)})$.
	
	For instance, consider the initial data represented as a Slater determinant:
	\begin{equation}\label{eq:Slater-initial-data}
		\Psi_{N_1,N_2} = \prod_{\alpha=1}^{2} \frac{1}{\sqrt{N_\alpha}} \det\{e_{\alpha i}(x_{\alpha j})\}_{i,j=1}^{N_\alpha}
	\end{equation}
	where $\{e_{\alpha j}\}_{j=1}^{N_\alpha}$ is a family of orthonormal bases in $L^2(\bR^3)$ for each $\alpha=1,2$.
	
	Given that $\Psi_{N_1, N_2, t}$ belongs to the space 
	$L^2_a(\bR^{3N_1})\otimes L^2_a(\bR^{3N_2})$
	and $m_{\alpha,t}$ is an element of $L^1(\bR^6)$, a transformation of $\Psi_{N_1, N_2, t}$ is required to establish a relationship between these two functions. To draw a comparison between the solutions derived from the Schr\"odinger equation and the Vlasov equations, we introduce the \emph{$(k,\ell)$-particle reduced density matrix}. The integral kernel of this matrix is defined as follows:
	\begin{equation}\label{eq:gammakl}
		\begin{aligned}
			&\gamma_{N_1, N_2, t}^{(k,\ell)} (x_1,\dots,x_k,\tilde{x}_1,\dots,\tilde{x}_\ell;y_1,\dots,y_k,\tilde{y}_1,\dots,\tilde{y}_\ell)\\
			&:= \frac{N_1!\,N_2!}{(N_1-k)!\,(N_2-\ell)!} \int_{\bR^{3(N_1-k)}\times\bR^{3(N_2-\ell)}} \dd Z \dd \tilde{Z}\, \\
			&\qquad\overline{\Psi_{N_1, N_2, t}(y_1,\dots,y_k,Z,\tilde{y}_1,\dots,\ty_{\ell},\tilde{Z})} \Psi_{N_1, N_2, t}(x_1,\dots,x_k,Z,\tilde{x}_1,\dots,\tx_{\ell},\tilde{Z}).
		\end{aligned}
	\end{equation}
	This formulation enables a meticulous comparison between the quantum mechanical behavior described by the Schr\"odinger equation and the classical dynamics depicted by the Vlasov equations.
	
	Using the definition of $(k,\ell)$-particle reduced density matrix, we define our \emph{$(k,\ell)$-particle Husimi measure} for two species as
	\begin{equation}\label{def:m_kl}
		\begin{aligned}
			&m^{(k,\ell)}_{N_1,N_2,t} (q_1, p_1, \dots, q_k, p_k,\tq_1, \tp_1, \dots, \tq_\ell, \tp_\ell)\\
			&:=
			\int (\dw\du)^{\otimes k}(\dtw\dtu)^{\otimes \ell} 
			\left( f^\hbar_{q,p}(w) \overline{f^\hbar_{q,p}(u)} \right)^{\otimes k}
			\left( f^\hbar_{\tq,\tp}(\tw)
			\overline{f^\hbar_{\tq,\tp}(\tu)} \right)^{\otimes \ell}\\
			&\qquad\qquad\qquad\times\gamma_{N_1,N_2,t}^{(k,\ell)}(u_1,\dots,u_k,\tu_1, \dots, \tu_\ell; w_1,\dots, w_k, \tw_1, \dots, \tw_\ell)
		\end{aligned}
	\end{equation}
	where the coherent state $f^\hbar_{q,p} $is defined as
	\begin{equation}\label{def_coherent}
		f^{\hbar}_{q, p} (y) := \hbar^{-\frac{3}{4}} f \left(\frac{y-q}{\sqrt{\hbar}} \right) e^{\frac{\mathrm{i}}{\hbar} p \cdot y }
	\end{equation}
	with some real-valued normalized function $f$ in an appropriate Hilbert space.
	
	\medskip
	
	Now, we aim to demonstrate that the expression \( m^{(k,\ell)}_{N_1,N_2,t} \) adheres to the two-species Vlasov hierarchy in the large \( N_1,N_2 \) limit. To accomplish this, we make the following assumptions:
	\begin{Assumption}\label{ass:main}\phantom{ }
		
		\begin{enumerate}[1.]
			\item {\normalfont(Interaction potential)} For $\alpha,\beta\in\{1,2\}$, 
			$\widehat{V}_{\alpha\beta}\in C_0$
			and 
			$V_{\alpha\beta}(x)=V_{\alpha\beta}(-x)$.
			\label{item:Interaction}
			
			\item {\normalfont(Coherent state)} $f \in H^1(\bR^3)\cap L^\infty(\bR^3)$ satisfies $\norm{f}_2 = 1$, and has compact support in $B_{R_1}$ {for a given $R_1>0$}.
			\label{item:Coherent}
			
			\item There exist $n_1,n_2>0$ with $n_1+n_2=1$ such that 
			\[
			\left|\frac{N_\alpha}{N} - n_\alpha\right| \to 0 \quad\text{as }N:=N_1+N_2\to\infty
			\]
			for $\alpha=1,2$.
			\item {\normalfont(Initial data)} $m_{N_1,N_2}^{(1,0)}$ converges weakly to $m_{1}$ in  $L^1(\bR^{6})$ and $m_{N_1,N_2}^{(0,1)}$ converges weakly to $m_{2}$ in  $L^1(\bR^{6})$. Furthermore, they satisfy
			\begin{equation}\label{eq:inip2q_ben}
				\begin{split}
					\int \dq \ddp\; (\,|p|^2 + |q|\,) \, m_{N_1,N_2}^{(1,0)}(q,p) &<\infty\\
					\int \dtq \dtp\; (\,|\tp|^2 + |\tq|\,)\,  m_{N_1,N_2}^{(0,1)}(\tq,\tp) &<\infty
				\end{split}  
			\end{equation}
			uniformly for all $N_1,N_2$.
			\label{item:initialdata}
		\end{enumerate}
	\end{Assumption}
	
	\begin{Remarks}\phantom{ }
		\begin{enumerate}
			\item 
			In Assumption \ref{ass:main}.\ref{item:Interaction}, the condition $\widehat{V}_{\alpha\beta}\in C_0$ is utilized to establish the uniqueness of the infinite hierarchy which is only used in the proof of Proposition \ref{prop:uniqueness_of_solution}. For those interested in addressing more general potentials, we recommend referring to, for example, \cite{Spohn1981}.
			
			\item
			It is likely that Assumption \ref{ass:main}.\ref{item:Coherent} can be extended to a larger class of functions. Specifically, by decomposing $f = f\chi_{|x|\leq R} + f\chi_{|x|>R}$, carrying out the same proof for $f\chi_{|x|\leq R}$, and using the “smallness” of the tail $f\chi_{|x|>R}$, it is conceivable that the same result can be achieved.
		\end{enumerate}
	\end{Remarks}
	
	\begin{Theorem}\label{thm:main}
		Consider the many-body Schrödinger equation \eqref{eq:Schrodinger} with initial data $\Psi_{N_1,N_2}$, yielding the solution $\Psi_{N_1,N_2,t}$ at time $t$. Define the $(k,\ell)$-particle Husimi measure $m_{N_1,N_2,t}^{(k,\ell)}$ associated with $\Psi_{N_1,N_2,t}$ as per \eqref{def:m_kl}. Assuming Assumption \ref{ass:main}, there exists a subsequence such that the two-species Husimi measure $m_{N_1,N_2,t}^{(k,\ell)}$ weakly converges in $L^1(\bR^{6k}\times\bR^{6\ell})$, as $N\rightarrow\infty$, to $m_{\infty,\infty,t}^{(k,\ell)}$, the unique weak solution of the following infinite Vlasov hierarchy in the distributional sense:
		\begin{equation} \label{eq:BBGKY_limit}
			\begin{aligned}
				&\partial_t m_{\infty,\infty,t}^{(k,\ell)}(q_1, p_1, \dots, q_k, p_k,\tq_1, \tp_1, \dots, \tq_\ell, \tp_\ell)\\
				&+ \mathbf{p}_k \cdot \nabla_{\mathbf{q}_k}m_{\infty,\infty,t}^{(k,\ell)}(q_1, p_1, \dots, q_k, p_k,\tq_1, \tp_1, \dots, \tq_\ell, \tp_\ell)
				{ + \mathbf{\tp}_{\ell} \cdot \nabla_{\mathbf{\tq}_{\ell}}m_{\infty,\infty,t}^{(k,\ell)}(q_1, p_1, \dots, q_k, p_k,\tq_1, \tp_1, \dots, \tq_\ell, \tp_\ell)}  \\
				&= \frac{1}{(2\pi)^3} \sum_{j=1}^{k} \nabla_{\mathbf{p}_j} \cdot \iint \dd q_{k+1}\dd p_{k+1} \nabla V_{11}(q_j - q_{k+1})m_{\infty,\infty,t}^{(k+1,\ell)}
				(q_1, p_1, \dots, q_{k+1}, p_{k+1},\tq_1, \tp_1, \dots, \tq_\ell, \tp_\ell)\\
				&+\frac{1}{(2\pi)^3} \sum_{j=1}^{k} \nabla_{\mathbf{p}_j} \cdot \iint \dd \tilde q_{\ell+1}\dd \tilde p_{\ell+1} \nabla V_{12}(q_j - \tilde q_{\ell+1})m_{\infty,\infty,t}^{(k,\ell+1)}
				(q_1, p_1, \dots, q_k, p_k,\tq_1, \tp_1, \dots, \tq_{\ell+1}, \tp_{\ell+1})\\
				&+ \frac{1}{(2\pi)^3} \sum_{j=1}^{\ell}  \nabla_{\mathbf{\tp}_j} \cdot \iint \dd q_{k+1}\dd p_{k+1} \nabla V_{21}(q_{k+1} - \tq_{j})m_{\infty,\infty,t}^{(k+1,\ell)}
				(q_1, p_1, \dots, q_{k+1}, p_{k+1},\tq_1, \tp_1, \dots, \tq_\ell, \tp_\ell)\\
				&+\frac{1}{(2\pi)^3} \sum_{j=1}^{\ell} \nabla_{\mathbf{\tp}_j} \cdot \iint \dd \tilde q_{\ell+1}\dd \tilde p_{\ell+1}\nabla V_{22}(\tq_j - \tilde q_{\ell+1})m_{\infty,\infty,t}^{(k,\ell+1)}
				(q_1, p_1, \dots, q_k, p_k,\tq_1, \tp_1, \dots, \tq_{\ell+1}, \tp_{\ell+1}).
			\end{aligned}
		\end{equation}
	\end{Theorem}
	
	\begin{Theorem}\label{thm:main2}
		Under the assumptions previously in Theorem \ref{thm:main}, we additionally consider the factorization of the initial data, i.e., for all $(k,\ell) \in \bN \times \bN \setminus \{(0,0)\}$,
		\begin{equation}
			\label{IC}
			\|m_{N_1,N_2}^{(k,\ell)}- {n_1^k n_2^\ell} m_{1}^{\otimes k}\otimes m_{2}^{\otimes \ell}\|_{L^1}\rightarrow 0, \quad \text{as } N\rightarrow\infty,
		\end{equation}
		the $(1,0)$-particle Husimi measure $m_{N_1,N_2,t}^{(1,0)}$ and the $(0,1)$-particle Husimi measure $m_{N_1,N_2,t}^{(0,1)}$ associated with $\Psi_{N_1,N_2,t}$ converge in the 1-Wasserstein distance to the solutions $m_{1,t}$ and $m_{2,t}$ of 
		\eqref{eq:Vlasov}
		
		with initial data $m_{1,0}=m_1$ and $m_{2,0}=m_2$, respectively.
	\end{Theorem}
	
	The theorem demonstrates that in an $N$-particle quantum system, encompassing electrons and ions as quantum particles, the model conforms to Vlasov equations expressed through Husimi measures as $N\rightarrow\infty$. In fact, the theorem extends its applicability to any finite multi-species system. For brevity and clarity, however, our presentation and discussion will concentrate on the context of a two-species system.
	
	\vspace{1em}
	
	\subsection*{History}
	As has been mentioned, the multi-species Vlasov equations have been studied in many previous works. In this part, we aim to focus on the historical evolution of deriving this system from a microscopic framework. Therefore, only the mean-field limit and the semi-classical limit, as well as simultaneous limits, are going to be discussed.
	
	By applying the mean-field limit, i.e., taking $N\to\infty$ and fixing $\hbar=1$, for one-species case, it is possible to derive a Hartree-Fock $N$-body quantum system from the interacting $N$-body quantum system. It is noteworthy that, at this stage, the system is not yet classical but retains its quantum nature. Nonetheless, comprehending the workings of the mean-field limit is crucial.
	The mean-field limit results have been established in \cite{ELGART20041241} for fixed values of $\hbar$, focusing on a short-time scenario. The convergence rates pertaining to both the trace norm and the Hilbert-Schmidt norm were accomplished in \cite{benedikter2014mean}, particularly when the initial data approximates the Slater determinant. In subsequent developments, the exploration expands to encompass mixed state initial data, as investigated in \cite{benediktermixed,benedikter2014rel}. Furthermore, the pursuit of convergence rates has extended its scope to include Coulomb and Riesz potentials, as in \cite{Porta2017,saffirio2017mean}.
	Recently, there were series of works for high density regime \cite{fresta2023cmp,fresta2024effective}.
	For a more comprehensive understanding and a broader array of references regarding this subject, we encourage an exploration of \cite{Frohlich2011,petrat2014derivation,Petrat2017,Petrat2016ANM,benedikter2022effective}, along with the associated references.
	
	The emergence of the Vlasov equation is rooted in the semi-classical limit, i.e., taking $\hbar\to0$, employed either in the context of the Hartree or the Hartree-Fock equations. The foundational exploration of this process involves the utilization of the Wigner measure for smooth potentials, with initial contributions found in \cite{Lions1993}. Notably, the Wigner transform of the Hartree (or Hartree-Fock) equation aligns in structure with the Vlasov equation, as seen in \cite[Eq. (6.15)]{Benedikter2016book}. The convergence rate analysis in \cite{benedikter2016hartree} were given  in the trace norm and Hilbert-Schmidt norm, incorporating regularity assumptions for mixed state initial data and a specific potential class. The employed $k$-particle Wigner measure is central to works such as \cite{Lions1993,benedikter2016hartree}.
	Recent researches provide the semi-classical limit field for inverse power law potentials \cite{Saffirio2019}, rates of convergence within Schatten norms \cite{lafleche2020strong}, and the convergence of mixed states within the influence of Coulomb potentials \cite{saffirio2019hartree}. These developments have been extended also to the relativistic fermionic systems \cite{Dietler2018}. Further analyses of the semi-classical limit are comprehensively presented in \cite{amour2013classical,amour2013,Athanassoulis2011,benedikter2022effective,Gasser1998, Markowich1993}.
	The Vlasov equation with Coulomb potential has been also derived from a Newtonian $N$-particle system  in \cite{lazarovici2017mean}. 
	
	For the one-species case, derivations of the Vlasov equations \eqref{eq:Vlasov} from the $N$-body Schrödinger equation \eqref{eq:Schrodinger} with the combined mean-field and the semi-classical limits were accomplished by Narnhofer and Sewell, as well as by Spohn, in \cite{Narnhofer1981,Spohn1981}, with $\hbar = N^{-1/3}$. The combined limit's rate of convergence was investigated in \cite{Golse2017,Golse2021,golse:hal-01334365} through the utilization of Wasserstein (pseudo-) distance. A rate of convergence was established in \cite{graffi2003mean}, operating within a weak formulation framework and avoiding any assumptions on $\hbar = N^{-1/3}$. While the Wigner measure lacks the status of a proper probability measure due to the potential presence of negative points, the Husimi measure, formed through the convolution of the Wigner measure with a Gaussian function, emerges as a nonnegative probability measure \cite{Combescure2012,Fournais2018,Zhang2008}. More precisely, as has been explained in \cite[p.21]{Fournais2018}, the relationship between the Husimi and Wigner measures can be expressed through convolution, where specific Gaussian coherent states are involved.  Furthermore, under a generalized Husimi measure framework, \cite{Chen2021} achieved convergence results for regular potentials. 
	Recently, \cite{Chong2021} managed to tackle the combined limit for cases involving singular potentials based on regular mixed state initial data.
	
	The Boltzmann equation for a mixture has been rigorously derived for two types of hard sphere gases, as detailed in \cite{ampatzoglou2022rigorous}. However, a comprehensive derivation of the multi-species Vlasov system starting from the many-body Schrödinger equation has not yet been explored in existing literature, to the best of our knowledge.
	
	\subsection*{Organization of the paper}
	This paper is organized as follows: In Section \ref{sec:FockSpace}, we provide a comprehensive review and formalization of the multi-species Fock space. Building on this foundation, we define the two-species $(k,\ell)$-particle Husimi measure, intricately linked to the two-species wave function. Subsequently, we present properties of the $(k,\ell)$-particle Husimi measure and offer bounds on the number operator and kinetic energy operator at time $t$ for the wave function $\Psi_{N_1,N_2,t}$. In Section \ref{sec:BBGKY}, we establish the two-species Bogoliubov–Born–Green–Kirkwood–Yvon (BBGKY) hierarchy of $(k,\ell)$-particle Husimi measures. Leveraging the tools developed in Section \ref{sec:FockSpace}, we prove the smallness of errors as $N\to\infty$ in Section \ref{sec:Bound of errors}. The final section, Section \ref{sec:Uniqueness Vlasov hierarchy}, is dedicated to proving the two main theorems, including the uniqueness of the two-species BBGKY hierarchy of $(k,\ell)$-particle Husimi measures. 
	
	\section{Fock space formalism and uniform estimates}\label{sec:FockSpace}
	
	\subsection{Fock space formalism for fermions with two-species}
	
	To avoid a lengthy paper, we cite the single-species Fock space
	formalism from previous works (e.g. \cite{Chen2021AHP}).
	
	We now consider the fermionic Fock space $\mathcal{F}_a=\bigoplus_{n=0}^{\infty}L^{2}(\mathbb{R}^{3})^{\otimes_{a}n}$:
	
	\[
	\cF^{(N)}_a = 0 \oplus \dots \oplus 0 \oplus L^{2}(\mathbb{R}^{3})^{\otimes_{a}N} \oplus 0 \oplus \dots \subset \cF_a.
	\]
	In our model, the state space for fermions with two species is  $\mathcal{F}_a {\otimes} \mathcal{F}_a.$
	Here, we use the standard tensor product of two Hilbert spaces.
	
	For $f\in L^{2}(\mathbb{R}^{3})$, we denote creation and annihilation
	operators for the first species to be
	$a^{*}(f)$ and $a(f)$, respectively on $\mathcal{F}_a$ such that
	\begin{align*}
		\big(a(f) \Psi \big)^{(n)} &:= \sqrt{n+1} \int \dd{x} \overline{f(x)} \psi^{(n+1)} (x, x_1, \dots, x_n),\\
		\big(a^*(f) \Psi  \big)^{(n)} &:= \frac{1}{\sqrt{n}} \sum_{j=1}^n f(x_j) \psi^{(n-1)} (x_1, \dots, {x}_{j-1}, {x}_{j+1} , \dots , x_n).
	\end{align*}
	Here, $\psi^{(n)}$ denotes the $n$-th particle sector of $\Psi = \{\psi^{(n)}\}_{n\geq0} \in \mathcal{F}_a$. 
	Following the notations from \cite{benedikter2014mean}, we will use the operator-valued distributions $a^*_x$ and $a_x$, to represent the creation and annihilation operators:
	\begin{equation}\label{distrb_annil_creation_def}
		a^*(f) = \int \dd x\, f(x) a^*_x, \quad a(f) = \int \dd x\  \overline{f(x)}a_x.
	\end{equation}
	Note that the operator-valued distribution $a^*_x$ formally creates a particle at position $x\in \bR^3$, whilst the operator-value distribution $a_x$ annihilates a particle at $x$.
	For the second species, we denote creation and annihilation
	operators to be $b^{*}(f)$, $b(f)$, respectively.
	
	The operator-valued distributions within the same species adhere to canonical anti-commutation relations (CAR) as follows:
	\begin{equation}\label{eq:CAR}
		\begin{split}
			0&=\{a^*_x,a^*_y\}=\{a_x,a_y\}=\{b^*_{\tx},b^*_{\ty}\}=\{b_{\tx},b_{\ty}\},\\
			\delta_{x=y}&=\{a^*_x,a_y\},\ \delta_{\tx=\ty}=\{b^*_{\tx},b_{\ty}\},
		\end{split}
	\end{equation}
	for $x,y,\tx,\ty\in \bR^3$, where the anti-commutator is defined as $\{A, B\} = AB + BA$.
	Conversely, for those between different species, they commute, i.e.,
	\begin{equation}\label{eq:CR}
		0=[a^*_x,b^*_{\ty}]=[a_x,b_{\ty}]=[a^*_x,b_{\ty}]=[b^*_{\tx},a_y],
	\end{equation}
	for $x,y,\tx,\ty\in \bR^3$, where the commutator is defined as $[A, B] = AB - BA$.
	
	Following the notation in \cite{benedikter2014mean}, we can write the corresponding Hamiltonian in terms of the operator-valued distributions $a_x$, $a^*_x$, $b_{\ty}$, and $b_{\ty}^*$ by
	\begin{equation}\label{eq:Fock_Hamil}
		\mathcal{H}_N =\sum_{\alpha=1}^{2} \mathcal{H}_{\alpha}+\frac{1}{N} \int \dx \dty\, V_{12}(x-\ty) a^*_x b^*_{\ty}\, b_{\ty} a_x,
	\end{equation}
	where
	\begin{align*}
		\mathcal{H}_1
		&= \frac{\hbar^2}{2} \int \dx\, \nabla_x a^*_x \nabla_x a_x+\frac{1}{2N} \int \dx\dy \, V_{11}(x-y) a^*_x a^*_y a_y a_x,\\
		\mathcal{H}_2
		&= \frac{\hbar^2}{2} \int \dtx\, \nabla_{\tx} b^*_{\tx} \nabla_{\tx} b_{\tx}+\frac{1}{2N} \int \dtx\dty\, V_{22}(\tx-\ty) b^*_{\tx} b^*_{\ty}\, b_{\ty} b_{\tx}.
	\end{align*}
	Next, we let $\mathcal{N}_{\alpha}$
	to denote the number of particles of {species} $\alpha$, i.e.,
	\begin{equation}\label{eq:number-op}
		\mathcal{N}_1 = \int \dx\, a_{x}^{*}a_{x}
		\quad\text{and}\quad
		\mathcal{N}_2 = \int \dtx\, b_{\tx}^{*}b_{\tx}.
	\end{equation}
	Then, we define
	\[
	\mathcal{N}_{\mathrm{total}}:=\mathcal{N}_1+\mathcal{N}_2.
	\]
	We also define the kinetic energy operators as follows:
	\begin{align}\label{eq:kinetic-op}
		\mathcal{K}_1 = \frac{\hbar^2}{2} \int \dx\, a_{x}^{*}(-\Delta)a_{x}
		\quad\text{and}\quad
		\mathcal{K}_2 = \frac{\hbar^2}{2} \int \dtx\, b_{\tx}^{*}(-\Delta)b_{\tx}.
	\end{align}
	Then, we define
	\[
	\mathcal{K}_{\mathrm{total}}:=\mathcal{K}_1+\mathcal{K}_2.
	\]
	Moreover, with this second quantization formalism, for $\Psi_{N_1, N_2, t} \in \mathcal{F}_a^{(N_1)}\otimes \mathcal{F}_a^{(N_2)}$, we define 
	
	\begin{equation}
		\begin{aligned}
			&\gamma^{(k,\ell)}_{N_1,N_2,t}(u_1,\dots,u_k,\tu_1, \dots, \tu_\ell; w_1,\dots, w_k, \tw_1, \dots, \tw_\ell)\\
			&:= \langle\Psi_{N_1, N_2, t}, a^*_{w_1} \cdots a^*_{w_k} b^*_{\tw_1} \cdots b^*_{\tw_k} b_{\tu_k} \cdots b_{\tu_1} a_{u_k} \cdots a_{u_1}\Psi_{N_1, N_2, t}\rangle.
		\end{aligned}
	\end{equation}
	For example, one can express
	\begin{equation}
		\begin{split}
			\gamma^{(1,0)}_{N_1,N_2,t}(u;w)
			:=\langle\Psi_{N_1, N_2, t},a_{w}^{*}a_{u}\Psi_{N_1, N_2, t}\rangle
			\quad\text{and}\quad
			\gamma^{(0,1)}_{N_1,N_2,t}(\tu;\tw)
			:=\langle\Psi_{N_1, N_2, t},b_{\tw}^{*}b_{\tu}\Psi_{N_1, N_2, t}\rangle.
		\end{split}
	\end{equation}
	
	\subsection{Husimi measure for multi-species fermions}
	Recall from \eqref{def_coherent} that the coherent state $f^\hbar_{q,p}$ is given by
	\begin{equation}
		f^{\hbar}_{q, p} (y) := \hbar^{-\frac{3}{4}} f \left(\frac{y-q}{\sqrt{\hbar}} \right) e^{\frac{\mathrm{i}}{\hbar} p \cdot y }.
	\end{equation}
	Then, the projection of a coherent state is given by
	\begin{equation*}\label{projection_f}
		\frac{1}{(2\pi \hbar)^3} \int \dq \ddp\, |f^{\hbar}_{q, p}\rangle\langle{f^{\hbar}_{q, p} }|  = \mathds{1}.
	\end{equation*}
	For any $\Psi_{N_1, N_2, t} \in \mathcal{F}_a^{(N_1)}\otimes \mathcal{F}_a^{(N_2)}$, $1 \leq k \leq N_1$ and $t \geq 0$, the $k$-particle Husimi measure for the first species is defined as 
	\begin{equation}\label{husimi_def_1}
		\begin{aligned}
			&m^{(k,0)}_{N_1,N_2,t} (q_1, p_1, \dots, q_k, p_k)\\
			&
			:= \int (\dw\du)^{\otimes k} \left( f^\hbar_{q,p}(w) \overline{f^\hbar_{q,p}(u)} \right)^{\otimes k}  \langle \Psi_{N_1, N_2, t},  a^*_{w_1} \cdots a^*_{w_k} a_{u_k} \cdots a_{u_1} \Psi_{N_1, N_2, t} \rangle\\
			&
			=\int (\dw\du)^{\otimes k} \left( f^\hbar_{q,p}(w) \overline{f^\hbar_{q,p}(u)} \right)^{\otimes k} \gamma_{N_1,N_2,t}^{(k,0)}(u_1, \dots, u_k; w_1, \dots, w_k),
		\end{aligned}
	\end{equation}
	where we use the short notations
	\[
	(\dw\du)^{\otimes k} := \dw_1 \du_1 \cdots \dw_k \du_k, \text{ and }
	\left( f^\hbar_{q,p}(w) \overline{f^\hbar_{q,p}(u)} \right)^{\otimes k} := \prod_{j=1}^k f^\hbar_{q_j,p_j}(w_j) \overline{f^\hbar_{q_j,p_j}(u_j)}.
	\]
	Intuitively speaking, one can understand that the Husimi measure defined in \eqref{husimi_def_1} measures how many of the first-species particles are in the $k$-semiclassical boxes with a length scale of $\sqrt{\hbar}$ centered at its respective phase-space pairs, $(q_1, p_1), \dots, (q_k, p_k)$. 
	Similarly, for $1 \leq k \leq N_2$, Husimi measure for the second species is:
	\begin{equation}\label{husimi_def_2}
		\begin{aligned}
			&m^{(0,\ell)}_{N_1,N_2,t} (\tq_1, \tp_1, \dots, \tq_k, \tp_k)\\
			&
			:= \int (\dtw\dtu)^{\otimes \ell} \left( f^\hbar_{\tq,\tp}(\tw) \overline{f^\hbar_{\tq,\tp}(\tu)} \right)^{\otimes \ell}  \langle \Psi_{N_1, N_2, t},  b^*_{\tw_1} \cdots b^*_{\tw_\ell} b_{\tu_\ell} \cdots b_{\tu_1} \Psi_{N_1, N_2, t} \rangle\\
			&
			= \int (\dtw\dtu)^{\otimes \ell} \left( f^\hbar_{\tq,\tp}(\tw) \overline{f^\hbar_{\tq,\tp}(\tu)} \right)^{\otimes \ell} \gamma_{N_1,N_2,t}^{(0,\ell)}(\tu_1, \dots, \tu_\ell; \tw_1, \dots, \tw_\ell).
		\end{aligned}
	\end{equation}
	Then we define our $(k,\ell)$-particle Husimi measure for two species as
	\begin{equation}\label{husimi_def_3}
		\begin{aligned}
			&m^{(k,\ell)}_{N_1,N_2,t} (q_1, p_1, \dots, q_k, p_k,\tq_1, \tp_1, \dots, \tq_\ell, \tp_\ell)\\
			&:=
			\int (\dw\du)^{\otimes k}(\dtw\dtu)^{\otimes \ell} 
			\left( f^\hbar_{q,p}(w) \overline{f^\hbar_{q,p}(u)} \right)^{\otimes k}
			\left( f^\hbar_{\tq,\tp}(\tw)
			\overline{f^\hbar_{\tq,\tp}(\tu)} \right)^{\otimes \ell}\\
			&\qquad\qquad\qquad\times\gamma_{N_1,N_2,t}^{(k,\ell)}(u_1,\dots,u_k,\tu_1, \dots, \tu_\ell; w_1,\dots, w_k, \tw_1, \dots, \tw_\ell).
		\end{aligned}
	\end{equation}
	
	The following lemma provides properties of Husimi measure for multi-species fermions.
	\begin{Lemma} \label{lem:prop_kHusimi}
		Suppose that $\Psi_{N_1, N_2, t} \in \mathcal{F}_a^{(N_1)}\otimes \mathcal{F}_a^{(N_2)}$ is normalized for any $t \geq 0$. Then, the following properties hold for $m^{(k,\ell)}_{N_1,N_2,t}$, for $0 \leq k \leq N_1$ and $0 \leq \ell \leq N_2$ except $k = \ell = 0$,
		\begin{enumerate}
			\item (Indistinguishability for each species)
			\[m^{(k,\ell)}_{N_1,N_2,t}(q_1,p_1,\dots,q_k,p_k,\tq_1,\tp_1,\dots,\tq_\ell,\tp_\ell)\] is symmetric with respect to both non-tilde and tilde variables,
			\item (Passivity and bound)
			\begin{align*}0 \leq  m^{(k,\ell)}_{N_1,N_2,t}(q_1,p_1,\dots,q_k,p_k,\tq_1,\tp_1,\dots,\tq_\ell,\tp_\ell) \leq 1 \text{ a.e.},
			\end{align*}
			\item ($L^1$-norm)\label{L^1-norm}
			\begin{align*}
				&\frac{1}{(2\pi)^{3(k+\ell)}} \int (\dq\ddp)^{\otimes k} (\dtq\dtp)^{\otimes \ell}\, m^{(k,\ell)}_{N_1,N_2,t}(q_1,p_1,\dots,q_k,p_k,\tq_1,\tp_1,\dots,\tq_\ell,\tp_\ell)\\
				&= \frac{N_1 (N_1-1)\cdots (N_1-k+1)}{N^k} \frac{N_2 (N_2-1)\cdots (N_2-\ell+1)}{N^\ell},
			\end{align*}
			\item (Recursion formula) \label{prop_k_1}
			\begin{align*}
				&\frac{1}{(2\pi)^{3}} \int \dq_k  \ddp_k\, m^{(k,\ell)}_{N_1,N_2,t}(q_1,p_1,\dots,q_k,p_k,\tq_1,\tp_1,\dots,\tq_\ell,\tp_\ell)\\
				&= \frac{N_1-k+1}{N}  m^{(k-1,\ell)}_{N_1,N_2,t}(q_1,p_1,\dots,q_{k-1},p_{k-1},\tq_1,\tp_1,\dots,\tq_\ell,\tp_\ell),
			\end{align*}
			\begin{align*}
				&\frac{1}{(2\pi)^{3}} \int \dtq_\ell \dtp_\ell\, m^{(k,\ell)}_{N_1,N_2,t}(q_1,p_1,\dots,q_k,p_k,\tq_1,\tp_1,\dots,\tq_\ell,\tp_\ell)\\
				&= \frac{N_2-\ell+1}{N} m^{(k,\ell-1)}_{N_1,N_2,t}(q_1,p_1,\dots,q_k,p_k,\tq_1,\tp_1,\dots,\tq_{\ell-1},\tp_{\ell-1}).
			\end{align*}
			
		\end{enumerate}
	\end{Lemma}
	
	\begin{proof}
		The proofs are not so different from the one-species case \cite{Fournais2018}. We only need to be careful that the $N$ in the denominators in items \ref{L^1-norm} and \ref{prop_k_1} are from $\hbar^3$. So it is not $N_\alpha$ but $N$.
	\end{proof}
	
	\subsection{Uniform in $N$ estimates}\label{sec:uniform_estimates}
	
	In this section, we will discuss several fundamental properties of the number operator and kinetic operator. These operators operate on the solution of \eqref{eq:Schrodinger} at a specific time, denoted as $t$.
	
	According to the following lemma, we can verify that each $k$-th moment of the number operator acting on the solution of the Schr\"odinger equation \eqref{eq:Schrodinger} equals $N^k$.
	
	Because $\Psi_{N_1,N_2,t}$ are the eigenfunctions of the number operators, we have the following lemma:
	\begin{Lemma}\label{lem:estimate-num-op-k}
		Let $\Psi_{N_1, N_2, t} \in \mathcal{F}_a^{(N_1)}\otimes \mathcal{F}_a^{(N_2)}$
		be the solution of \eqref{eq:Schrodinger} with initial data
		\[
		\Psi_{N_1,N_2} = (0,\dots,0,\psi_{N_1},0,\dots)\otimes (0,\dots,0,\psi_{N_2},0,\dots)
		\]
		where
		$\psi_{N_\alpha}\in L^2_a(\bR^{3N_\alpha})$ for $\alpha=1,2$ and
		$\|\Psi_{N_1,N_2}\|=1$.
		Refer to the number operator $\cN_\alpha$ defined in \eqref{eq:number-op}. Then, for finite $0 \leq k \leq N_1$ and $0 \leq \ell \leq N_2$,
		we have
		\[
		\left< \Psi_{N_1, N_2, t}, \frac{\cN_1^k}{N_1^k}\frac{\cN_2^\ell}{N_2^\ell} \Psi_{N_1, N_2, t} \right>  = 1.
		\]
	\end{Lemma}

	Now we prove that the kinetic energy for typical particles is not growing too fast.
	
	\begin{Lemma}\label{lem:kinetic_finite}
		Let $\alpha=1$ and $\beta=2$ or vice versa. Assume $\nabla V_{\alpha\alpha}$ and $\nabla V_{\alpha\beta}$ are bounded, then the kinetic energy is bounded such that
		\begin{equation}\label{k_kinetic_bounded}
			\left<\Psi_{N_1, N_2, t}, \frac{\mathcal{K}_\alpha}{N} \Psi_{N_1, N_2, t}\right> 
			\leq C_\alpha \left<\Psi_{N_1,N_2}, \frac{\mathcal{K}_\alpha}{N} \Psi_{N_1,N_2}\right> + C_\alpha{t^2},
		\end{equation}
		
		where the constants only depend on $\|\nabla V_{\alpha\alpha}\|_{\infty}$, $\|\nabla V_{\alpha\beta}\|_{\infty}$.
	\end{Lemma}
	
	\begin{proof}
		
		To demonstrate the proof, let's begin by considering the specific case where $\alpha=1$ and $\beta=2$. The proof for the case when $\alpha=2$ and $\beta=1$ can be easily derived by following the steps outlined below:
		
		Starting from the Schrödinger equation, we obtain
		\begin{equation}\label{eq:kinetic_finite_2}
			\begin{aligned}
				&i \hbar \frac{\dd }{\dt} \left<\Psi_{N_1, N_2, t} , \mathcal{K}_1 \Psi_{N_1, N_2, t} \right>\\
				&= \left<\Psi_{N_1, N_2, t}, [\mathcal{K}_1,\mathcal{H}_N] \Psi_{N_1, N_2, t} \right>\\
				&= \left<\Psi_{N_1, N_2, t}, \left[\frac{\hbar^2}{2}\sum_{l=1}^{N_{1}}(-\Delta_{x_{1 l}}),\frac{1}{N}\sum_{1\leq j<k\leq N_{1}}V_{1 1}(x_{1 j}-x_{1 k})+\frac{1}{N}\sum_{j=1}^{N_{1}}\sum_{k=1}^{N_{2}}V_{1 2}(x_{1 j}-\tx_{2 k})\right] \Psi_{N_1, N_2, t} \right>\\
				&= \frac{i\hbar^2}{N} \Im \sum_{l=1}^{N_{1}}\left<\nabla_{x_{1 l}}\Psi_{N_1, N_2, t},\sum_{1\leq j<k\leq N_{1}}\left(\nabla_{x_{1 l}}V_{1 1}(x_{1 j}-x_{1 k}) \right)\Psi_{N_1, N_2, t} \right>\\
				& \quad+ \frac{i\hbar^2}{N} \Im \sum_{l=1}^{N_{1}}\left<\nabla_{x_{1 l}}\Psi_{N_1, N_2, t},\sum_{j=1}^{N_{1}}\sum_{k=1}^{N_{2}}\left(\nabla_{x_{1 l}}V_{1 2}(x_{1 j}-\tx_{2 k}) \right)\Psi_{N_1, N_2, t} \right>\\
				&= \frac{i\hbar^2}{N} \Im  \int \dx\dy \left<a_y\nabla_x a_x\Psi_{N_1, N_2, t},\left(\nabla_x V_{1 1}(x-y) \right)a_y a_x\Psi_{N_1, N_2, t} \right>\\
				& \quad+ \frac{i\hbar^2}{N} \Im \int \dx\dy  \left<b_{\ty}\nabla_x a_x\Psi_{N_1, N_2, t},\left(\nabla_x V_{1 2}(x-\ty)\right)b_{\ty} a_x \Psi_{N_1, N_2, t} \right>\\
			\end{aligned}
		\end{equation}
		Then we have that
		\begin{align*}
			&\frac{1}{N} \frac{\dd }{\dt}  \left<\Psi_{N_1, N_2, t} , \mathcal{K}_1 \Psi_{N_1, N_2, t} \right>\\
			&= \frac{\hbar}{N^2} \Im  \int \dx\dy \left<a_y\nabla_x a_x\Psi_{N_1, N_2, t},\left(\nabla_x V_{1 1}(x-y) \right)a_y a_x\Psi_{N_1, N_2, t} \right>\\
			&\quad+ \frac{\hbar}{N^2} \Im \int \dx\dy  \left<b_{\ty}\nabla_x a_x\Psi_{N_1, N_2, t},\left(\nabla_x V_{1 2}(x-\ty)\right)b_{\ty} a_x \Psi_{N_1, N_2, t} \right>.
		\end{align*}
		Now, we compute
		\begin{align*}
			& \left|\frac{\hbar}{N^2} \int \dx\dy \left<a_y\nabla_x a_x\Psi_{N_1, N_2, t},\left(\nabla_x V_{1 1}(x-y) \right)a_y a_x\Psi_{N_1, N_2, t} \right>\right| \\
			&\leq  \frac{C\hbar}{N^2}  \norm{\nabla V_{11}}_{L^\infty} 
			\left( \int \dx\dy  
			\left< \Psi_{N_1, N_2, t},\nabla_x a^*_x a^*_y   a_y \nabla_x a_x \Psi_{N_1, N_2, t} \right>
			\right)^\frac{1}{2} 
			\left(\int \dx\dy  
			\left< \Psi_{N_1, N_2, t}, a^*_x a^*_y   a_y  a_x \Psi_{N_1, N_2, t} \right> 
			\right)^\frac{1}{2} \\
			&\leq C \norm{\nabla V_{11}}_{L^\infty} \left( \frac{\hbar^2 }{N} \int \dx  \left< \Psi_{N_1, N_2, t},\nabla_x a^*_x \frac{\cN_1}{N} \nabla_x a_x \Psi_{N_1, N_2, t} \right> \right)^\frac{1}{2} 
			\left< \Psi_{N_1, N_2, t}, \frac{\cN_1^2}{N_1^2} \Psi_{N_1, N_2, t} \right>^\frac{1}{2}\\
			&\leq C \norm{\nabla V_{11}}_{L^\infty}  \left< \Psi_{N_1, N_2, t},\frac{\mathcal{K}_1}{N} \Psi_{N_1, N_2, t} \right> ^\frac{1}{2}.
		\end{align*}
		Similarly,
		\begin{align*}
			& \left| \frac{\hbar}{N^2} \int \dx\dy  \left<b_{\ty}\nabla_x a_x\Psi_{N_1, N_2, t},\left(\nabla_x V_{1 2}(x-\ty)\right)b_{\ty} a_x \Psi_{N_1, N_2, t} \right> \right| \\
			&\leq C \norm{\nabla V_{12}}_{L^\infty} \left( \frac{\hbar^2 }{N} \int \dx  \left< \Psi_{N_1, N_2, t},\nabla_x a^*_x \frac{\cN_2}{N} \nabla_x a_x \Psi_{N_1, N_2, t} \right> \right)^\frac{1}{2} \left< \Psi_{N_1, N_2, t}, \frac{\cN_1 \cN_2}{N^2} \Psi_{N_1, N_2, t} \right>^\frac{1}{2}\\
			&\leq C \norm{\nabla V_{12}}_{L^\infty} \left< \Psi_{N_1, N_2, t},\frac{\mathcal{K}_1}{N} \Psi_{N_1, N_2, t} \right>^\frac{1}{2}.
		\end{align*}
		Therefore, we get
		\[
		\left< \Psi_{N_1, N_2, t},\frac{\mathcal{K}_1}{N} \Psi_{N_1, N_2, t} \right>
		\leq 
		C \left< \Psi_{N_1,N_2},\frac{\mathcal{K}_1}{N} \Psi_{N_1,N_2} \right>
		+
		C  
		(\norm{\nabla V_{11}}_{L^\infty}^2+\norm{\nabla V_{12}}_{L^\infty}^2)\,{t^2}.
		\]
		
		This gives the desired result.
		
	\end{proof}
	
	\begin{Proposition} \label{prop:finite-moment} For $t \geq 0$, let $m^{(k,\ell)}_{N_1, N_2, t}$ to be the $(k,\ell)$-particle Husimi measure. Assume $\nabla V_{\alpha\beta}$ is bounded for $\alpha, \beta \in \{1, 2\}$. Denoting the phase-space vectors $\mathbf{q}_k = (q_1, \dots, q_k)$ and $\mathbf{p}_k = (p_1, \dots, p_k)$, we have the following finite moments,
		\[
		\int(\dd q\dd p)^{\otimes k}(\dd \tq\dd \tp)^{\otimes \ell}\big(|\mathbf{q}_{k}|+|\mathbf{p}_{k}|^{2}+|\tilde{\mathbf{q}}_{\ell}|+|\tilde{\mathbf{p}}_{\ell}|^{2}\big)m_{N_1, N_2, t}^{(k,\ell)}(q_1,p_1,\dots,q_k,p_k,\tq_1,\tp_1,\dots,\tq_\ell,\tp_\ell)\leq C(1+t^{3})
		\]
		where $C$ is a constant dependent on $k$, $\iint \dq_1 \ddp_1 (|q_1| + |p_1|^2 ) m^{(1,0)}_{N_1,N_2}(q_1, p_1) $, $\iint \dtq_1 \dtp_1 (|\tq_1| + |\tp_1|^2) m^{(0,1)}_{N_1,N_2}(\tq_1, \tp_1) $, and $\|\nabla V_{\alpha\beta}\|_\infty$ for $\alpha, \beta \in \{1, 2\}$.
	\end{Proposition}
	
	\begin{proof}
		For the case of $(k,\ell)=(1,0)$, we can apply the findings from \cite[Proposition 2.3]{Chen2021} and obtain the following:
		\begin{equation}\label{p_1^2_bdd}
			\begin{aligned}
				\left<\Psi_{N_1, N_2, t},\frac{\mathcal{K}_1}{N}\Psi_{N_1, N_2, t} \right> =  \frac{1}{(2\pi)^3}\iint \dd q_1\dd p_1\ |p_1|^2  m_{N_1, N_2, t}^{(1,0)}(q_1, p_1)
				+ \hbar \int \dd q \left| \nabla f\left(q\right) \right|^2
			\end{aligned}
		\end{equation}
		and
		\begin{equation}\label{q_1_bdd}
			\p_t \iint \dd q_1 \dd p_1\ |q_1| m_{N_1, N_2, t}^{(1,0)}(q_1,p_1) \leq (2\pi)^3  \left(1 + 2\left<\Psi_{N_1,N_2},\frac{\mathcal{K}_1}{N} \Psi_{N_1,N_2} \right> + Ct^2 + C \sqrt{\hbar} \right).
		\end{equation}
		Utilizing \eqref{p_1^2_bdd} and referencing Lemma \ref{lem:kinetic_finite}, we can deduce:
		\begin{equation}
			\begin{aligned}
				&\frac{1}{(2\pi)^3}\iint \dd q_1\dd p_1\ |p_1|^2  m_{N_1, N_2, t}^{(1,0)}(q_1, p_1)\leq\left<\Psi_{N_1, N_2, t}, \frac{\mathcal{K}_1}{N} \Psi_{N_1, N_2, t}\right> 
				\leq C \left<\Psi_{N_1,N_2}, \frac{\mathcal{K}_1}{N} \Psi_{N_1,N_2}\right> + C{t^2}\\
				&= C\iint \dd q_1\dd p_1\ |p_1|^2  m_{N_1, N_2}^{(1,0)}(q_1, p_1)
				+ C\hbar \int \dd q \left| \nabla f\left(q\right) \right|^2  + C{t^2} \leq C + C t^2.
			\end{aligned}
		\end{equation}
		Integrate both sides of inequality \eqref{q_1_bdd} with respect to time $t$, and we obtain:
		\begin{equation}
			\begin{aligned}
				&\iint \dd q_1 \dd p_1\ |q_1| m_{N_1, N_2, t}^{(1,0)}(q_1,p_1) \\
				&\leq (2\pi)^3  \left(1 + 2\left<\Psi_{N_1,N_2},\frac{\mathcal{K}_1}{N} \Psi_{N_1,N_2} \right> + Ct^2 + C \sqrt{\hbar} \right)t + \iint \dd q_1 \dd p_1\ |q_1| m_{N_1, N_2}^{(1,0)}(q_1,p_1)\leq C (1 + t^3).
			\end{aligned}
		\end{equation}
		
		By applying Lemma \ref{lem:prop_kHusimi} \eqref{prop_k_1}, we notice that for fixed $1\leq k\leq N_1$ and $0\leq \ell \leq N_2$, the following holds:
		\begin{align*}
			&\int(\dd q\dd p)^{\otimes k}(\dd \tq\dd \tp)^{\otimes \ell}\sum_{j=1}^k |p_j|^2 m_{N_1, N_2, t}^{(k,\ell)}(q_1,p_1,\dots,q_k,p_k,\tq_1,\tp_1,\dots,\tq_\ell,\tp_\ell)\\
			&= \sum_{j=1}^k \iint  \dd q_j\dd p_j\ |p_j|^2 \int  \dd q_1\dd p_1 \cdots \widehat{\dd q_j}\widehat{\dd p_j} \cdots \dd q_k \dd p_k \int(\dd \tq\dd \tp)^{\otimes \ell}m_{N_1, N_2, t}^{(k,\ell)}(q_1,p_1,\dots,q_k,p_k,\tq_1,\tp_1,\dots,\tq_\ell,\tp_\ell)\\
			&= k \iint  \dd q\dd p\ |p|^2 \int  ( \dd q\dd p)^{\otimes k-1}  \int(\dd \tq\dd \tp)^{\otimes \ell} m_{N_1, N_2, t}^{(k,\ell)}(q,p, q_1,p_1\dots,q_{k-1},p_{k-1},\tq_1,\tp_1,\dots,\tq_\ell,\tp_\ell)\\
			&	= (2\pi)^{3(k-1+\ell)}k \frac{(N_1-1)\cdots(N_1-k+1)}{N^{k-1}} \frac{N_2\cdots(N_2-\ell+1)}{N^{\ell}} \iint \dd q\dd p\ |p|^2 m_{N_1, N_2, t}^{(1,0)} (q,p)\\
			&	\leq C(1+t^2),
		\end{align*}
		where we make use of the symmetry of $m_{N_1, N_2, t}^{(k,\ell)}$ in the second equality. Using a similar strategy, we have:
		\begin{align*}
			&\int(\dd q\dd p)^{\otimes k}(\dd \tq\dd \tp)^{\otimes \ell}\sum_{j=1}^k |q_j| m_{N_1, N_2, t}^{(k,\ell)}(q_1,p_1,\dots,q_k,p_k,\tq_1,\tp_1,\dots,\tq_\ell,\tp_\ell)\\
			&	= (2\pi)^{3(k-1+\ell)}k \frac{(N_1-1)\cdots(N_1-k+1)}{N^{k-1}} \frac{N_2\cdots(N_2-\ell+1)}{N^{\ell}} \iint \dd q\dd p\ |q| m_{N_1, N_2, t}^{(1,0)} (q,p)\\
			&
			\leq C(1+t^3).
		\end{align*}
		
		By reversing the roles of the two species, we also obtain:
		\[
		\int(\dd q\dd p)^{\otimes k}(\dd \tq\dd \tp)^{\otimes \ell}\big(|\tilde{\mathbf{q}}_{\ell}|+|\tilde{\mathbf{p}}_{\ell}|^{2}\big)m_{N_1, N_2, t}^{(k,\ell)}(q_1,p_1,\dots,q_k,p_k,\tq_1,\tp_1,\dots,\tq_\ell,\tp_\ell)\leq C(1+t^{3}).
		\]
		As a result, we have successfully shown what we aimed to prove.
	\end{proof}
	
	\section{Constructing BBGKY Hierarchy}\label{sec:BBGKY}
	
	In this section, we construct the two-species BBGKY Hierarchy for the $(k,\ell)$-particle Husimi measure for multi-species fermions. 
	To accomplish this, we initially observe that $m^{(1,0)}$ can be represented in terms of $m^{(1,1)}$ and $m^{(2,0)}$, incorporating certain error terms. Subsequently, an iterative process will allow us to express each $m^{(k,0)}$ in terms of $m^{(k+1,0)}$ and $m^{(k,1)}$, with the inclusion of error terms. 
	Furthermore, we demonstrate the approximate representation of $m^{(k,\ell)}$ in terms of $m^{(k+1,\ell)}$ and $m^{(k,\ell+1)}$. 
	In the upcoming sections, we provide the smallness of the error terms.
	
	In the following two propositions, we only consider $(1,0)$- and $(k,0)$-particle Husimi measures. 
	Regarding the investigation of the dynamics inherent in the $(0,\ell)$-particle Husimi measure, where $1\leq \ell \leq N_2$, it's noteworthy that we can switch the roles of the two species in 
	these propositions,
	
	taking advantage of their equivalent status. This reciprocal role exchange paves the way for a streamlined derivation of the ensuing dynamics. 
	
	Therefore, in the proposition afterward, we will exclude this particular scenario for general $(k,\ell)$ with $1\leq k\leq N_1$ and $1\leq \ell \leq N_2$.
	
	\begin{Proposition} \label{prop:vla_k=1}
		Suppose $\Psi_{N_1, N_2, t} \in \mathcal{F}_a^{(N_1)} \otimes \mathcal{F}_a^{(N_2)}$ is $(N_1,N_2)$-particle state satisfying the Schr\"odinger equation in \eqref{eq:Schrodinger}. Moreover, if $V(-x)=V(x)$ then we have the following equation for $k=1$ and $\ell = 0$,
		\begin{equation} \label{BBGKY_k=1}
			\begin{aligned}
				&\partial_t m_{N_1,N_2,t}^{(1,0)}(q_1,p_1) + p_1 \cdot \nabla_{q_1} m_{N_1,N_2,t}^{(1,0)}(q_1,p_1)\\
				&=\frac{1}{(2\pi)^3}\nabla_{p_1} \cdot  \iint \dd q_2\dd p_2 \nabla V_{11}(q_1-q_2) m_{N_1,N_2,t}^{(2,0)}(q_1,q_2,p_1,p_2)\\
				&\quad+ \frac{1}{(2\pi)^3} {\nabla_{p_1}\cdot}\iint \dd \tq_1\dd \tp_1 \nabla V_{12}(q_1-\tq_1) m_{N_1,N_2,t}^{(1,1)}(q_1,p_1,\tq_1,\tp_1)\\
				&\quad + \nabla_{q_1}\cdot \mathcal{R}_1 +\nabla_{p_1}\cdot \widetilde{\mathcal{R}}_{11} +\nabla_{p_1}\cdot \widetilde{\mathcal{R}}_{12},
			\end{aligned}
		\end{equation}
		where the remainder terms $\mathcal{R}_1$ and $\widetilde{\mathcal{R}}_1$, are given by
		\begin{equation}\label{bbgky_remainder_1}
			\begin{aligned}
				\mathcal{R}_1 :=  &  \hbar \Im \left< \nabla_{q_1} a (f^\hbar_{q_1,p_1}) \Psi_{N_1, N_2, t}, a (f^\hbar_{q_1,p_1}) \Psi_{N_1, N_2, t} \right>,\\
				\widetilde{\mathcal{R}}_{11}  := &\frac{1}{(2\pi)^3} \cdot \iint \dw\du \iint \dy\dd v \iint \dd q_2 \dd p_2 \int_0^1 \dd s\,\\
				&\hspace{1cm}\nabla V_{11}\big(su+(1-s)w - y \big) f_{q_1,p_1}^\hbar (w) \overline{f_{q_1,p_1}^\hbar (u)} f_{q_2,p_2}^\hbar (y) \overline{f_{q_2,p_2}^\hbar (v)}  \left< a_y a_w \Psi_{N_1, N_2, t}, a_v a_u \Psi_{N_1, N_2, t} \right>\\
				& -\frac{1}{(2\pi)^3}\nabla_{p_1} \cdot  \iint \dd q_2\dd p_2 \nabla V_{11}(q_1-q_2) m_{N_1,N_2,t}^{(2,0)}(q_1,q_2,p_1,p_2)\\
				\widetilde{\mathcal{R}}_{12}  :=&\frac{1}{(2\pi)^3} \cdot  \iint \dw\du \iint \dd \ty\dd \tv \iint \dd \tq_1 \dd \tp_1 \int_0^1 \dd s\,\\
				&\hspace{1cm}\nabla V_{12}\big(su+(1-s)w - \ty \big) f_{q_1,p_1}^\hbar (w) \overline{f_{q_1,p_1}^\hbar (u)} f_{\tq_1,\tp_1}^\hbar (\ty) \overline{f_{\tq_1,\tp_1}^\hbar (\tv)}  \left< a_w b_{\ty} \Psi_{N_1, N_2, t}, a_u b_{\tv} \Psi_{N_1, N_2, t} \right>\\
				&- \frac{1}{(2\pi)^3}\iint \dd \tq_1\dd \tp_1 \nabla V_{12}(q_1-\tq_1) m_{N_1,N_2,t}^{(1,1)}(q_1,p_1,\tq_1,\tp_1).
			\end{aligned}
		\end{equation}
	\end{Proposition}
	
	\begin{Proposition}\label{prop:vla_k>1}
		For every $1 \leq i,j \leq k$ and $q_j,p_j \in \bR^3$, denote $\mathbf{q}_k = (q_1,\dots,q_k)$ and  $\mathbf{p}_k = (p_1,\dots,p_k)$. Under the assumption in Proposition \ref{prop:vla_k=1}, we have the following equation for $1<k\leq N_1$,
		\begin{align}
			&\partial_t m_{N_1,N_2,t}^{(k,0)}(q_1, p_1, \dots, q_k, p_k) + \mathbf{p}_k \cdot \nabla_{\mathbf{q}_k} m_{N_1,N_2,t}^{(k,0)}(q_1, p_1, \dots, q_k, p_k)\nonumber\\
			&= \frac{1}{(2\pi)^3}\nabla_{\mathbf{p}_k}\cdot\iint\dd q_{k+1}\dd p_{k+1}\,\nabla V_{11}(q_{j}-q_{k+1})m_{N_{1},N_{2},t}^{(k+1,0)}(q_1, p_1, \dots, q_{k+1}, p_{k+1})\nonumber\\
			&\quad+ \frac{1}{(2\pi)^{3}}\nabla_{\mathbf{p}_k}\cdot \iint\dd \widetilde{q}_1 \dd \widetilde{p}_1\,\nabla V_{12}(q_{j}-\tq_1)m_{N_{1},N_{2},t}^{(k,1)}(q_1, p_1, \dots, q_k, p_k,\tq_1,\tp_1)\nonumber\\
			&\quad+ \nabla_{\mathbf{q}_k}\cdot \mathcal{R}_{1,k} +\nabla_{\mathbf{p}_k}\cdot \widetilde{\mathcal{R}}_{11,k} +\nabla_{\mathbf{p}_k}\cdot \widetilde{\mathcal{R}}_{12,k} + \widehat{\mathcal{R}}_{11,k}, \label{BBGKY_k>1}
		\end{align}
		where the remainder terms $\mathcal{R}_{1,k}$, $\widetilde{\mathcal{R}}_{11,k}$, $\widetilde{\mathcal{R}}_{12,k}$ and $\widehat{\mathcal{R}}_{11,k}$ are given by
		\begin{align}
			\mathcal{R}_{1,k}:=&\hbar\Im\left<\nabla_{\mathbf{q}_{k}}\big(a(f_{q_{k},p_{k}}^{\hbar})\cdots a(f_{q_{1},p_{1}}^{\hbar})\big)\Psi_{N_{1},N_{2},t},a(f_{q_{k},p_{k}}^{\hbar})\cdots a(f_{q_{1},p_{1}}^{\hbar})\Psi_{N_{1},N_{2},t}\right>,\nonumber\\
			(\widetilde{\mathcal{R}}_{11,k})_{j}:=
			&\frac{1}{(2\pi)^3}\int(\dw\du)^{\otimes k}\int\dy\left[\int_{0}^{1}\dd s\,\nabla V_{11}(su_{j}+(1-s)w_{j}-y)\right] \left(f_{q,p}^{\hbar}(w)\overline{f_{q,p}^{\hbar}(u)}\right)^{\otimes k}\nonumber\\
			&\iint\dd q_{k+1}\dd p_{k+1}\ f_{q_{k+1},p_{k+1}}^{\hbar}(y)\int\dd v\ \overline{f_{{q_{k+1}},{p_{k+1}}}^{\hbar}(v)}\left<a_{w_{k}}\cdots a_{w_{1}}a_{y}\Psi_{N_{1},N_{2},t}a_{u_{k}}\cdots a_{u_{1}}a_{v}\Psi_{N_{1},N_{2},t}\right>\nonumber\\
			&-\frac{1}{(2\pi)^3}\iint\dd q_{k+1}\dd p_{k+1}\,\nabla V_{11}(q_{j}-q_{k+1})m_{N_{1},N_{2},t}^{(k+1,0)}(q_1, p_1, \dots, q_{k+1}, p_{k+1}),\nonumber\\
			(\widetilde{\mathcal{R}}_{12,k})_{j}:=
			&\frac{1}{(2\pi)^{3}}\int(\dw\du)^{\otimes k}\int\dd \ty\left[\int_{0}^{1}\dd s\,\nabla V_{12}(su_{j}+(1-s)w_{j}-\ty)\right] \left(f_{q,p}^{\hbar}(w)\overline{f_{q,p}^{\hbar}(u)}\right)^{\otimes k}\nonumber\\
			&\iint\dd \widetilde{q}_1 \dd \widetilde{p}_1\ f_{\widetilde{q}_1,\widetilde{p}_1}^{\hbar}(\ty)\int\dd \tv\ \overline{f_{\widetilde{q}_1,\widetilde{p}_1}^{\hbar}(\tv)}\left<b_{\ty}a_{w_{k}}\cdots a_{w_{1}}\Psi_{N_{1},N_{2},t},b_{\tv}a_{u_{k}}\cdots a_{u_{1}}\Psi_{N_{1},N_{2},t}\right>\nonumber\\
			&-\frac{1}{(2\pi)^{3}}\iint\dd \widetilde{q}_1 \dd \widetilde{p}_1\,\nabla V_{12}(q_{j}-\tq_1)m_{N_{1},N_{2},t}^{(k,1)}(q_1, p_1, \dots, q_k, p_k,\tq_1,\tp_1),\nonumber\\
			\widehat{\mathcal{R}}_{11,k}:=
			&\frac{\mathrm{i}\hbar^{2}}{2}\int(\dw\du)^{\otimes k}\sum_{j\neq i}^{k}\bigg[V_{11}(u_{j}-u_{i})-V_{11}(w_{j}-w_{i})\bigg]\nonumber\\
			&\left(f_{q,p}^{\hbar}(w)\overline{f_{q,p}^{\hbar}(u)}\right)^{\otimes k}\left<a_{w_{k}}\cdots a_{w_{1}}\Psi_{N_{1},N_{2},t},a_{u_{k}}\cdots a_{u_{1}}\Psi_{N_{1},N_{2},t}\right>.\label{BBGKY_remainder_k>1}
		\end{align}
	\end{Proposition}
	
	\begin{Proposition} \label{prop:vla_bbgky_hierarchy}
		Under the assumption in Proposition \ref{prop:vla_k=1}, then for all $1 \leq k \leq N_1$ and $1\leq\ell\leq N_2$, we have the following hierarchy
		\begin{equation} \label{eq:BBGKY_k}
			\begin{aligned}
				&\partial_t m_{N_1,N_2,t}^{(k,\ell)}(q_1, p_1, \dots, q_k, p_k,\tq_1, \tp_1, \dots, \tq_\ell, \tp_\ell)\\
				&+ \mathbf{p}_k \cdot \nabla_{\mathbf{q}_k}m_{N_1,N_2,t}^{(k,\ell)}(q_1, p_1, \dots, q_k, p_k,\tq_1, \tp_1, \dots, \tq_\ell, \tp_\ell)
				{ + \mathbf{\tp}_{\ell} \cdot \nabla_{\mathbf{\tq}_{\ell}}m_{N_1,N_2,t}^{(k,\ell)}(q_1, p_1, \dots, q_k, p_k,\tq_1, \tp_1, \dots, \tq_\ell, \tp_\ell)}  \\
				&= \frac{1}{(2\pi)^3} \nabla_{\mathbf{p}_k} \cdot \iint \dd q_{k+1}\dd p_{k+1} \nabla V_{11}(q_j - q_{k+1})m_{N_1,N_2,t}^{(k+1,\ell)}
				(q_1, p_1, \dots, q_{k+1}, p_{k+1},\tq_1, \tp_1, \dots, \tq_\ell, \tp_\ell)\\
				&+\frac{1}{(2\pi)^3} \nabla_{\mathbf{p}_k} \cdot \iint \dd \tilde q_{\ell+1}\dd \tilde p_{\ell+1} \nabla V_{12}(q_j - \tilde q_{\ell+1})m_{N_1,N_2,t}^{(k,\ell+1)}
				(q_1, p_1, \dots, q_k, p_k,\tq_1, \tp_1, \dots, \tq_{\ell+1}, \tp_{\ell+1})\\
				&+ \frac{1}{(2\pi)^3} \nabla_{\mathbf{\tp}_\ell} \cdot \iint \dd q_{k+1}\dd p_{k+1} \nabla V_{12}(q_{k+1} - \tq_{j})m_{N_1,N_2,t}^{(k+1,\ell)}
				(q_1, p_1, \dots, q_{k+1}, p_{k+1},\tq_1, \tp_1, \dots, \tq_\ell, \tp_\ell)\\
				&+\frac{1}{(2\pi)^3}  \nabla_{\mathbf{\tp}_\ell} \cdot \iint \dd \tilde q_{\ell+1}\dd \tilde p_{\ell+1}\nabla V_{22}(\tq_j - \tilde q_{\ell+1})m_{N_1,N_2,t}^{(k,\ell+1)}
				(q_1, p_1, \dots, q_k, p_k,\tq_1, \tp_1, \dots, \tq_{\ell+1}, \tp_{\ell+1})\\
				&\hspace{1cm} + \nabla_{\mathbf{q}_k}  \cdot \mathcal{R}_{1,k,\ell}
				+\nabla_{\mathbf{\tq}_\ell} \cdot \mathcal{R}_{2,k,\ell}
				+  \nabla_{\mathbf{p}_k} \cdot \widetilde{\mathcal{R}}_{11, k,\ell} 
				+  \nabla_{\mathbf{p}_k} \cdot \widetilde{\mathcal{R}}_{12, 1, k,\ell} 
				+ \nabla_{\mathbf{\tp}_\ell} \cdot \widetilde{\mathcal{R}}_{12, 2, k,\ell}
				+ \nabla_{\mathbf{\tp}_\ell} \cdot \widetilde{\mathcal{R}}_{22, k,\ell}\\
				&\hspace{1cm} + \widehat{\mathcal{R}}_{11, k,\ell} + \widehat{\mathcal{R}}_{22, k,\ell} + \widehat{\mathcal{R}}_{12, k,\ell},
			\end{aligned}
		\end{equation}
		where the remainder terms are denoted as
		\begingroup
		\allowdisplaybreaks
		\begin{align*}
			\mathcal{R}_{1,k,\ell}:=&\hbar\Im\left<b(f_{\widetilde{q}_{\ell},\widetilde{p}_{\ell}}^{\hbar})\cdots b(f_{\widetilde{q}_{1},\widetilde{p}_{1}}^{\hbar})\nabla_{\mathbf{q}_{k}}\big(a(f_{q_{k},p_{k}}^{\hbar})\cdots a(f_{q_{1},p_{1}}^{\hbar})\big)\Psi_{N_1, N_2, t},\right.\\
			&\hspace{15em}\left.b(f_{\widetilde{q}_{\ell},\widetilde{p}_{\ell}}^{\hbar})\cdots b(f_{\widetilde{q}_{1},\widetilde{p}_{1}}^{\hbar})a(f_{q_{k},p_{k}}^{\hbar})\cdots a(f_{q_{1},p_{1}}^{\hbar})\Psi_{N_1, N_2, t}\right>,\\
			\mathcal{R}_{2,k,\ell}:=&\hbar\Im\left<\nabla_{\mathbf{\tq}_{\ell}}\big(b(f_{\widetilde{q}_{\ell},\widetilde{p}_{\ell}}^{\hbar})\cdots b(f_{\widetilde{q}_{1},\widetilde{p}_{1}}^{\hbar})a(f_{q_{k},p_{k}}^{\hbar})\cdots a(f_{q_{1},p_{1}}^{\hbar})\big)\Psi_{N_1, N_2, t},\right.\\
			&\hspace{15em}\left.b(f_{\widetilde{q}_{\ell},\widetilde{p}_{\ell}}^{\hbar})\cdots b(f_{\widetilde{q}_{1},\widetilde{p}_{1}}^{\hbar})a(f_{q_{k},p_{k}}^{\hbar})\cdots a(f_{q_{1},p_{1}}^{\hbar})\Psi_{N_1, N_2, t}\right>,\\
			(\widetilde{\mathcal{R}}
			_{11,k,\ell})_{j}:= & \frac{1}{(2\pi)^3}\int(\dw\du)^{\otimes k}(\dtw\dtu)^{\otimes\ell}\int\dy\left[\int_{0}^{1}\dd s\,\nabla V_{11}(su_{j}+(1-s)w_{j}-y)\right]\\
			& \left(f_{q,p}^{\hbar}(w)\overline{f_{q,p}^{\hbar}(u)}\right)^{\otimes k}\left(f_{\tq,\tp}^{\hbar}(\tw)\overline{f_{\tq,\tp}^{\hbar}(\tu)}\right)^{\otimes\ell}
			\iint\dq_{k+1} \dd  p_{k+1} \ f_{q_{k+1},p_{k+1}}^{\hbar}(y)\int\dd v\ \overline{f_{{q_{k+1}},{p_{k+1}}}^{\hbar}(v)}\\
			&\left<b_{\tw_{\ell}}\dots b_{\tw_{1}}a_{w_{k}}\cdots a_{w_{1}}a_{y}\Psi_{N_1, N_2, t}, b_{\tu_{\ell}}\dots b_{\tu_{1}}a_{u_{k}}\cdots a_{u_{1}}a_{v}\Psi_{N_1, N_2, t}\right>\\
			& -\frac{1}{(2\pi)^3}\iint\dd q_{k+1}\dd p_{k+1}\,\nabla V_{11}(q_{j}-q_{k+1})m_{N_{1},N_{2},t}^{(k+1,\ell)}(q_1, p_1, \dots, q_{k+1}, p_{k+1},\tq_1, \tp_1, \dots, \tq_\ell, \tp_\ell),\\
			(\widetilde{\mathcal{R}}_{12,1,k,\ell})_{j}:= & \frac{1}{(2\pi)^3}  \int (\dw\du)^{\otimes k}(\dtw\dtu)^{\otimes \ell}  \int \dty\,  \left[\int_{0}^{1}\dd s\,\nabla V_{12}(su_{j}+(1-s)w_{j}-\ty)\right]\\
			&\left( f^\hbar_{q,p}(w) \overline{f^\hbar_{q,p}(u)} \right)^{\otimes k}
			\left( f^\hbar_{\tq,\tp}(\tw) \overline{f^\hbar_{\tq,\tp}(\tu)} \right)^{\otimes \ell} \iint\dd \tq_{\ell+1}\dd \tp_{\ell+1}\ f_{\tq_{\ell+1},\tp_{\ell+1}}^{\hbar}(\ty)\int\dd \tv\ \overline{f_{\tq_{\ell+1},\tp_{\ell+1}}^{\hbar}(\tv)}\\
			&\langle b_{\tw_\ell}\cdots b_{\tw_1}  b_{\ty} a_{w_k}\cdots  a_{w_1}\Psi_{N_1, N_2, t}, b_{\tu_\ell}\cdots b_{\tu_1} b_{\tv} a_{u_k}\cdots  a_{u_1} \Psi_{N_1, N_2, t}\rangle\\
			& -\frac{1}{(2\pi)^{3}} \iint\dd \tq_{\ell+1}\dd \tp_{\ell+1}\nabla V_{12}(q_{j}-\tq_{\ell+1})m_{N_{1},N_{2},t}^{(k,\ell+1)}(q_1, p_1, \dots, q_k, p_k,\tq_1, \tp_1, \dots, \tq_{\ell+1}, \tp_{\ell+1}),\\
			(\widetilde{\mathcal{R}}_{12,2,k,\ell})_{j}:= & \frac{1}{(2\pi)^3}   \int (\dw\du)^{\otimes k}(\dtw\dtu)^{\otimes \ell}  \int \dx\left[\int_{0}^{1}\dd s\,\nabla V_{12}(s\tu_{j}+(1-s)\tw_{j}-x)\right]\\
			&\left( f^\hbar_{q,p}(w) \overline{f^\hbar_{q,p}(u)} \right)^{\otimes k}
			\left( f^\hbar_{\tq,\tp}(\tw) \overline{f^\hbar_{\tq,\tp}(\tu)} \right)^{\otimes \ell} \iint\dd q_{k+1}\dd p_{k+1}\ f_{q_{k+1},p_{k+1}}^{\hbar}(x)\int\dd v\ \overline{f_{q_{k+1},p_{k+1}}^{\hbar}(v)}\\
			&\langle b_{\tw_\ell}\cdots  b_{\tw_1}   a_{w_k}\cdots a_{w_1} a_x \Psi_{N_1, N_2, t}, b_{\tu_\ell}\cdots  b_{\tu_1}   a_{u_k}\cdots a_{u_1} a_v \Psi_{N_1, N_2, t}\rangle\\
			& -\frac{1}{(2\pi)^{3}}\iint\dd q_{k+1}\dd p_{k+1}\,\nabla V_{12}(q_{k+1}-\tq_{j})m_{N_{1},N_{2},t}^{(k+1,\ell)}(q_1, p_1, \dots, q_{k+1}, p_{k+1},\tq_1, \tp_1, \dots, \tq_\ell, \tp_\ell),\\
			(\widetilde{\mathcal{R}}_{22,k,\ell})_j:=&\frac{1}{(2\pi)^{3}}  \int(\dw\du)^{\otimes k}(\dtw\dtu)^{\otimes\ell}\int\dty\,\left[\int_{0}^{1}\dd s\,\nabla V_{22}(s\tu_{j}+(1-s)\tw_{j}-\ty)\right]\\
			& \left(f_{q,p}^{\hbar}(w)\overline{f_{q,p}^{\hbar}(u)}\right)^{\otimes k}  \left(f_{\tq,\tp}^{\hbar}(\tw)\overline{f_{\tq,\tp}^{\hbar}(\tu)}\right)^{\otimes\ell}\iint\dd \tq_{\ell+1}\dd \tp_{\ell+1}\ f_{\tq_{\ell+1},\tp_{\ell+1}}^{\hbar}(\ty)\int\dd \tv\ \overline{f_{\tq_{\ell+1},\tp_{\ell+1}}^{\hbar}(\tv)}\\
			& \left<b_{\tw_{\ell}}\dots b_{\tw_{1}}b_{\ty}a_{w_{k}}\cdots a_{w_{1}}\Psi_{N_1, N_2, t},b_{\tu_{\ell}}\dots b_{\tu_{1}}b_{\tv}a_{u_{k}}\cdots a_{u_{1}}\Psi_{N_1, N_2, t}\right>\\
			& -\frac{1}{(2\pi)^{3}} \iint\dd \tq_{\ell+1}\dd \tp_{\ell+1}\nabla V_{22}(\tq_{j}-\tq_{\ell+1})m_{N_{1},N_{2},t}^{(k,\ell+1)}(q_1, p_1, \dots, q_k, p_k,\tq_1, \tp_1, \dots, \tq_{\ell+1}, \tp_{\ell+1}),\\
			\widehat{\mathcal{R}}_{11,k,\ell}:= & \frac{\mathrm{i}\hbar^{2}}{2}\int(\dw\du)^{\otimes k}(\dtw\dtu)^{\otimes\ell}\sum_{j\neq i}^{k}\bigg[V_{11}(u_{j}-u_{i})-V_{11}(w_{j}-w_{i})\bigg]
			\left(f_{q,p}^{\hbar}(w)\overline{f_{q,p}^{\hbar}(u)}\right)^{\otimes k}\\
			& \left(f_{\tq,\tp}^{\hbar}(\tw)\overline{f_{\tq,\tp}^{\hbar}(\tu)}\right)^{\otimes\ell}
			\left<b_{\tw_{\ell}}\dots b_{\tw_{1}}a_{w_{k}}\cdots a_{w_{1}}\Psi_{N_1, N_2, t},b_{\tu_{\ell}}\dots b_{\tu_{1}}a_{u_{k}}\cdots a_{u_{1}}\Psi_{N_1, N_2, t}\right>,\\
			\widehat{\mathcal{R}}_{12,k,\ell}:= & \mathrm{i}\hbar^{2}\int(\dw\du)^{\otimes k}(\dtw\dtu)^{\otimes\ell}\sum_{j=1}^{k}\sum_{i=1}^{\ell}\bigg[V_{12}(u_{j}-\tu_{i})-V_{12}(w_{j}-\tw_{i})\bigg]\left(f_{q,p}^{\hbar}(w)\overline{f_{q,p}^{\hbar}(u)}\right)^{\otimes k}\\
			& \left(f_{\tq,\tp}^{\hbar}(\tw)\overline{f_{\tq,\tp}^{\hbar}(\tu)}\right)^{\otimes\ell}
			\left<b_{\tw_{\ell}}\dots b_{\tw_{1}}a_{w_{k}}\cdots a_{w_{1}}\Psi_{N_1, N_2, t},b_{\tu_{\ell}}\dots b_{\tu_{1}}a_{u_{k}}\cdots a_{u_{1}}\Psi_{N_1, N_2, t}\right>,\\
			\widehat{\mathcal{R}}_{22,k,\ell} := &  \frac{\mathrm{i}\hbar^2}{2}\int(\dw\du)^{\otimes k}(\dtw\dtu)^{\otimes\ell}\sum_{j\neq i}^{\ell}\bigg[V_{22}(\tu_{j}-\tu_{i})-V_{22}(\tw_{j}-\tw_{i})\bigg]
			\left(f_{q,p}^{\hbar}(w)\overline{f_{q,p}^{\hbar}(u)}\right)^{\otimes k}\\
			& \left(f_{\tq,\tp}^{\hbar}(\tw)\overline{f_{\tq,\tp}^{\hbar}(\tu)}\right)^{\otimes\ell} \left<b_{\tw_{\ell}}\dots b_{\tw_{1}}a_{w_{k}}\cdots a_{w_{1}}\Psi_{N_1, N_2, t},b_{\tu_{\ell}}\dots b_{\tu_{1}}a_{u_{k}}\cdots a_{u_{1}}\Psi_{N_1, N_2, t}\right>.
		\end{align*}
		\endgroup
	\end{Proposition}
	
	\begin{Remark}
		Although it is not explicitly stated in the proposition, when $k = 1$, the sum $\sum_{j\neq i}^{k}$ is an empty sum, and by natural convention, the sum is 0. In other words, when $k = 1$, $\widehat{\mathcal{R}}_{11,k,\ell} = 0$. Similarly, when $\ell = 1$, $\widehat{\mathcal{R}}_{22,k,\ell} = 0$.
	\end{Remark}
	
	\begin{Remark}
		Since the proofs for Propositions \ref{prop:vla_k=1}, \ref{prop:vla_k>1}, and \ref{prop:vla_bbgky_hierarchy} are quite similar, we focus on detailing the proof for the more complex one, Proposition \ref{prop:vla_bbgky_hierarchy}, within the main text to save space. 
	\end{Remark}
	
	\begin{proof}[Proof of Proposition \ref{prop:vla_k=1}]
		First, let's consider taking the time derivative on the Husimi measure and utilizing the Schrödinger equation \eqref{eq:Schrodinger}. We have
		\begin{equation} \label{derive_husimi_10}
			\begin{aligned}
				&2\mathrm{i}\hbar \,\p_tm^{(1,0)}_{N_1,N_2,t}(q_1,p_1)\\
				&=\bigg(\hbar^2 \int \dw\du\dx \, f_{q_1,p_1}^\hbar (w) \overline{f_{q_1,p_1}^\hbar (u)} \left<\Psi_{N_1, N_2, t}, a^*_w a_u \nabla_x a_x^* \nabla_x a_x \Psi_{N_1, N_2, t} \right> \\
				&\qquad\qquad- \hbar^2  \int \dw\du\dx \, \overline{f_{q_1,p_1}^\hbar (w)}f_{q_1,p_1}^\hbar (u)\left< \Psi_{N_1, N_2, t}, \nabla_x a_x^* \nabla_x a_x  a^*_u a_w \Psi_{N_1, N_2, t} \right> \bigg)\\
				&\quad+ \bigg(  \frac{1}{N}  \int\dw\du \dx\dy \,  f_{q_1,p_1}^\hbar (w) \overline{f_{q_1,p_1}^\hbar (u)} \left< \Psi_{N_1, N_2, t},  V_{11}(x-y) a^*_w a_u a^*_x a^*_y a_y a_x \Psi_{N_1, N_2, t} \right>\\
				&\qquad\qquad-  \frac{1}{N}   \int\dw\du \dx\dy \,  \overline{f_{q_1,p_1}^\hbar (w)} f_{q_1,p_1}^\hbar (u) \left< \Psi_{N_1, N_2, t},  V_{11}(x-y) a^*_x a^*_y a_y a_x a^*_u a_w \Psi_{N_1, N_2, t} \right> \bigg)\\
				&\quad+ \bigg(  \frac{2}{N}  \int\dw\du \dx \dty\,\, f_{q_1,p_1}^\hbar (w) \overline{f_{q_1,p_1}^\hbar (u)} \left< \Psi_{N_1, N_2, t},  V_{12}(x-\ty) a^*_w a_u a^*_x b^*_{\ty}\, b_{\ty} a_x \Psi_{N_1, N_2, t} \right>\\
				&\qquad\qquad-  \frac{2}{N}   \int\dw\du \dx \dty\,\, \overline{f_{q_1,p_1}^\hbar (w)} f_{q_1,p_1}^\hbar (u) \left< \Psi_{N_1, N_2, t},  V_{12}(x-\ty) a^*_x b^*_{\ty}\, b_{\ty} a_x a^*_u a_w \Psi_{N_1, N_2, t} \right> \bigg)\\
				&\quad+\bigg(\hbar^2 \int \dw\du\dtx\, f_{q_1,p_1}^\hbar (w) \overline{f_{q_1,p_1}^\hbar (u)} \left<\Psi_{N_1, N_2, t}, a^*_w a_u \nabla_{\tx} b_{\tx}^* \nabla_{\tx} b_{\tx} \Psi_{N_1, N_2, t} \right> \\
				&\qquad\qquad- \hbar^2  \int \dw\du\dx \, \overline{f_{q_1,p_1}^\hbar (w)}f_{q_1,p_1}^\hbar (u)\left< \Psi_{N_1, N_2, t}, \nabla_{\tx} b_{\tx}^* \nabla_{\tx} b_{\tx}  a^*_u a_w \Psi_{N_1, N_2, t} \right> \bigg)\\
				&\quad+ \bigg(  \frac{1}{N}  \int\dw\du \dtx\dty\,\,  f_{q_1,p_1}^\hbar (w) \overline{f_{q_1,p_1}^\hbar (u)} \left< \Psi_{N_1, N_2, t},  V_{22}(\tx-\ty) a^*_w a_u b^*_{\tx} b^*_{\ty}\, b_{\ty} b_{\tx} \Psi_{N_1, N_2, t} \right>\\
				&\qquad\qquad-  \frac{1}{N}   \int\dw\du \dtx\dty\,\,  \overline{f_{q_1,p_1}^\hbar (w)} f_{q_1,p_1}^\hbar (u) \left< \Psi_{N_1, N_2, t},  V_{22}(\tx-\ty) b^*_{\tx} b^*_{\ty}\, b_{\ty} b_{\tx} a^*_u a_w \Psi_{N_1, N_2, t} \right> \bigg)\\
				&=:\mathscr{K}_{1}+\mathscr{V}_{11}+\mathscr{V}_{12}+\mathscr{K}_{2}+\mathscr{V}_{22}.
			\end{aligned}
		\end{equation}
		Since only the first species is actually involved in $\mathscr{K}_{1}$ and $\mathscr{V}_{11}$, they are similar to the one-species case \cite{Chen2021}. So we have
		\begin{equation} \label{BBGKY_k1}
			\begin{aligned}
				&\frac{1}{2\mathrm{i}\hbar}(\mathscr{K}_{1}+\mathscr{V}_{11}) + p_1 \cdot \nabla_{q_1} m_{N_1,N_2,t}^{(1,0)}(q_1,p_1)\\
				&= \frac{1}{(2\pi)^3}\nabla_{p_1} \cdot  \iint \dd q_2\dd p_2 \nabla V_{11}(q_1-q_2) m_{N_1,N_2,t}^{(2,0)}(q_1,q_2,p_1,p_2) + \nabla_{q_1}\cdot \mathcal{R}_1 +\nabla_{p_1}\cdot \widetilde{\mathcal{R}}_{11}.
			\end{aligned}
		\end{equation}
		
		For part $\mathscr{K}_{2}$, we interchange $w$ and $u$ in the second line. Using the commutation relation \eqref{eq:CR}, then the second line is just canceled by the first line. So we get $\mathscr{K}_{2}=0$. For the same reason, we know that $\mathscr{V}_{22}=0$.
		
		For part $\mathscr{V}_{12}$, we also interchange $w$ and $u$ in the second line. Then we observe that
		\begin{align*}
			a^*_w a_u a^*_x b^*_{\ty}\, b_{\ty} a_x-a^*_x b^*_{\ty}\, b_{\ty} a_x a^*_w a_u
			= (\delta_{u=x} a^*_w a_x -  \delta_{w=x} a^*_x a_u) b^*_{\ty}\, b_{\ty}. 
		\end{align*}
		Therefore, we have
		\begin{equation}\label{eq:Fock_Potential_V_12}
			\begin{aligned}
				\mathscr{{V}}_{12} &= \frac{2}{N}\int \dw\du\dty\,\, f_{q_1,p_1}^\hbar (w) \overline{f_{q_1,p_1}^\hbar (u)}\left< \Psi_{N_1, N_2, t},  V_{12}(u-\ty)  a^*_w a_u  b^*_{\ty}\, b_{\ty}\, \Psi_{N_1, N_2, t} \right>\\
				&\quad-   \frac{2}{N}\int \dw\du\dty\,\, f_{q_1,p_1}^\hbar (w) \overline{f_{q_1,p_1}^\hbar (u)} \left< \Psi_{N_1, N_2, t},  V_{12}(w-\ty)  a^*_w a_u  b^*_{\ty}\, b_{\ty}\, \Psi_{N_1, N_2, t} \right>\\
				&=  \frac{2}{N}\int \dw\du\dty\,\, f_{q_1,p_1}^\hbar (w) \overline{f_{q_1,p_1}^\hbar (u)}\bigg( V_{12}(u-\ty)- V_{12}(w-\ty) \bigg)\left< \Psi_{N_1, N_2, t},    a^*_w a_u  b^*_{\ty}\, b_{\ty}\, \Psi_{N_1, N_2, t} \right>.
			\end{aligned}
		\end{equation}
		Then we use the mean value theorem
		\begin{equation}\label{eq:MeanValueTheorem_V}
			V_{12}(u-\ty)- V_{12}(w-\ty) = \int_0^1 \ds\, \nabla V_{12}\big(s(u-\ty)+(1-s)(w-\ty) \big)\cdot (u-w).
		\end{equation}
		Note that due to the property $V_{12}\big(s(u-\tilde{y})+(1-s)(w-\tilde{y}) \big) = V_{12}\big(su+(1-s)w - \tilde{y} \big)$, we can derive the following from \eqref{eq:Fock_Potential_V_12}:
		
		\begin{equation}\label{eq:Fock_Potential_V_12_2}
			\begin{aligned}
				\mathscr{{V}}_{12}
				&= \frac{2}{N} \int \dw\du\dty\,\, f_{q_1,p_1}^\hbar (w) \overline{f_{q_1,p_1}^\hbar (u)}\\
				&\qquad\qquad\qquad\left( \int_0^1 \ds\, \nabla V_{12}\big(su+(1-s)w - \ty \big) \right)\cdot (u-w) \left< \Psi_{N_1, N_2, t}, a^*_w a_u  b^*_{\ty}\, b_{\ty}\Psi_{N_1, N_2, t} \right> \\
				&= \frac{2\mathrm{i}\hbar}{N} \int \dw\du\dty\, \int_0^1 \ds\, \nabla V_{12}\big(su+(1-s)w - \ty \big)\\
				&\qquad\qquad\qquad \cdot \nabla_{p_1} \left( f_{q_1,p_1}^\hbar (w) \overline{f_{q_1,p_1}^\hbar (u)} \right)\left< \Psi_{N_1, N_2, t}, a^*_w a_u  b^*_{\ty}\, b_{\ty}\Psi_{N_1, N_2, t} \right> \\
				&= \frac{2\mathrm{i}\hbar}{N} \int \dw\du\dty\, \int_0^1 \ds\, \nabla V_{12}\big(su+(1-s)w - \ty \big)\\
				&\qquad\qquad\qquad \cdot \nabla_{p_1} \left(f_{q_1,p_1}^\hbar (w) \overline{f_{q_1,p_1}^\hbar (u)}\right)\left<a_w b_{\ty}\, \Psi_{N_1, N_2, t},  a_u b_{\ty}\,\Psi_{N_1, N_2, t} \right>,
			\end{aligned}
		\end{equation}
		where we use the fact that
		\begin{equation}\label{eq:grad_p_fhbar}
			\nabla_{p_1} \left(f_{q_1,p_1}^\hbar (w) \overline{f_{q_1,p_1}^\hbar (u)}\right) = \frac{ \mathrm{i}}{\hbar} (w-u)\cdot f_{q_1,p_1}^\hbar (w) \overline{f_{q_1,p_1}^\hbar (u)}.
		\end{equation}
		Applying the following identity
		\begin{equation} \label{eq:coherent_projection}
			\frac{1}{(2\pi \hbar)^3} \int \dd \tq_2 \dd \tp_2 \ket{f^\hbar_{\tq_2,\tp_2}}\bra{f^\hbar_{\tq_2,\tp_2}} =  \mathds{1},
		\end{equation}
		onto $b_{\ty}\, \Psi_{N_1, N_2, t}$, we get
		\begin{align*}
			b_{\ty}\, \Psi_{N_1, N_2, t} = \frac{1}{(2\pi \hbar)^3} \int \dd \tq_2 \dd \tp_2\, f^\hbar_{\tq_2,\tp_2}(\ty) \int \ddv\, \overline{f^\hbar_{\tq_2,\tp_2}(v)}\, b_{\tv}\, \Psi_{N_1, N_2, t} .
		\end{align*}
		Putting this back into \eqref{eq:Fock_Potential_V_12_2} and using $\hbar^3 = N^{-1}$, we get the following
		\begin{equation}\label{eq:Fock_Potential_V_12_3}
			\begin{aligned}
				\mathscr{{V}}_{12} = & \frac{2\mathrm{i}\hbar}{(2\pi)^3} \int\dw\du  \dty\,\dd \tv \dd \tq_2 \dd \tp_2 \int_0^1 \ds\, \nabla V_{12}\big(su+(1-s)w - \ty \big) \\
				&\qquad\qquad\cdot \nabla_{p_1} \left( f_{q_1,p_1}^\hbar (w) \overline{f_{q_1,p_1}^\hbar (u)} \right) f_{\tq_2,\tp_2}^\hbar (\ty) \overline{f_{\tq_2,\tp_2}^\hbar (\tv)} \left< a_w b_{\ty}\, \Psi_{N_1, N_2, t},  a_u b_{\tv}\, \Psi_{N_1, N_2, t} \right>.
			\end{aligned}
		\end{equation}
		Therefore, we have
		\begin{equation}
			\begin{aligned}
				\label{eq:Fock_Potential_V_12_4}
				\frac{\mathscr{{V}}_{12}}{2i\hbar}  
				&=\frac{1}{(2\pi)^3} \int\dw\du \dty\,\dd \tv \dd \tq_1 \dd \tp_1 \int_0^1 \ds\, \nabla V_{12}\big(su+(1-s)w - \ty \big) \\
				&\qquad\qquad\cdot \nabla_{p_1} \left( f_{q_1,p_1}^\hbar (w) \overline{f_{q_1,p_1}^\hbar (u)} \right) f_{\tq_1,\tp_1}^\hbar (\ty) \overline{f_{\tq_1,\tp_1}^\hbar (\tv)} \left< a_w b_{\ty}\, \Psi_{N_1, N_2, t},  a_u b_{\tv}\, \Psi_{N_1, N_2, t} \right>.
			\end{aligned}
		\end{equation}
		Therefore, we divide \eqref{derive_husimi_10} by $2\mathrm{i}\hbar$ and have
		\begin{equation}
			\begin{aligned}
				&\partial_t m_{N_1,N_2,t}^{(1,0)}(q_1,p_1) + p_1 \cdot \nabla_{q_1} m_{N_1,N_2,t}^{(1,0)}(q_1,p_1)\\
				&=\nabla_{q_1}\cdot  \hbar \Im \left< \nabla_{q_1} a (f^\hbar_{q_1,p_1}) \Psi_{N_1, N_2, t}, a (f^\hbar_{q_1,p_1}) \Psi_{N_1, N_2, t} \right> \\
				&\quad +\frac{1}{(2\pi)^3} \nabla_{p_1}\cdot  \iint \dw\du \iint \dy\dd v \iint \dd q_2 \dd p_2 \int_0^1 \dd s\,\\
				&\hspace{1cm}\nabla V_{11}\big(su+(1-s)w - y \big) f_{q_1,p_1}^\hbar (w) \overline{f_{q_1,p_1}^\hbar (u)} f_{q_2,p_2}^\hbar (y) \overline{f_{q_2,p_2}^\hbar (v)}  \left< a_y a_w \Psi_{N_1, N_2, t}, a_v a_u \Psi_{N_1, N_2, t} \right>\\
				&\quad +\frac{1}{(2\pi)^3} \nabla_{p_1}\cdot  \iint \dw\du \iint \dd \ty\dd \tv \iint \dd \tq_1 \dd \tp_1 \int_0^1 \dd s\,\\
				&\hspace{1cm}\nabla V_{12}\big(su+(1-s)w - \ty \big) f_{q_1,p_1}^\hbar (w) \overline{f_{q_1,p_1}^\hbar (u)} f_{\tq_1,\tp_1}^\hbar (\ty) \overline{f_{\tq_1,\tp_1}^\hbar (\tv)}  \left< a_w b_{\ty} \Psi_{N_1, N_2, t}, a_u b_{\tv} \Psi_{N_1, N_2, t} \right>.
			\end{aligned}
		\end{equation}
		Thus, we obtain the desired equation for $m_{N_1,N_2,t}^{(1,0)}(q_1,p_1)$.
	\end{proof}
	
	\begin{proof}[Proof of Proposition \ref{prop:vla_k>1}]
		Similar to the proof in Proposition \ref{prop:vla_k=1}, when considering the time derivative of the Husimi measure involving only one species of particles, the term $h_2$ within the Hamiltonian $H_N$ exhibits no influence. In accordance with the demonstration found in \cite[Proposition 2.2]{Chen2021}, we obtain
		\begin{equation}\label{eq:derive_husimi_k0}
			2\mathrm{i}\hbar\p_{t}m_{N_{1},N_{2},t}^{(k,0)}(q_1, p_1, \dots, q_k, p_k)=\mathscr{K}_{1,k}+\mathscr{V}_{11,k}+\mathscr{V}_{12,k},  
		\end{equation}
		where
		\begin{align*}
			\mathscr{K}_{1,k}&:=\bigg(-\hbar^{2}\int(\dw\du)^{\otimes k}\left(f_{q,p}^{\hbar}(w)\overline{f_{q,p}^{\hbar}(u)}\right)^{\otimes k}\,\Delta_{x}\langle\Psi_{N_{1},N_{2},t},a_{w_{1}}^{*}\cdots a_{w_{k}}^{*}a_{u_{k}}\cdots a_{u_{1}}a_{x}^{*}a_{x}\Psi_{N_{1},N_{2},t}\rangle\\
			&\qquad+\hbar^{2}\int(\dw\du)^{\otimes k}\left(f_{q,p}^{\hbar}(w)\overline{f_{q,p}^{\hbar}(u)}\right)^{\otimes k}\,\Delta_{x}\langle\Psi_{N_{1},N_{2},t},a_{x}^{*}a_{x}a_{w_{1}}^{*}\cdots a_{w_{k}}^{*}a_{u_{k}}\cdots a_{u_{1}}\Psi_{N_{1},N_{2},t}\rangle\bigg)\\
			\mathscr{V}_{11,k}&:=\bigg(\frac{1}{N}\int(\dw\du)^{\otimes k}\iint\dx\dy\; V_{11}(x-y)\left(f_{q,p}^{\hbar}(w)\overline{f_{q,p}^{\hbar}(u)}\right)^{\otimes k}\\
			&\hspace{20em}\langle\Psi_{N_{1},N_{2},t},a_{w_{1}}^{*}\cdots a_{w_{k}}^{*}a_{u_{k}}\cdots a_{u_{1}}a_{x}^{*}a_{y}^{*}a_{y}a_{x}\Psi_{N_{1},N_{2},t}\rangle\\
			&\qquad\quad-\frac{1}{N}\int(\dw\du)^{\otimes k}\iint\dx\dy\; V_{11}(x-y)\left(f_{q,p}^{\hbar}(w)\overline{f_{q,p}^{\hbar}(u)}\right)^{\otimes k}\\
			&\hspace{20em}\langle\Psi_{N_{1},N_{2},t},a_{x}^{*}a_{y}^{*}a_{y}a_{x}a_{w_{1}}^{*}\cdots a_{w_{k}}^{*}a_{u_{k}}\cdots a_{u_{1}}\Psi_{N_{1},N_{2},t}\rangle\bigg)\\
			\mathscr{V}_{12,k}&:=\bigg(\frac{2}{N}\int(\dw\du)^{\otimes k}\iint\dx\dty\;V_{12}(x-\ty)\left(f_{q,p}^{\hbar}(w)\overline{f_{q,p}^{\hbar}(u)}\right)^{\otimes k}\\
			&\hspace{20em}\langle\Psi_{N_{1},N_{2},t},a_{w_{1}}^{*}\cdots a_{w_{k}}^{*}a_{u_{k}}\cdots a_{u_{1}}a_{x}^{*}b_{\ty}^{*}b_{\ty}a_{x}\Psi_{N_{1},N_{2},t}\rangle\\
			&\qquad\quad-\frac{2}{N}\int(\dw\du)^{\otimes k}\iint\dx\dty\;V_{12}(x-\ty)\left(f_{q,p}^{\hbar}(w)\overline{f_{q,p}^{\hbar}(u)}\right)^{\otimes k}\\
			&\hspace{20em}\langle\Psi_{N_{1},N_{2},t},a_{x}^{*}b_{\ty}^{*}b_{\ty}a_{x}a_{w_{1}}^{*}\cdots a_{w_{k}}^{*}a_{u_{k}}\cdots a_{u_{1}}\Psi_{N_{1},N_{2},t}\rangle\bigg).
		\end{align*}
		Similarly as the proposition before, for $\mathscr{K}_{1,k}$ and $\mathscr{V}_{11,k}$, the analysis is exactly the same as in \cite{Chen2021}. We have
		\begin{equation}\label{K_1k}
			\begin{aligned}
				\mathscr{K}_{1,k} = & - 2 \mathrm{i} \hbar  \mathbf{p}_k \cdot \nabla_{\mathbf{q}_k}m_{N_1,N_2,t}^{(k,0)}(q_1, p_1, \dots, q_k, p_k)\\
				& + 2 \mathrm{i} \hbar^2 \Im\left<\Delta_{\mathbf{q}_{k}}\big(a(f_{q_{k},p_{k}}^{\hbar})\cdots a(f_{q_{1},p_{1}}^{\hbar})\big)\Psi_{N_1, N_2, t}, a(f_{q_{k},p_{k}}^{\hbar})\cdots a(f_{q_{1},p_{1}}^{\hbar})\Psi_{N_1, N_2, t}\right>,
			\end{aligned}
		\end{equation}
		and
		\begin{equation}\label{V_11k}
			\begin{aligned}
				\mathscr{V}_{11,k}  = & \frac{2 \mathrm{i} \hbar}{(2\pi\hbar)^{3}N} \sum_{j=1}^k \int(\dw\du)^{\otimes k}\int\dy\left[\int_{0}^{1}\dd s\,\nabla V_{11}(su_{j}+(1-s)w_{j}-y)\right]\\
				&\cdot \nabla_{p_j} \left(f_{q,p}^{\hbar}(w)\overline{f_{q,p}^{\hbar}(u)}\right)^{\otimes k}\iint\dd q_{k+1}\dd p_{k+1}\ f_{q_{k+1},p_{k+1}}^{\hbar}(y)\int\dd v\ \overline{f_{q_{k+1},p_{k+1}}^{\hbar}(v)}\\
				& \left<a_{w_{k}}\cdots a_{w_{1}}a_{y}\Psi_{N_1, N_2, t}, a_{u_{k}}\cdots a_{u_{1}}a_{v}\Psi_{N_1, N_2, t}\right>\\
				& - \frac{1}{N}\int(\dw\du)^{\otimes k}\sum_{j\neq i}^{k}\bigg[V_{11}(u_{j}-u_{i})-V_{11}(w_{j}-w_{i})\bigg]\\
				& \left(f_{q,p}^{\hbar}(w)\overline{f_{q,p}^{\hbar}(u)}\right)^{\otimes k} \left<a_{w_{k}}\cdots a_{w_{1}}\Psi_{N_1, N_2, t},a_{u_{k}}\cdots a_{u_{1}}\Psi_{N_1, N_2, t}\right>.
			\end{aligned}
		\end{equation}
		The only new term is $\mathscr{V}_{12,k}$. We have
		\begin{equation} \label{anil_magic_K_new}
			\begin{aligned}
				&a^*_{w_1} \cdots a^*_{w_k}  a_{u_k}\cdots a_{u_1} a^*_x b^*_{\ty} b_{\ty} a_x\\
				& =  (-1)^{4k} a^*_x b^*_{\ty} b_{\ty} a_x a^*_{w_1} \cdots a^*_{w_k} a_{u_k}\cdots a_{u_1}\\ 
				&\qquad - a^*_{w_1} \cdots a^*_{w_k} \left( \sum_{j=1}^k (-1)^j \delta_{x=u_j}  a_{u_k}\cdots \widehat{a_{u_j}}\cdots a_{u_1}  \right)  b^*_{\ty} b_{\ty} a_x\\
				&\qquad + a^*_x b^*_{\ty} b_{\ty} \left( \sum_{j=1}^k (-1)^j \delta_{x=w_j} a^*_{w_1} \cdots \widehat{a^*_{w_j}} \cdots a^*_{w_k} \right)   a_{u_k}\cdots a_{u_1},
			\end{aligned}
		\end{equation}
		where the \textit{hat}, $\widehat{\cdot}$, indicates the exclusion of that element. So we get
		\begin{equation}
			\begin{aligned}
				&\iint \dx\dty\; V_{12}(x-\ty)
				(a^*_{w_1} \cdots a^*_{w_k}  a_{u_k}\cdots a_{u_1} a^*_x b^*_{\ty} b_{\ty} a_x \\
				& \hspace{8em} - a^*_x b^*_{\ty} b_{\ty} a_x  a^*_{w_1} \cdots a^*_{w_k} a_{u_k}\cdots a_{u_1})\\
				& = - \sum_{j=1}^k (-1)^j \int \dty\, \bigg[ V_{12}(u_j - \ty) a^*_{w_1} \cdots a^*_{w_k}   a_{u_k}\cdots \widehat{a_{u_j}}\cdots a_{u_1} b^*_{\ty} b_{\ty}  a_{u_j}\\
				&\hspace{8em} - V_{12}(w_j-\ty) a^*_{w_j} b^*_{\ty} b_{\ty} a^*_{w_1} \cdots \widehat{a^*_{w_j}} \cdots a^*_{w_k}   a_{u_k}\cdots a_{u_1} \bigg]\\
				& = \sum_{j=1}^k \int \dty\, (V_{12}(u_j - \ty) - V_{12}(w_j-\ty))   a^*_{w_1} \cdots a^*_{w_k} b^*_{\ty} b_{\ty} a_{u_k}\cdots  a_{u_1}.
			\end{aligned}
		\end{equation}
		In the last equality, we have employed the commutation relation \eqref{eq:CR}. By following the steps from equation \eqref{eq:MeanValueTheorem_V} to equation \eqref{eq:Fock_Potential_V_12_4}, we arrive at
		\begin{equation}\label{V_12k}
			\begin{aligned}
				\mathscr{V}_{12,k} =
				&\frac{2\mathrm{i}\hbar}{(2\pi)^{3}}\int(\dw\du)^{\otimes k}\int\dd \ty\left[\int_{0}^{1}\dd s\,\nabla V_{12}(su_{j}+(1-s)w_{j}-\ty)\right] \left(f_{q,p}^{\hbar}(w)\overline{f_{q,p}^{\hbar}(u)}\right)^{\otimes k}\\
				&\iint\dd \widetilde{q}_1 \dd \widetilde{p}_1\ f_{\widetilde{q}_1,\widetilde{p}_1}^{\hbar}(\ty)\int\dd \tv\ \overline{f_{\widetilde{q}_1,\widetilde{p}_1}^{\hbar}(\tv)}\left<b_{\ty}a_{w_{k}}\cdots a_{w_{1}}\Psi_{N_{1},N_{2},t},b_{\tv}a_{u_{k}}\cdots a_{u_{1}}\Psi_{N_{1},N_{2},t}\right>.
			\end{aligned}
		\end{equation}
		By substituting equations \eqref{K_1k}, \eqref{V_11k}, and \eqref{V_12k} into \eqref{eq:derive_husimi_k0}, we complete the proof.
	\end{proof}
	
	\begin{proof}[Proof of Proposition \ref{prop:vla_bbgky_hierarchy}] For convenience, we rewrite \eqref{husimi_def_3} as
		\begin{equation}\label{husimi_def_4}
			\begin{aligned}
				&m^{(k,\ell)}_{N_1,N_2,t} (q_1, p_1, \dots, q_k, p_k,\tq_1, \tp_1, \dots, \tq_\ell, \tp_\ell)\\
				&=
				\int (\dw\du)^{\otimes k}(\dtw\dtu)^{\otimes \ell} 
				\left( f^\hbar_{q,p}(w) \overline{f^\hbar_{q,p}(u)} \right)^{\otimes k}
				\left( f^\hbar_{\tq,\tp}(\tw)
				\overline{f^\hbar_{\tq,\tp}(\tu)} \right)^{\otimes \ell}\\
				&\qquad\qquad\qquad\times\langle\Psi_{N_1, N_2, t},a^*_{w_1} \cdots a^*_{w_k} b^*_{\tw_1} \cdots b^*_{\tw_\ell} b_{\tu_\ell}\cdots b_{\tu_1} a_{u_k}\cdots a_{u_1}\Psi_{N_1, N_2, t}\rangle.
			\end{aligned}
		\end{equation}
		Then by taking the time derivative on the Husimi measure, we have
		\begin{equation}\label{derive_husimi_kl}
			\begin{aligned}
				&2 \mathrm{i} \hbar \p_t m^{(k,\ell)}_{N_1,N_2,t} (q_1, p_1, \dots, q_k, p_k,\tq_1, \tp_1, \dots, \tq_\ell, \tp_\ell)
				=\mathscr{K}_{1,k,\ell} + \mathscr{V}_{11,k,\ell}  + \mathscr{K}_{2,k,\ell} + \mathscr{V}_{22,k,\ell} + \mathscr{V}_{12,k,\ell}.
			\end{aligned}
		\end{equation}
		where  
		\begin{align*}
			\mathscr{K}_{1,k,\ell}&:=\bigg(- \hbar^2 \int (\dw\du)^{\otimes k}(\dtw\dtu)^{\otimes \ell} 
			{\dx}
			\left( f^\hbar_{q,p}(w) \overline{f^\hbar_{q,p}(u)} \right)^{\otimes k}
			\left( f^\hbar_{\tq,\tp}(\tw)
			\overline{f^\hbar_{\tq,\tp}(\tu)} \right)^{\otimes \ell}\\
			&\qquad\qquad\qquad\times \Delta_x \langle\Psi_{N_1, N_2, t},a^*_{w_1} \cdots a^*_{w_k} b^*_{\tw_1} \cdots b^*_{\tw_\ell} b_{\tu_\ell}\cdots b_{\tu_1} a_{u_k}\cdots a_{u_1} a^*_x a_x \Psi_{N_1, N_2, t}\rangle\\
			&\qquad + \hbar^2 \int (\dw\du)^{\otimes k}(\dtw\dtu)^{\otimes \ell} 
			\left( f^\hbar_{q,p}(w) \overline{f^\hbar_{q,p}(u)} \right)^{\otimes k}
			\left( f^\hbar_{\tq,\tp}(\tw)
			\overline{f^\hbar_{\tq,\tp}(\tu)} \right)^{\otimes \ell}\\
			&\qquad\qquad\qquad\times \Delta_x \langle\Psi_{N_1, N_2, t},a^*_x a_x a^*_{w_1} \cdots a^*_{w_k} b^*_{\tw_1} \cdots b^*_{\tw_\ell} b_{\tu_\ell}\cdots b_{\tu_1} a_{u_k}\cdots a_{u_1} \Psi_{N_1, N_2, t}\rangle \bigg)\\
			\mathscr{V}_{11,k,\ell}&:= \bigg(\frac{1}{N} \int (\dw\du)^{\otimes k}(\dtw\dtu)^{\otimes \ell}  \iint \dx\dy\; V_{11}(x-y)
			\left( f^\hbar_{q,p}(w) \overline{f^\hbar_{q,p}(u)} \right)^{\otimes k}
			\left( f^\hbar_{\tq,\tp}(\tw)
			\overline{f^\hbar_{\tq,\tp}(\tu)} \right)^{\otimes \ell}\\
			&\qquad\qquad\qquad\times \langle\Psi_{N_1, N_2, t},a^*_{w_1} \cdots a^*_{w_k} b^*_{\tw_1} \cdots b^*_{\tw_\ell} b_{\tu_\ell}\cdots b_{\tu_1} a_{u_k}\cdots a_{u_1} a^*_x a^*_y a_y a_x \Psi_{N_1, N_2, t}\rangle\\
			&\qquad - \frac{1}{N} \int (\dw\du)^{\otimes k}(\dtw\dtu)^{\otimes \ell}  \iint \dx\dy\; V_{11}(x-y)
			\left( f^\hbar_{q,p}(w) \overline{f^\hbar_{q,p}(u)} \right)^{\otimes k}
			\left( f^\hbar_{\tq,\tp}(\tw)
			\overline{f^\hbar_{\tq,\tp}(\tu)} \right)^{\otimes \ell}\\
			&\qquad\qquad\qquad\times \langle\Psi_{N_1, N_2, t},a^*_x a^*_y a_y a_x a^*_{w_1} \cdots a^*_{w_k} b^*_{\tw_1} \cdots b^*_{\tw_\ell} b_{\tu_\ell}\cdots b_{\tu_1} a_{u_k}\cdots a_{u_1} \Psi_{N_1, N_2, t}\rangle \bigg)\\
			\mathscr{K}_{2,k,\ell}&:= \bigg(- \hbar^2 \int (\dw\du)^{\otimes k}(\dtw\dtu)^{\otimes \ell} {\mathrm{d}\tilde{x}}
			\left( f^\hbar_{q,p}(w) \overline{f^\hbar_{q,p}(u)} \right)^{\otimes k}
			\left( f^\hbar_{\tq,\tp}(\tw)
			\overline{f^\hbar_{\tq,\tp}(\tu)} \right)^{\otimes \ell}\\
			&\qquad\qquad\qquad\times \Delta_{\tx} \langle\Psi_{N_1, N_2, t},a^*_{w_1} \cdots a^*_{w_k} b^*_{\tw_1} \cdots b^*_{\tw_\ell} b_{\tu_\ell}\cdots b_{\tu_1} a_{u_k}\cdots a_{u_1} b^*_{\tx} b_{\tx} \Psi_{N_1, N_2, t}\rangle\\
			&\qquad + \hbar^2 \int (\dw\du)^{\otimes k}(\dtw\dtu)^{\otimes \ell} 
			\left( f^\hbar_{q,p}(w) \overline{f^\hbar_{q,p}(u)} \right)^{\otimes k}
			\left( f^\hbar_{\tq,\tp}(\tw)
			\overline{f^\hbar_{\tq,\tp}(\tu)} \right)^{\otimes \ell}\\
			&\qquad\qquad\qquad\times \Delta_{\tx} \langle\Psi_{N_1, N_2, t}, b^*_{\tx} b_{\tx} a^*_{w_1} \cdots a^*_{w_k} b^*_{\tw_1} \cdots b^*_{\tw_\ell} b_{\tu_\ell}\cdots b_{\tu_1} a_{u_k}\cdots a_{u_1} \Psi_{N_1, N_2, t}\rangle \bigg)\\
			\mathscr{V}_{22,k,\ell} &:= \bigg(\frac{1}{N} \int (\dw\du)^{\otimes k}(\dtw\dtu)^{\otimes \ell}  \iint \dtx\dty\, V_{22}(\tx-\ty)
			\left( f^\hbar_{q,p}(w) \overline{f^\hbar_{q,p}(u)} \right)^{\otimes k}
			\left( f^\hbar_{\tq,\tp}(\tw)
			\overline{f^\hbar_{\tq,\tp}(\tu)} \right)^{\otimes \ell}\\
			&\qquad\qquad\qquad\times \langle\Psi_{N_1, N_2, t},a^*_{w_1} \cdots a^*_{w_k} b^*_{\tw_1} \cdots b^*_{\tw_\ell} b_{\tu_\ell}\cdots b_{\tu_1} a_{u_k}\cdots a_{u_1} b^*_{\tx} b^*_{\ty} b_{\ty} b_{\tx} \Psi_{N_1, N_2, t}\rangle\\
			&\qquad - \frac{1}{N} \int (\dw\du)^{\otimes k}(\dtw\dtu)^{\otimes \ell}  \iint  \dtx\dty\, V_{22}(\tx-\ty)
			\left( f^\hbar_{q,p}(w) \overline{f^\hbar_{q,p}(u)} \right)^{\otimes k}
			\left( f^\hbar_{\tq,\tp}(\tw)
			\overline{f^\hbar_{\tq,\tp}(\tu)} \right)^{\otimes \ell}\\
			&\qquad\qquad\qquad\times \langle\Psi_{N_1, N_2, t},b^*_{\tx} b^*_{\ty} b_{\ty} b_{\tx} a^*_{w_1} \cdots a^*_{w_k} b^*_{\tw_1} \cdots b^*_{\tw_\ell} b_{\tu_\ell}\cdots b_{\tu_1} a_{u_k}\cdots a_{u_1} \Psi_{N_1, N_2, t}\rangle \bigg)\\
			\mathscr{V}_{12,k,\ell}&:=\bigg(\frac{2}{N} \int (\dw\du)^{\otimes k}(\dtw\dtu)^{\otimes \ell}  \iint \dx\dty\; V_{12}(x-\ty)
			\left( f^\hbar_{q,p}(w) \overline{f^\hbar_{q,p}(u)} \right)^{\otimes k}
			\left( f^\hbar_{\tq,\tp}(\tw)
			\overline{f^\hbar_{\tq,\tp}(\tu)} \right)^{\otimes \ell}\\
			&\qquad\qquad\qquad\times \langle\Psi_{N_1, N_2, t},a^*_{w_1} \cdots a^*_{w_k} b^*_{\tw_1} \cdots b^*_{\tw_\ell} b_{\tu_\ell}\cdots b_{\tu_1} a_{u_k}\cdots a_{u_1} a^*_x b^*_{\ty} b_{\ty} a_x \Psi_{N_1, N_2, t}\rangle\\
			&\qquad - \frac{2}{N} \int (\dw\du)^{\otimes k}(\dtw\dtu)^{\otimes \ell}  \iint  \dx\dty\; V_{12}(x-\ty)
			\left( f^\hbar_{q,p}(w) \overline{f^\hbar_{q,p}(u)} \right)^{\otimes k}
			\left( f^\hbar_{\tq,\tp}(\tw)
			\overline{f^\hbar_{\tq,\tp}(\tu)} \right)^{\otimes \ell}\\
			&\qquad\qquad\qquad\times \langle\Psi_{N_1, N_2, t},a^*_x b^*_{\ty} b_{\ty} a_x  a^*_{w_1} \cdots a^*_{w_k} b^*_{\tw_1} \cdots b^*_{\tw_\ell} b_{\tu_\ell}\cdots b_{\tu_1} a_{u_k}\cdots a_{u_1} \Psi_{N_1, N_2, t}\rangle \bigg).
		\end{align*}
		Since the Hamiltonian $h_1$ acts only on the first species, While the terms $\mathscr{K}_{1,k,\ell}$ and $\mathscr{V}_{11,k,\ell}$ contain $b$ and $b^*$, we can still apply the analysis for a single species to them. By utilizing the commutation relation \eqref{eq:CR} and treating the expression $b_{\tu_\ell}\cdots b_{\tu_1} \Psi_{N_1, N_2, t}$ as a unified entity, we obtain the following result
		\begin{equation}
			\begin{aligned}
				\mathscr{K}_{1,k,\ell} = & \int (\dtw\dtu)^{\otimes \ell}{\dx} \left( f^\hbar_{\tq,\tp}(\tw)
				\overline{f^\hbar_{\tq,\tp}(\tu)} \right)^{\otimes \ell}\\
				&\bigg[- \hbar^2 \int (\dw\du)^{\otimes k}
				\left( f^\hbar_{q,p}(w) \overline{f^\hbar_{q,p}(u)} \right)^{\otimes k}
				\\
				&\qquad\qquad\qquad\times \Delta_x \langle (b_{\tu_\ell}\cdots b_{\tu_1}\Psi_{N_1, N_2, t}),a^*_{w_1} \cdots a^*_{w_k} a_{u_k}\cdots a_{u_1} a^*_x a_x (b_{\tu_\ell}\cdots b_{\tu_1}\Psi_{N_1, N_2, t})\rangle\\
				&\qquad + \hbar^2 \int (\dw\du)^{\otimes k}
				\left( f^\hbar_{q,p}(w) \overline{f^\hbar_{q,p}(u)} \right)^{\otimes k}
				\\
				&\qquad\qquad\qquad\times \Delta_x \langle (b_{\tu_\ell}\cdots b_{\tu_1}\Psi_{N_1, N_2, t}),a^*_x a_x a^*_{w_1} \cdots a^*_{w_k}  a_{u_k}\cdots a_{u_1} (b_{\tu_\ell}\cdots b_{\tu_1}\Psi_{N_1, N_2, t})\rangle\bigg].\\
			\end{aligned}
		\end{equation}
		The part within the square brackets exhibits the same structure as $\mathscr{K}_{1,k}$. Applying equation \eqref{K_1k} to it yields
		\begin{equation}\label{K_1kl}
			\begin{aligned}
				\mathscr{K}_{1,k,\ell} = & - 2 \mathrm{i} \hbar  \mathbf{p}_k \cdot \nabla_{\mathbf{q}_k}m_{N_1,N_2,t}^{(k,\ell)}(q_1, p_1, \dots, q_k, p_k,\tq_1, \tp_1, \dots, \tq_\ell, \tp_\ell)\\
				& + 2 \mathrm{i} \hbar^2 \Im\left<\Delta_{\mathbf{q}_{k}}\big(b(f_{\widetilde{q}_{\ell},\widetilde{p}_{\ell}}^{\hbar})\cdots b(f_{\widetilde{q}_{1},\widetilde{p}_{1}}^{\hbar})a(f_{q_{k},p_{k}}^{\hbar})\cdots a(f_{q_{1},p_{1}}^{\hbar})\big)\Psi_{N_1, N_2, t},\right.\\
				&\hspace{12em}\left. b(f_{\widetilde{q}_{\ell},\widetilde{p}_{\ell}}^{\hbar})\cdots b(f_{\widetilde{q}_{1},\widetilde{p}_{1}}^{\hbar})a(f_{q_{k},p_{k}}^{\hbar})\cdots a(f_{q_{1},p_{1}}^{\hbar})\Psi_{N_1, N_2, t}\right>.
			\end{aligned}
		\end{equation}
		Similarly, we have
		\begin{equation}\label{V_11kl}
			\begin{aligned}
				\mathscr{V}_{11,k,\ell}  = & \frac{2 \mathrm{i} \hbar}{(2\pi\hbar)^{3}N} \sum_{j=1}^k \int(\dw\du)^{\otimes k}(\dtw\dtu)^{\otimes\ell}\int\dy\left[\int_{0}^{1}\dd s\,\nabla V_{11}(su_{j}+(1-s)w_{j}-y)\right]\\
				&\cdot \nabla_{p_j} \left(f_{q,p}^{\hbar}(w)\overline{f_{q,p}^{\hbar}(u)}\right)^{\otimes k}\left(f_{\tq,\tp}^{\hbar}(\tw)\overline{f_{\tq,\tp}^{\hbar}(\tu)}\right)^{\otimes\ell}\iint\dd q_{k+1}\dd p_{k+1}\ f_{q_{k+1},p_{k+1}}^{\hbar}(y)\int\dd v\ \overline{f_{q_{k+1},p_{k+1}}^{\hbar}(v)}\\
				& \left<b_{\tw_{\ell}}\dots b_{\tw_{1}}a_{w_{k}}\cdots a_{w_{1}}a_{y}\Psi_{N_1, N_2, t}, b_{\tu_{\ell}}\dots b_{\tu_{1}}a_{u_{k}}\cdots a_{u_{1}}a_{v}\Psi_{N_1, N_2, t}\right>\\
				& - \frac{1}{N}\int(\dw\du)^{\otimes k}(\dtw\dtu)^{\otimes\ell}\sum_{j\neq i}^{k}\bigg[V_{11}(u_{j}-u_{i})-V_{11}(w_{j}-w_{i})\bigg]\\
				& \left(f_{q,p}^{\hbar}(w)\overline{f_{q,p}^{\hbar}(u)}\right)^{\otimes k}\left(f_{\tq,\tp}^{\hbar}(\tw)\overline{f_{\tq,\tp}^{\hbar}(\tu)}\right)^{\otimes\ell} \left<b_{\tw_{\ell}}\dots b_{\tw_{1}}a_{w_{k}}\cdots a_{w_{1}}\Psi_{N_1, N_2, t},b_{\tu_{\ell}}\dots b_{\tu_{1}}a_{u_{k}}\cdots a_{u_{1}}\Psi_{N_1, N_2, t}\right>.
			\end{aligned}
		\end{equation}
		
		For the same reasons as discussed at the beginning of this section, we can deduce the following results regarding the second species by interchanging the roles of the two species:
		\begin{equation}\label{K_2kl}
			\begin{aligned}
				\mathscr{K}_{2,k,\ell} = & - 2 \mathrm{i} \hbar  \mathbf{\tilde p}_\ell \cdot \nabla_{\mathbf{\tilde q}_\ell}m_{N_1,N_2,t}^{(k,\ell)}(q_1, p_1, \dots, q_k, p_k,\tq_1, \tp_1, \dots, \tq_\ell, \tp_\ell)\\
				& + 2 \mathrm{i} \hbar^2 \Im\left<\Delta_{\mathbf{\tilde q}_\ell}\big(b(f_{\widetilde{q}_{\ell},\widetilde{p}_{\ell}}^{\hbar})\cdots b(f_{\widetilde{q}_{1},\widetilde{p}_{1}}^{\hbar})a(f_{q_{k},p_{k}}^{\hbar})\cdots a(f_{q_{1},p_{1}}^{\hbar})\big)\Psi_{N_1, N_2, t},\right.\\
				&\hspace{12em}\left. b(f_{\widetilde{q}_{\ell},\widetilde{p}_{\ell}}^{\hbar})\cdots b(f_{\widetilde{q}_{1},\widetilde{p}_{1}}^{\hbar})a(f_{q_{k},p_{k}}^{\hbar})\cdots a(f_{q_{1},p_{1}}^{\hbar})\Psi_{N_1, N_2, t}\right>,
			\end{aligned}
		\end{equation}
		and
		\begin{equation}\label{V_22kl}
			\begin{aligned}
				\mathscr{V}_{22,k,\ell} = & \frac{2 \mathrm{i} \hbar}{(2\pi\hbar)^{3}N} \sum_{j=1}^\ell \int(\dw\du)^{\otimes k}(\dtw\dtu)^{\otimes\ell}\int\dty\,\left[\int_{0}^{1}\dd s\,\nabla V_{22}(s\tu_{j}+(1-s)\tw_{j}-\ty)\right]\\
				& \cdot  \left(f_{q,p}^{\hbar}(w)\overline{f_{q,p}^{\hbar}(u)}\right)^{\otimes k} \nabla_{\tp_j}  \left(f_{\tq,\tp}^{\hbar}(\tw)\overline{f_{\tq,\tp}^{\hbar}(\tu)}\right)^{\otimes\ell}\iint\dd \tq_{\ell+1}\dd \tp_{\ell+1}\ f_{\tq_{\ell+1},\tp_{\ell+1}}^{\hbar}(\ty)\int\dd \tv\ \overline{f_{\tq_{\ell+1},\tp_{\ell+1}}^{\hbar}(\tv)}\\
				& \left<b_{\tw_{\ell}}\dots b_{\tw_{1}}b_{\ty}a_{w_{k}}\cdots a_{w_{1}}\Psi_{N_1, N_2, t}, b_{\tu_{\ell}}\dots b_{\tu_{1}}b_{\tv}a_{u_{k}}\cdots a_{u_{1}}\Psi_{N_1, N_2, t}\right>\\
				& - \frac{1}{N}\int(\dw\du)^{\otimes k}(\dtw\dtu)^{\otimes\ell}\sum_{j\neq i}^{\ell}\bigg[V_{22}(\tu_{j}-\tu_{i})-V_{22}(\tw_{j}-\tw_{i})\bigg]\\
				& \left(f_{q,p}^{\hbar}(w)\overline{f_{q,p}^{\hbar}(u)}\right)^{\otimes k}\left(f_{\tq,\tp}^{\hbar}(\tw)\overline{f_{\tq,\tp}^{\hbar}(\tu)}\right)^{\otimes\ell} \left<b_{\tw_{\ell}}\dots b_{\tw_{1}}a_{w_{k}}\cdots a_{w_{1}}\Psi_{N_1, N_2, t},b_{\tu_{\ell}}\dots b_{\tu_{1}}a_{u_{k}}\cdots a_{u_{1}}\Psi_{N_1, N_2, t}\right>.\\
			\end{aligned}
		\end{equation}
		Then we consider part $\mathscr{V}_{12,k,\ell}$ of \eqref{derive_husimi_kl},
		\begin{equation}\label{part_E_2}
			\begin{aligned}
				\mathscr{V}_{12,k,\ell} = & \frac{2}{N} \int (\dw\du)^{\otimes k}(\dtw\dtu)^{\otimes \ell}  \iint \dx\dty\; V_{12}(x-\ty)
				\left( f^\hbar_{q,p}(w) \overline{f^\hbar_{q,p}(u)} \right)^{\otimes k}
				\left( f^\hbar_{\tq,\tp}(\tw)
				\overline{f^\hbar_{\tq,\tp}(\tu)} \right)^{\otimes \ell}\\
				&\qquad\qquad\qquad\times \langle\Psi_{N_1, N_2, t},a^*_{w_1} \cdots a^*_{w_k} b^*_{\tw_1} \cdots b^*_{\tw_\ell} b_{\tu_\ell}\cdots b_{\tu_1} a_{u_k}\cdots a_{u_1} a^*_x b^*_{\ty} b_{\ty} a_x \Psi_{N_1, N_2, t}\rangle\\
				&\qquad - \frac{2}{N} \int (\dw\du)^{\otimes k}(\dtw\dtu)^{\otimes \ell}  \iint  \dx\dty\; V_{12}(x-\ty)
				\left( f^\hbar_{q,p}(w) \overline{f^\hbar_{q,p}(u)} \right)^{\otimes k}
				\left( f^\hbar_{\tq,\tp}(\tw)
				\overline{f^\hbar_{\tq,\tp}(\tu)} \right)^{\otimes \ell}\\
				&\qquad\qquad\qquad\times \langle\Psi_{N_1, N_2, t},a^*_x b^*_{\ty} b_{\ty} a_x  a^*_{w_1} \cdots a^*_{w_k} b^*_{\tw_1} \cdots b^*_{\tw_\ell} b_{\tu_\ell}\cdots b_{\tu_1} a_{u_k}\cdots a_{u_1} \Psi_{N_1, N_2, t}\rangle.
			\end{aligned}
		\end{equation}
		
		Observe that we have
		\begin{align}
			&a^*_{w_1} \cdots a^*_{w_k} b^*_{\tw_1} \cdots b^*_{\tw_\ell} b_{\tu_\ell}\cdots b_{\tu_1} a_{u_k}\cdots a_{u_1} a^*_x b^*_{\ty} b_{\ty} a_x\nonumber\\
			& =  (-1)^{4(k+\ell)} a^*_x b^*_{\ty} b_{\ty} a_x a^*_{w_1} \cdots a^*_{w_k}  b^*_{\tw_1} \cdots b^*_{\tw_\ell} b_{\tu_\ell}\cdots b_{\tu_1} a_{u_k}\cdots a_{u_1}\nonumber\\ 
			&\qquad - a^*_{w_1} \cdots a^*_{w_k} b^*_{\tw_1} \cdots b^*_{\tw_\ell} b_{\tu_\ell}\cdots b_{\tu_1} \left( \sum_{j=1}^k (-1)^j \delta_{x=u_j}  a_{u_k}\cdots \widehat{a_{u_j}}\cdots a_{u_1}  \right)  b^*_{\ty} b_{\ty} a_x\nonumber\\
			&\qquad -  a^*_x  a^*_{w_1} \cdots a^*_{w_k}  b^*_{\tw_1} \cdots b^*_{\tw_\ell}\left( \sum_{j=1}^\ell (-1)^j \delta_{\ty=\tu_j}  b_{\tu_\ell}\cdots \widehat{b_{\tu_j}}\cdots b_{\tu_1}  \right) a_{u_k}\cdots a_{u_1}  b_{\ty} a_x\nonumber\\
			&\qquad +  a^*_x b^*_{\ty} a^*_{w_1} \cdots a^*_{w_k}   \left( \sum_{j=1}^\ell (-1)^j \delta_{\ty=\tw_j} b^*_{\tw_1} \cdots \widehat{b^*_{\tw_j}} \cdots b^*_{\tw_\ell} \right) b_{\tu_\ell}\cdots b_{\tu_1} a_{u_k}\cdots a_{u_1} a_x \nonumber\\
			&\qquad + a^*_x b^*_{\ty} b_{\ty} \left( \sum_{j=1}^k (-1)^j \delta_{x=w_j} a^*_{w_1} \cdots \widehat{a^*_{w_j}} \cdots a^*_{w_k} \right)  b^*_{\tw_1} \cdots b^*_{\tw_\ell} b_{\tu_\ell}\cdots b_{\tu_1} a_{u_k}\cdots a_{u_1}.\label{anil_magic_KL_new}
		\end{align}
		
		From \eqref{part_E_2}, we have that
		\begin{align}
			&\iint \dx\dty\; V_{12}(x-\ty)
			(a^*_{w_1} \cdots a^*_{w_k} b^*_{\tw_1} \cdots b^*_{\tw_\ell} b_{\tu_\ell}\cdots b_{\tu_1} a_{u_k}\cdots a_{u_1} a^*_x b^*_{\ty} b_{\ty} a_x \nonumber\\
			& \hspace{8em} - a^*_x b^*_{\ty} b_{\ty} a_x  a^*_{w_1} \cdots a^*_{w_k} b^*_{\tw_1} \cdots b^*_{\tw_\ell} b_{\tu_\ell}\cdots b_{\tu_1} a_{u_k}\cdots a_{u_1})\nonumber\\
			& = - \sum_{j=1}^k (-1)^j \int \dty\, \bigg[ V_{12}(u_j - \ty) a^*_{w_1} \cdots a^*_{w_k} b^*_{\tw_1} \cdots b^*_{\tw_\ell} b_{\tu_\ell}\cdots b_{\tu_1}  a_{u_k}\cdots \widehat{a_{u_j}}\cdots a_{u_1} b^*_{\ty} b_{\ty}  a_{u_j}\nonumber\\
			&\hspace{8em} - V_{12}(w_j-\ty) a^*_{w_j} b^*_{\ty} b_{\ty} a^*_{w_1} \cdots \widehat{a^*_{w_j}} \cdots a^*_{w_k}  b^*_{\tw_1} \cdots b^*_{\tw_\ell} b_{\tu_\ell}\cdots b_{\tu_1} a_{u_k}\cdots a_{u_1} \bigg]\nonumber\\
			&\quad - \sum_{j=1}^\ell (-1)^j \int \dx \bigg[ V_{12}(x - \tu_j) a^*_x  a^*_{w_1} \cdots a^*_{w_k}  b^*_{\tw_1} \cdots b^*_{\tw_\ell} b_{\tu_\ell}\cdots \widehat{b_{\tu_j}}\cdots b_{\tu_1}   a_{u_k}\cdots a_{u_1} b_{\tu_j} a_x\nonumber\\
			&\hspace{8em} - V_{12}(x-\tw_j) a^*_x b^*_{\tw_j} a^*_{w_1} \cdots a^*_{w_k}    b^*_{\tw_1} \cdots \widehat{b^*_{\tw_j}} \cdots b^*_{\tw_\ell} b_{\tu_\ell}\cdots b_{\tu_1} a_{u_k}\cdots a_{u_1} a_x \bigg]\nonumber\\
			& = \sum_{j=1}^k \int \dty\, \bigg[ V_{12}(u_j - \ty) a^*_{w_1} \cdots a^*_{w_k} b^*_{\tw_1} \cdots b^*_{\tw_\ell} b_{\tu_\ell}\cdots b_{\tu_1}  a_{u_k}\cdots  a_{u_1} b^*_{\ty} b_{\ty}\nonumber\\
			&\hspace{8em} - V_{12}(w_j-\ty)  b^*_{\ty} b_{\ty} a^*_{w_1} \cdots a^*_{w_k}  b^*_{\tw_1} \cdots b^*_{\tw_\ell} b_{\tu_\ell}\cdots b_{\tu_1} a_{u_k}\cdots a_{u_1} \bigg]\nonumber\\
			&\quad + \sum_{j=1}^\ell \int \dx \bigg[ V_{12}(x - \tu_j) a^*_x  a^*_{w_1} \cdots a^*_{w_k}  b^*_{\tw_1} \cdots b^*_{\tw_\ell} b_{\tu_\ell}\cdots  b_{\tu_1}   a_{u_k}\cdots a_{u_1} a_x\nonumber\\
			&\hspace{8em} - V_{12}(x-\tw_j) a^*_x a^*_{w_1} \cdots a^*_{w_k}   b^*_{\tw_1} \cdots  b^*_{\tw_\ell} b_{\tu_\ell}\cdots b_{\tu_1} a_{u_k}\cdots a_{u_1}  a_x\bigg]\nonumber\\
			&=:\mathscr{J}_1 + \mathscr{J}_2.\label{part_E_2_decomp_new}
		\end{align}
		
		We use CAR to obtain
		\begin{align*}
			\mathscr{J}_1 =
			&\sum_{j=1}^k \int \dty\, ( V_{12}(u_j - \ty) -  V_{12}(w_j-\ty))  a^*_{w_1} \cdots a^*_{w_k}  b^*_{\ty} b^*_{\tw_1} \cdots b^*_{\tw_\ell} b_{\tu_\ell}\cdots b_{\tu_1}  b_{\ty} a_{u_k}\cdots  a_{u_1}\\
			& + \sum_{j=1}^k \sum_{i=1}^\ell (-1)^i \int \dty\, V_{12}(u_j - \ty) \delta_{\ty=\tu_i} a^*_{w_1} \cdots a^*_{w_k}  b^*_{\tw_1} \cdots b^*_{\tw_\ell} b_{\tu_\ell}\cdots \widehat{b_{\tu_i}}\cdots b_{\tu_1}  b_{\ty} a_{u_k}\cdots  a_{u_1}\\
			& - \sum_{j=1}^k \sum_{i=1}^\ell (-1)^i \int \dty\, V_{12}(w_j - \ty) \delta_{\ty=\tw_i} a^*_{w_1} \cdots a^*_{w_k}  b^*_{\ty} b^*_{\tw_1} \cdots \widehat{b^*_{\tw_i}} \cdots b^*_{\tw_\ell} b_{\tu_\ell}\cdots b_{\tu_1} a_{u_k}\cdots  a_{u_1}\\
			= &\sum_{j=1}^k \int \dty\, ( V_{12}(u_j - \ty) -  V_{12}(w_j-\ty))  a^*_{w_1} \cdots a^*_{w_k}  b^*_{\ty} b^*_{\tw_1} \cdots b^*_{\tw_\ell} b_{\tu_\ell}\cdots b_{\tu_1}  b_{\ty} a_{u_k}\cdots  a_{u_1}\\
			& - \sum_{j=1}^k\sum_{i=1}^\ell  \int \dty\, (V_{12}(u_j - \tu_i) - V_{12}(w_j - \tw_i)) a^*_{w_1} \cdots a^*_{w_k}  b^*_{\tw_1} \cdots b^*_{\tw_\ell} b_{\tu_\ell}\cdots  b_{\tu_1}  a_{u_k}\cdots  a_{u_1}.
		\end{align*}
		For another part, we have
		\begin{equation*}
			\mathscr{J}_2 
			=  \sum_{j=1}^\ell \int \dx ( V_{12}(x - \tu_j) -  V_{12}(x-\tw_j)) a^*_x  a^*_{w_1} \cdots a^*_{w_k}  b^*_{\tw_1} \cdots b^*_{\tw_\ell} b_{\tu_\ell}\cdots  b_{\tu_1}   a_{u_k}\cdots a_{u_1} a_x.
		\end{equation*}
		Thus we have
		\begin{align*}
			\mathscr{V}_{12,k,\ell} 
			= & \frac{2}{N} \int (\dw\du)^{\otimes k}(\dtw\dtu)^{\otimes \ell}  \int \dty\, \sum_{j=1}^k ( V_{12}(u_j - \ty) -  V_{12}(w_j-\ty))
			\left( f^\hbar_{q,p}(w) \overline{f^\hbar_{q,p}(u)} \right)^{\otimes k}\\
			&\qquad \left( f^\hbar_{\tq,\tp}(\tw) \overline{f^\hbar_{\tq,\tp}(\tu)} \right)^{\otimes \ell}
			\langle b_{\tw_\ell}\cdots b_{\tw_1}  b_{\ty} a_{w_k}\cdots  a_{w_1}\Psi_{N_1, N_2, t}, b_{\tu_\ell}\cdots b_{\tu_1}  b_{\ty} a_{u_k}\cdots  a_{u_1} \Psi_{N_1, N_2, t}\rangle\\
			& + \frac{2}{N} \int (\dw\du)^{\otimes k}(\dtw\dtu)^{\otimes \ell}  \int \dx \sum_{j=1}^\ell ( V_{12}(x - \tu_j) -  V_{12}(x-\tw_j))
			\left( f^\hbar_{q,p}(w) \overline{f^\hbar_{q,p}(u)} \right)^{\otimes k}\\
			&\qquad \left( f^\hbar_{\tq,\tp}(\tw) \overline{f^\hbar_{\tq,\tp}(\tu)} \right)^{\otimes \ell}
			\langle b_{\tw_\ell}\cdots  b_{\tw_1}   a_{w_k}\cdots a_{w_1} a_x \Psi_{N_1, N_2, t}, b_{\tu_\ell}\cdots  b_{\tu_1}   a_{u_k}\cdots a_{u_1} a_x \Psi_{N_1, N_2, t}\rangle\\
			& - \frac{2}{N} \int (\dw\du)^{\otimes k}(\dtw\dtu)^{\otimes \ell}  \int \dty\, \sum_{j=1}^k \sum_{i=1}^\ell  ( V_{12}(u_j - \tu_i) - V_{12}(w_j - \tw_i))
			\left( f^\hbar_{q,p}(w) \overline{f^\hbar_{q,p}(u)} \right)^{\otimes k}\\
			&\qquad \left( f^\hbar_{\tq,\tp}(\tw) \overline{f^\hbar_{\tq,\tp}(\tu)} \right)^{\otimes \ell}
			\langle b_{\tw_\ell}\cdots b_{\tw_1} a_{w_k}\cdots a_{w_1}\Psi_{N_1, N_2, t}, b_{\tu_\ell}\cdots b_{\tu_1} a_{u_k}\cdots  a_{u_1} \Psi_{N_1, N_2, t}\rangle\\
			= & \frac{2 i \hbar}{N} \sum_{j=1}^k  \int (\dw\du)^{\otimes k}(\dtw\dtu)^{\otimes \ell}  \int \dty\,  \int_{0}^{1}\dd s\,\nabla V_{12}(su_{j}+(1-s)w_{j}-\ty) \cdot \nabla_{p_j}
			\left( f^\hbar_{q,p}(w) \overline{f^\hbar_{q,p}(u)} \right)^{\otimes k}\\
			&\qquad \left( f^\hbar_{\tq,\tp}(\tw) \overline{f^\hbar_{\tq,\tp}(\tu)} \right)^{\otimes \ell}
			\langle b_{\tw_\ell}\cdots b_{\tw_1}  b_{\ty} a_{w_k}\cdots  a_{w_1}\Psi_{N_1, N_2, t}, b_{\tu_\ell}\cdots b_{\tu_1}  b_{\ty} a_{u_k}\cdots  a_{u_1} \Psi_{N_1, N_2, t}\rangle\\
			& + \frac{2i \hbar}{N}  \sum_{j=1}^\ell  \int (\dw\du)^{\otimes k}(\dtw\dtu)^{\otimes \ell}  \int \dx\int_{0}^{1}\dd s\,\nabla V_{12}(s\tu_{j}+(1-s)\tw_{j}-x) \cdot 
			\left( f^\hbar_{q,p}(w) \overline{f^\hbar_{q,p}(u)} \right)^{\otimes k}\\
			&\qquad \nabla_{\tp_j} \left( f^\hbar_{\tq,\tp}(\tw) \overline{f^\hbar_{\tq,\tp}(\tu)} \right)^{\otimes \ell}
			\langle b_{\tw_\ell}\cdots  b_{\tw_1}   a_{w_k}\cdots a_{w_1} a_x \Psi_{N_1, N_2, t}, b_{\tu_\ell}\cdots  b_{\tu_1}   a_{u_k}\cdots a_{u_1} a_x \Psi_{N_1, N_2, t}\rangle\\
			& - \frac{2}{N} \sum_{j=1}^k \sum_{i=1}^\ell  \int (\dw\du)^{\otimes k}(\dtw\dtu)^{\otimes \ell}  \int \dty\, ( V_{12}(u_j - \tu_i) - V_{12}(w_j - \tw_i))
			\left( f^\hbar_{q,p}(w) \overline{f^\hbar_{q,p}(u)} \right)^{\otimes k}\\
			&\qquad \left( f^\hbar_{\tq,\tp}(\tw) \overline{f^\hbar_{\tq,\tp}(\tu)} \right)^{\otimes \ell}
			\langle b_{\tw_\ell}\cdots b_{\tw_1} a_{w_k}\cdots a_{w_1}\Psi_{N_1, N_2, t}, b_{\tu_\ell}\cdots b_{\tu_1} a_{u_k}\cdots  a_{u_1} \Psi_{N_1, N_2, t}\rangle.
		\end{align*}
		We apply the projection \eqref{projection_f} and get
		\begin{align}
			\mathscr{V}_{12,k,\ell} 
			= & \frac{2 i \hbar}{N} \sum_{j=1}^k  \int (\dw\du)^{\otimes k}(\dtw\dtu)^{\otimes \ell}  \int \dty\,  \int_{0}^{1}\dd s\,\nabla V_{12}(su_{j}+(1-s)w_{j}-\ty) \cdot \nabla_{p_j}
			\left( f^\hbar_{q,p}(w) \overline{f^\hbar_{q,p}(u)} \right)^{\otimes k}\nonumber\\
			&\qquad \left( f^\hbar_{\tq,\tp}(\tw) \overline{f^\hbar_{\tq,\tp}(\tu)} \right)^{\otimes \ell}
			\langle b_{\tw_\ell}\cdots b_{\tw_1}  b_{\ty} a_{w_k}\cdots  a_{w_1}\Psi_{N_1, N_2, t}, b_{\tu_\ell}\cdots b_{\tu_1} \mathds{1} b_{\ty} a_{u_k}\cdots  a_{u_1} \Psi_{N_1, N_2, t}\rangle\nonumber\\
			& + \frac{2i \hbar}{N}  \sum_{j=1}^\ell  \int (\dw\du)^{\otimes k}(\dtw\dtu)^{\otimes \ell}  \int \dx\int_{0}^{1}\dd s\,\nabla V_{12}(s\tu_{j}+(1-s)\tw_{j}-x) \cdot 
			\left( f^\hbar_{q,p}(w) \overline{f^\hbar_{q,p}(u)} \right)^{\otimes k}\nonumber\\
			&\qquad \nabla_{\tp_j} \left( f^\hbar_{\tq,\tp}(\tw) \overline{f^\hbar_{\tq,\tp}(\tu)} \right)^{\otimes \ell}
			\langle b_{\tw_\ell}\cdots  b_{\tw_1}   a_{w_k}\cdots a_{w_1} a_x \Psi_{N_1, N_2, t}, b_{\tu_\ell}\cdots  b_{\tu_1}   a_{u_k}\cdots a_{u_1} \mathds{1} a_x \Psi_{N_1, N_2, t}\rangle\nonumber\\
			& - \frac{2}{N} \sum_{j=1}^k \sum_{i=1}^\ell  \int (\dw\du)^{\otimes k}(\dtw\dtu)^{\otimes \ell}  \int \dty\, ( V_{12}(u_j - \tu_i) - V_{12}(w_j - \tw_i))
			\left( f^\hbar_{q,p}(w) \overline{f^\hbar_{q,p}(u)} \right)^{\otimes k}\nonumber\\
			&\qquad \left( f^\hbar_{\tq,\tp}(\tw) \overline{f^\hbar_{\tq,\tp}(\tu)} \right)^{\otimes \ell}
			\langle b_{\tw_\ell}\cdots b_{\tw_1} a_{w_k}\cdots a_{w_1}\Psi_{N_1, N_2, t}, b_{\tu_\ell}\cdots b_{\tu_1} a_{u_k}\cdots  a_{u_1} \Psi_{N_1, N_2, t}\rangle\nonumber\\
			= & \frac{2 i \hbar}{(2\pi)^3} \sum_{j=1}^k  \int (\dw\du)^{\otimes k}(\dtw\dtu)^{\otimes \ell}  \int \dty\,  \int_{0}^{1}\dd s\,\nabla V_{12}(su_{j}+(1-s)w_{j}-\ty) \cdot \nabla_{p_j}
			\left( f^\hbar_{q,p}(w) \overline{f^\hbar_{q,p}(u)} \right)^{\otimes k}\nonumber\\
			&\qquad \left( f^\hbar_{\tq,\tp}(\tw) \overline{f^\hbar_{\tq,\tp}(\tu)} \right)^{\otimes \ell} \iint\dd \tq_{\ell+1}\dd \tp_{\ell+1}\ f_{\tq_{\ell+1},\tp_{\ell+1}}^{\hbar}(\ty)\int\dd \tv\ \overline{f_{\tq_{\ell+1},\tp_{\ell+1}}^{\hbar}(\tv)}\nonumber\\
			&\qquad\langle b_{\tw_\ell}\cdots b_{\tw_1}  b_{\ty} a_{w_k}\cdots  a_{w_1}\Psi_{N_1, N_2, t}, b_{\tu_\ell}\cdots b_{\tu_1} b_{\tv} a_{u_k}\cdots  a_{u_1} \Psi_{N_1, N_2, t}\rangle\nonumber\\
			& + \frac{2i \hbar}{(2\pi)^3}  \sum_{j=1}^\ell  \int (\dw\du)^{\otimes k}(\dtw\dtu)^{\otimes \ell}  \int \dx\int_{0}^{1}\dd s\,\nabla V_{12}(s\tu_{j}+(1-s)\tw_{j}-x) \cdot 
			\left( f^\hbar_{q,p}(w) \overline{f^\hbar_{q,p}(u)} \right)^{\otimes k}\nonumber\\
			&\qquad \nabla_{\tp_j} \left( f^\hbar_{\tq,\tp}(\tw) \overline{f^\hbar_{\tq,\tp}(\tu)} \right)^{\otimes \ell} \iint\dd q_{k+1}\dd p_{k+1}\ f_{q_{k+1},p_{k+1}}^{\hbar}(x)\int\dd v\ \overline{f_{q_{k+1},p_{k+1}}^{\hbar}(v)}\nonumber\\
			&\qquad \langle b_{\tw_\ell}\cdots  b_{\tw_1}   a_{w_k}\cdots a_{w_1} a_x \Psi_{N_1, N_2, t}, b_{\tu_\ell}\cdots  b_{\tu_1}   a_{u_k}\cdots a_{u_1} a_v \Psi_{N_1, N_2, t}\rangle\nonumber\\
			& - \frac{2}{N} \sum_{j=1}^k \sum_{i=1}^\ell  \int (\dw\du)^{\otimes k}(\dtw\dtu)^{\otimes \ell}  \int \dty\; ( V_{12}(u_j - \tu_i) - V_{12}(w_j - \tw_i))
			\left( f^\hbar_{q,p}(w) \overline{f^\hbar_{q,p}(u)} \right)^{\otimes k}\nonumber\\
			&\qquad \left( f^\hbar_{\tq,\tp}(\tw) \overline{f^\hbar_{\tq,\tp}(\tu)} \right)^{\otimes \ell}
			\langle b_{\tw_\ell}\cdots b_{\tw_1} a_{w_k}\cdots a_{w_1}\Psi_{N_1, N_2, t}, b_{\tu_\ell}\cdots b_{\tu_1} a_{u_k}\cdots  a_{u_1} \Psi_{N_1, N_2, t}\rangle.\label{V_12kl}
		\end{align}
		The proof is completed by substituting \eqref{K_1kl}, \eqref{V_11kl}, \eqref{K_2kl}, \eqref{V_22kl}, and \eqref{V_12kl} into \eqref{derive_husimi_kl}.
	\end{proof}
	
	\section{Bounds of errors}\label{sec:Bound of errors}
	Throughout this section, we provide bounds on the errors in a weak sense, which is why we use test functions. In each proposition, $\varphi, \phi, \tilde\varphi, \tilde\phi$ will denote test functions. Moreover, $k,\ell$ are non-negative integers.
	\subsection{$\cR_{1,k,\ell}$ and $\cR_{2,k,\ell}$}
	Recall the remainder terms in Proposition \ref{prop:vla_bbgky_hierarchy}, i.e., 
	\begin{equation}
		\begin{aligned}
			\mathcal{R}_{1,k,\ell}:=&\hbar\Im\left<b(f_{\widetilde{q}_{\ell},\widetilde{p}_{\ell}}^{\hbar})\cdots b(f_{\widetilde{q}_{1},\widetilde{p}_{1}}^{\hbar})\nabla_{\mathbf{q}_{k}}\big(a(f_{q_{k},p_{k}}^{\hbar})\cdots a(f_{q_{1},p_{1}}^{\hbar})\big)\Psi_{N_1, N_2, t},\right.\\
			&\hspace{15em}\left.b(f_{\widetilde{q}_{\ell},\widetilde{p}_{\ell}}^{\hbar})\cdots b(f_{\widetilde{q}_{1},\widetilde{p}_{1}}^{\hbar})a(f_{q_{k},p_{k}}^{\hbar})\cdots a(f_{q_{1},p_{1}}^{\hbar})\Psi_{N_1, N_2, t}\right>,\\
			\mathcal{R}_{2,k,\ell}:=&\hbar\Im\left<\nabla_{\mathbf{\tq}_{\ell}}\big(b(f_{\widetilde{q}_{\ell},\widetilde{p}_{\ell}}^{\hbar})\cdots b(f_{\widetilde{q}_{1},\widetilde{p}_{1}}^{\hbar})\big)a(f_{q_{k},p_{k}}^{\hbar})\cdots a(f_{q_{1},p_{1}}^{\hbar})\Psi_{N_1, N_2, t},\right.\\
			&\hspace{15em}\left.b(f_{\widetilde{q}_{\ell},\widetilde{p}_{\ell}}^{\hbar})\cdots b(f_{\widetilde{q}_{1},\widetilde{p}_{1}}^{\hbar})a(f_{q_{k},p_{k}}^{\hbar})\cdots a(f_{q_{1},p_{1}}^{\hbar})\Psi_{N_1, N_2, t}\right>.
		\end{aligned}
	\end{equation}
	\begin{Proposition}\label{prop:R1kell}
		Under Assumption \ref{ass:main}, we have: for $1 \leq k \leq N_1$ and $0\leq\ell\leq N_2$
		\begin{align}\label{eq:R_1kl_bdd}
			\bigg|\int\dq^{\otimes k} \ddp^{\otimes k} \dtq^{\otimes \ell} \dtp^{\otimes \ell} \, \varphi^{\otimes k} \phi^{\otimes k} \tilde\varphi^{\otimes \ell} \tilde\phi^{\otimes \ell} \nabla_{\mathbf{q}_k}\cdot \cR_{1,k,\ell}(q_{1},p_{1},\dots,q_{k},p_{k},\tq_{1},\tp_{1},\dots,\tq_{\ell},\tp_{\ell}) \bigg| \leq {C_t}  \hbar^{\frac{1}{2}},
		\end{align}
		and for $0 \leq k \leq N_1$ and $1\leq\ell\leq N_2$
		\begin{align}\label{eq:R_2kl_bdd}
			\bigg|\int\dq^{\otimes k} \ddp^{\otimes k} \dtq^{\otimes \ell} \dtp^{\otimes \ell} \, \varphi^{\otimes k} \phi^{\otimes k} \tilde\varphi^{\otimes \ell} \tilde\phi^{\otimes \ell} \nabla_{\mathbf{\tq}_{\ell}}\cdot \cR_{2,k,\ell}(q_{1},p_{1},\dots,q_{k},p_{k},\tq_{1},\tp_{1},\dots,\tq_{\ell},\tp_{\ell}) \bigg| \leq {C_t} \hbar^{\frac{1}{2}},
		\end{align}
		where  $C$ depends on $\varphi, \phi, \tilde\varphi, \tilde\phi, f, k, \ell$.
	\end{Proposition}
	\begin{Remark}
		The following proof, as well as the subsequent proofs of Proposition \ref{prop:tildeR11kell}, \ref{prop:tildeR12kell}, \ref{prop:hatR11kell}, and \ref{prop:hatR12kell}, actually only deal with the errors given in Proposition \ref{prop:vla_bbgky_hierarchy}, that is, the most complex situation where $k, \ell \geq 1$. For the control of the errors given in Propositions \ref{prop:vla_k=1} and \ref{prop:vla_k>1}, that is, the case where one of $k$ and $l$ is 0, the calculation has been included in the most complex case. Finally, the proofs of the following propositions will not involve simplified cases.
	\end{Remark}
	
	\begin{proof}[Proof of Proposition \ref{prop:R1kell}]
		For $1\leq j\leq k$, the $j$-th component of $\mathcal{R}_{1,k,\ell}$ is given as
		\begin{align*}
			(\mathcal{R}_{1,k,\ell})_j:=&\hbar\Im\left<\big(b(f_{\widetilde{q}_{\ell},\widetilde{p}_{\ell}}^{\hbar})\cdots b(f_{\widetilde{q}_{1},\widetilde{p}_{1}}^{\hbar})\nabla_{q_j}a(f_{q_{k},p_{k}}^{\hbar})\cdots a(f_{q_{1},p_{1}}^{\hbar})\big)\Psi_{N_1, N_2, t},\right.\\
			&\hspace{15em}\left.b(f_{\widetilde{q}_{\ell},\widetilde{p}_{\ell}}^{\hbar})\cdots b(f_{\widetilde{q}_{1},\widetilde{p}_{1}}^{\hbar})a(f_{q_{k},p_{k}}^{\hbar})\cdots a(f_{q_{1},p_{1}}^{\hbar})\Psi_{N_1, N_2, t}\right>.
		\end{align*}
		So we have
		\begin{align*}
			&\int\ddp_{1}\dots\ddp_{k}\dtp_{1}\dots\dtp_{\ell}\,\big|(\mathcal{R}_{1,k,\ell})_j(q_{1},p_{1},\dots,q_{k},p_{k},\tq_{1},\tp_{1},\dots,\tq_{\ell},\tp_{\ell})\big|\\
			&\leq\hbar\int\ddp_{1}\dots\ddp_{k}\dtp_{1}\dots\dtp_{\ell}\,\|b(f_{\widetilde{q}_{\ell},\widetilde{p}_{\ell}}^{\hbar})\cdots b(f_{\widetilde{q}_{1},\widetilde{p}_{1}}^{\hbar})\nabla_{q_{j}}a(f_{q_{k},p_{k}}^{\hbar})\cdots a(f_{q_{1},p_{1}}^{\hbar})\Psi_{N_{1},N_{2},t}\|\\
			&\qquad\times\|b(f_{\widetilde{q}_{\ell},\widetilde{p}_{\ell}}^{\hbar})\cdots b(f_{\widetilde{q}_{1},\widetilde{p}_{1}}^{\hbar})a(f_{q_{k},p_{k}}^{\hbar})\cdots a(f_{q_{1},p_{1}}^{\hbar})\Psi_{N_{1},N_{2},t}\|\\
			&\leq\bigg[\hbar^{2}\int\ddp_{1}\dots\ddp_{k}\dtp_{1}\dots\dtp_{\ell}\bigg<b(f_{\widetilde{q}_{\ell},\widetilde{p}_{\ell}}^{\hbar})\cdots b(f_{\widetilde{q}_{1},\widetilde{p}_{1}}^{\hbar})\nabla_{q_{j}}a(f_{q_{k},p_{k}}^{\hbar})\cdots a(f_{q_{1},p_{1}}^{\hbar})\Psi_{N_{1},N_{2},t},\\
			&\hspace{15em}b(f_{\widetilde{q}_{\ell},\widetilde{p}_{\ell}}^{\hbar})\cdots b(f_{\widetilde{q}_{1},\widetilde{p}_{1}}^{\hbar})\nabla_{q_{j}}a(f_{q_{k},p_{k}}^{\hbar})\cdots a(f_{q_{1},p_{1}}^{\hbar})\Psi_{N_{1},N_{2},t}\bigg>\bigg]^{\frac{1}{2}}\\
			&\qquad\times\bigg[\int\ddp_{1}\dots\ddp_{k}\dtp_{1}\dots\dtp_{\ell}\,m_{N_{1},N_{2},t}(q_{1},p_{1},\dots,q_{k},p_{k},\tq_{1},\tp_{1},\dots,\tq_{\ell},\tp_{\ell})\bigg]^{\frac{1}{2}}\\
			&=(2\pi)^\frac32\bigg[\hbar^{2}\int\ddp_{1}\dots\ddp_{k}\dtp_{1}\dots\dtp_{\ell}\bigg<b(f_{\widetilde{q}_{\ell},\widetilde{p}_{\ell}}^{\hbar})\cdots b(f_{\widetilde{q}_{1},\widetilde{p}_{1}}^{\hbar})\nabla_{q_{j}}a(f_{q_{k},p_{k}}^{\hbar})\cdots a(f_{q_{1},p_{1}}^{\hbar})\Psi_{N_{1},N_{2},t},\\
			&\hspace{5em} b(f_{\widetilde{q}_{\ell},\widetilde{p}_{\ell}}^{\hbar})\cdots b(f_{\widetilde{q}_{1},\widetilde{p}_{1}}^{\hbar})\nabla_{q_{j}}a(f_{q_{k},p_{k}}^{\hbar})\cdots a(f_{q_{1},p_{1}}^{\hbar})\Psi_{N_{1},N_{2},t}\bigg>\bigg]^{\frac{1}{2}}
			\rho_{N_{1},N_{2},t}^{\frac{1}{2}}(q_{1},\dots,q_{k},\tq_{1},\dots,\tq_{\ell})\\
			&=(2\pi)^\frac32\Bigg[\hbar^{2-\frac{3}{2}(k+\ell)}\int\ddp_{1}\dots\ddp_{k}\dtp_{1}\dots\dtp_{\ell}\int\dw_{1}\dots\dw_{k}\du_{1}\dots\du_{k}\dtw_{1}\dots\dtw_{\ell}\dtu_{1}\dots\dtu_{\ell}\\
			&\qquad\Bigg(\prod_{i\neq j}f\left(\frac{w_{i}-q_{i}}{\sqrt{\hbar}}\right)f\left(\frac{u_{i}-q_{i}}{\sqrt{\hbar}}\right)e^{\frac{\ii}{\hbar}p_{i}\cdot(w_{i}-u_{i})}\Bigg)
			\Bigg(\nabla_{q_{j}}f\left(\frac{w_{j}-q_{j}}{\sqrt{\hbar}}\right)\nabla_{q_{j}}f\left(\frac{u_{j}-q_{j}}{\sqrt{\hbar}}\right)e^{\frac{\ii}{\hbar}p_{j}\cdot(w_{j}-u_{j})}\Bigg)\\
			&\qquad\Bigg(\prod_{i=1}^{\ell}f\left(\frac{\tw_{i}-\tq_{i}}{\sqrt{\hbar}}\right)f\left(\frac{\tu_{i}-\tq_{i}}{\sqrt{\hbar}}\right)e^{\frac{\ii}{\hbar}\tp_{i}\cdot(w_{i}-u_{i})}\Bigg)\\
			&\qquad\times\left<\Psi_{N_{1},N_{2},t},a_{w_{1}}^{*}\dots a_{w_{k}}^{*}b_{\tw_{1}}^{*}\dots b_{\tw_{\ell}}^{*}b_{\tu_{1}}\dots b_{\tu_{\ell}}a_{u_{1}}\dots a_{u_{k}}\Psi_{N_{1},N_{2},t}\right>\Bigg]^{\frac{1}{2}}\rho_{N_{1},N_{2},t}^{\frac{1}{2}}(q_{1},\dots,q_{k},\tq_{1},\dots,\tq_{\ell})\\
			&\leq(2\pi)^{3(k+\ell+\frac12)}\Bigg[\hbar^{2+\frac{3}{2}(k+\ell)}\int\dw_{1}\dots\dw_{k}\dtw_{1}\dots\dtw_{\ell}\ \hbar^{-1}\left|\nabla f\left(\frac{w_{j}-q_{j}}{\sqrt{\hbar}}\right)\right|^{2}\prod_{i\neq j}\left|f\left(\frac{w_{i}-q_{i}}{\sqrt{\hbar}}\right)\right|^{2}\\
			&\quad\prod_{i=1}^{\ell}\left|f\left(\frac{\tw_{i}-\tq_{i}}{\sqrt{\hbar}}\right)\right|^{2}\left<\Psi_{N_{1},N_{2},t},a_{w_{1}}^{*}\dots a_{w_{k}}^{*}b_{\tw_{1}}^{*}\dots b_{\tw_{\ell}}^{*}b_{\tu_{1}}\dots b_{\tu_{\ell}}a_{u_{1}}\dots a_{u_{k}}\Psi_{N_{1},N_{2},t}\right>\Bigg]^{\frac{1}{2}}\rho_{N_{1},N_{2},t}^{\frac{1}{2}}(q_{1},\dots,q_{k})\\
			&=(2\pi)^{3(k+\ell+\frac12)}\hbar^{\frac{1}{2}(1+3k+3\ell)}\Bigg[\int\dw_{1}\dots\dw_{k}\dtw_{1}\dots\dtw_{\ell}\ \hbar^{-\frac{3}{2}k}\left|\nabla f\left(\frac{w_{j}-q_{j}}{\sqrt{\hbar}}\right)\right|^{2}\prod_{i\neq j}\left|f\left(\frac{w_{i}-q_{i}}{\sqrt{\hbar}}\right)\right|^{2}\\
			&\qquad h^{-\frac{3}{2}\ell}\prod_{i=1}^{\ell}\left|f\left(\frac{\tw_{i}-\tq_{i}}{\sqrt{\hbar}}\right)\right|^{2}\left<\Psi_{N_{1},N_{2},t},a_{w_{1}}^{*}\dots a_{w_{k}}^{*}b_{\tw_{1}}^{*}\dots b_{\tw_{\ell}}^{*}b_{\tu_{1}}\dots b_{\tu_{\ell}}a_{u_{1}}\dots a_{u_{k}}\Psi_{N_{1},N_{2},t}\right>\Bigg]^{\frac{1}{2}}\\
			&\qquad\rho_{N_{1},N_{2},t}^{\frac{1}{2}}(q_{1},\dots,q_{k},\tq_{1},\dots,\tq_{\ell}).
		\end{align*}
		This implies, by using H\"older inequality, that
		\begin{align*}
			&\int\dq_{1}\dots\dq_{k}\dtq_{1}\dots\dtq_{\ell}\Big|\int\ddp_{1}\dots\ddp_{k}\dtp_{1}\dots\dtp_{k}\,\big|(\mathcal{R}_{1,k,\ell})_j(q_{1},p_{1},\dots,q_{k},p_{k},\tq_{1},\tp_{1},\dots,\tq_{\ell},\tp_{\ell})\big|\Big|\\
			&\leq(2\pi)^{3(k+\ell+\frac12)}\hbar^{\frac{1}{2}(1+3k+3\ell)}\Bigg[\int\dq_{1}\dots\dq_{k}\dtq_{1}\dots\dtq_{\ell}\dw_{1}\dots\dw_{k}\dtw_{1}\dots\dtw_{\ell}\ \Bigg(\prod_{i\neq j}\hbar^{-\frac{3}{2}}\left|f\left(\frac{w_{i}-q_{i}}{\sqrt{\hbar}}\right)\right|^{ 4}\Bigg)\\
			&\qquad\Bigg(\hbar^{-\frac{3}{2}}\left|\nabla f\left(\frac{w_{j}-q_{j}}{\sqrt{\hbar}}\right)\right|^{ 4}\Bigg)\Bigg(\prod_{j=1}^{\ell}\hbar^{-\frac{3}{2}}\left|f\left(\frac{\tw_{j}-\tq_{j}}{\sqrt{\hbar}}\right)\right|^{ 4}\Bigg)
			\left<\Psi_{N_1, N_2, t},a_{w_{1}}^{*}\dots a_{w_{k}}^{*}a_{\tw_{1}}^{*}\dots a_{\tw_{\ell}}^{*}a_{\tu_{1}}\dots a_{\tu_{\ell}}a_{u_{1}}\dots a_{u_{k}}\Psi_{N_1, N_2, t}\right>\Bigg]^{\frac12}\\
			&\qquad\times\left[\int\dq_{1}\dots\dq_{k}\dtq_{1}\dots\dtq_{\ell}\,\rho_{N_1, N_2, t}(q_{1},\dots,q_{k},\tq_{1},\dots,\tq_{\ell})\right]^{{\frac12}}\\
			&\leq {C_t} \hbar^{\frac{1}{2}(1+3k+3\ell)}\|\nabla f\|_{ 4}\|f\|_{4}^{k+\ell-1}\left<\Psi_{N_1, N_2, t},\cN^{k}_1\cN^{\ell}_2\Psi_{N_1, N_2, t}\right>^{\frac{1}{2}}\\
			&\leq {C_t}  \hbar^{\frac{1}{2}(1+3k+3\ell)} N_1^{\frac{k}{2}} N_2^{\frac{\ell}{2}}
			\leq {C_t}  \hbar^{\frac{1}{2}(1+3k+3\ell)} N^{\frac{k+\ell}{2}} 
			= {C_t}  \hbar^{\frac{1}{2}(1+3k+3\ell)-\frac{3}{2}(k+\ell)}
			= {C_t}  \hbar^{\frac{1}{2}}.
		\end{align*}
		In the above estimate, we utilized the bounded nature of $\|\rho_{N_1, N_2, t}\|_{\frac53}$, which can be constrained by the kinetic energy through the fundamental kinetic energy inequality \cite{Lieb2009}. Additionally, we acquire:
		\begin{align*}
			&\bigg|\int\dq^{\otimes k}\ddp^{\otimes k}\dtq^{\otimes\ell}\dtp^{\otimes\ell}\,\varphi^{\otimes k}\phi^{\otimes k}\tilde{\varphi}^{\otimes\ell}\tilde{\phi}^{\otimes\ell}\,\,\nabla_{\mathbf{q}_{k}}\cdot\cR_{1,k,\ell}(q_{1},p_{1},\dots,q_{k},p_{k},\tq_{1},\tp_{1},\dots,\tq_{\ell},\tp_{\ell})\bigg|\\
			&\leq\sum_{j=1}^{k}\bigg|\int\dq^{\otimes k}\ddp^{\otimes k}\dtq^{\otimes\ell}\dtp^{\otimes\ell}\,\varphi^{\otimes k}\phi^{\otimes k}\tilde{\varphi}^{\otimes\ell}\tilde{\phi}^{\otimes\ell}\,\nabla_{q_j}\cdot(\mathcal{R}_{1,k,\ell})_j(q_{1},p_{1},\dots,q_{k},p_{k},\tq_{1},\tp_{1},\dots,\tq_{\ell},\tp_{\ell})\bigg|\\
			&\leq\sum_{j=1}^{k}\|\nabla\varphi\|_{L^{{\infty}}(\bR^{3})}\|\varphi\|_{L^{{\infty}}(\bR^{3})}^{k+\ell-1}\Bigg(\int\dq^{\otimes k}\dtq^{\otimes\ell}\Bigg|\int\ddp^{\otimes k}\dtp^{\otimes\ell}\,\Big(\prod_{i=1}^{k}\phi(p_{i})\Big)\,\Big(\prod_{i=1}^{\ell}\phi(\tp_{i})\Big)\\
			&\hspace{15em}(\mathcal{R}_{1,k,\ell})_j(q_{1},p_{1},\dots,q_{k},p_{k},\tq_{1},\tp_{1},\dots,\tq_{\ell},\tp_{\ell})\Bigg|\Bigg)\\
			&\leq\sum_{j=1}^{k}\|\nabla\varphi\|_{L^{{\infty}}(\bR^{3})}\|\varphi\|_{L^{{\infty}}(\bR^{3})}^{k+\ell-1}\|\phi\|_{L^{\infty}(\bR^{3})}^{k+\ell}\Bigg(\int\dq^{\otimes k}\dtq^{\otimes\ell}\Bigg|\int\ddp^{\otimes k}\dtp^{\otimes\ell}
			(\mathcal{R}_{1,k,\ell})_j(q_{1},p_{1},\dots,q_{k},p_{k},\tq_{1},\tp_{1},\dots,\tq_{\ell},\tp_{\ell})\Bigg|\Bigg)\\
			&\leq {C_t} \hbar^{\frac{1}{2}}.
		\end{align*}
		By interchanging the roles of the two species, we can get \eqref{eq:R_2kl_bdd} correspondingly.
	\end{proof}
	\subsection{$\widetilde{\cR}_{11,k,\ell}$, $\widetilde{\cR}_{12,k,\ell}$, and $\widetilde{\cR}_{22,k,\ell}$}
	Recall the remainder terms in Proposition \ref{prop:vla_k>1}, i.e., 
	\begin{align*}
		(\widetilde{\mathcal{R}}
		_{11,k,\ell})_{j}:= & \frac{1}{(2\pi)^3}\int(\dw\du)^{\otimes k}(\dtw\dtu)^{\otimes\ell}\int\dy\left[\int_{0}^{1}\dd s\,\nabla V_{11}(su_{j}+(1-s)w_{j}-y)\right]\\
		& \left(f_{q,p}^{\hbar}(w)\overline{f_{q,p}^{\hbar}(u)}\right)^{\otimes k}\left(f_{\tq,\tp}^{\hbar}(\tw)\overline{f_{\tq,\tp}^{\hbar}(\tu)}\right)^{\otimes\ell}
		\iint\dq_{k+1} \dd  p_{k+1} \ f_{q_{k+1},p_{k+1}}^{\hbar}(y)\int\dd v\ \overline{f_{{q_{k+1}},{p_{k+1}}}^{\hbar}(v)}\\
		&\left<b_{\tw_{\ell}}\dots b_{\tw_{1}}a_{w_{k}}\cdots a_{w_{1}}a_{y}\Psi_{N_1, N_2, t}, b_{\tu_{\ell}}\dots b_{\tu_{1}}a_{u_{k}}\cdots a_{u_{1}}a_{v}\Psi_{N_1, N_2, t}\right>\\
		& -\frac{1}{(2\pi)^3}\iint\dd q_{k+1}\dd p_{k+1}\,\nabla V_{11}(q_{j}-q_{k+1})m_{N_{1},N_{2},t}^{(k+1,\ell)}(q_1, p_1, \dots, q_{k+1}, p_{k+1},\tq_1, \tp_1, \dots, \tq_\ell, \tp_\ell),\\
		(\widetilde{\mathcal{R}}_{12,1,k,\ell})_{j}:= & \frac{1}{(2\pi)^3}  \int (\dw\du)^{\otimes k}(\dtw\dtu)^{\otimes \ell}  \int \dty\,  \left[\int_{0}^{1}\dd s\,\nabla V_{12}(su_{j}+(1-s)w_{j}-\ty)\right]\\
		&\left( f^\hbar_{q,p}(w) \overline{f^\hbar_{q,p}(u)} \right)^{\otimes k}
		\left( f^\hbar_{\tq,\tp}(\tw) \overline{f^\hbar_{\tq,\tp}(\tu)} \right)^{\otimes \ell} \iint\dd \tq_{\ell+1}\dd \tp_{\ell+1}\ f_{\tq_{\ell+1},\tp_{\ell+1}}^{\hbar}(\ty)\int\dd \tv\ \overline{f_{\tq_{\ell+1},\tp_{\ell+1}}^{\hbar}(\tv)}\\
		&\langle b_{\tw_\ell}\cdots b_{\tw_1}  b_{\ty} a_{w_k}\cdots  a_{w_1}\Psi_{N_1, N_2, t}, b_{\tu_\ell}\cdots b_{\tu_1} b_{\tv} a_{u_k}\cdots  a_{u_1} \Psi_{N_1, N_2, t}\rangle\\
		& -\frac{1}{(2\pi)^{3}} \iint\dd \tq_{\ell+1}\dd \tp_{\ell+1}\nabla V_{12}(q_{j}-\tq_{\ell+1})m_{N_{1},N_{2},t}^{(k,\ell+1)}(q_1, p_1, \dots, q_k, p_k,\tq_1, \tp_1, \dots, \tq_{\ell+1}, \tp_{\ell+1}),\\
		(\widetilde{\mathcal{R}}_{12,2,k,\ell})_{j}:= & \frac{1}{(2\pi)^3}   \int (\dw\du)^{\otimes k}(\dtw\dtu)^{\otimes \ell}  \int \dx\left[\int_{0}^{1}\dd s\,\nabla V_{12}(s\tu_{j}+(1-s)\tw_{j}-x)\right]\\
		&\left( f^\hbar_{q,p}(w) \overline{f^\hbar_{q,p}(u)} \right)^{\otimes k}
		\left( f^\hbar_{\tq,\tp}(\tw) \overline{f^\hbar_{\tq,\tp}(\tu)} \right)^{\otimes \ell} \iint\dd q_{k+1}\dd p_{k+1}\ f_{q_{k+1},p_{k+1}}^{\hbar}(x)\int\dd v\ \overline{f_{q_{k+1},p_{k+1}}^{\hbar}(v)}\\
		&\langle b_{\tw_\ell}\cdots  b_{\tw_1}   a_{w_k}\cdots a_{w_1} a_x \Psi_{N_1, N_2, t}, b_{\tu_\ell}\cdots  b_{\tu_1}   a_{u_k}\cdots a_{u_1} a_v \Psi_{N_1, N_2, t}\rangle\\
		& -\frac{1}{(2\pi)^{3}}\iint\dd q_{k+1}\dd p_{k+1}\,\nabla V_{12}(q_{k+1}-\tq_{j})m_{N_{1},N_{2},t}^{(k+1,\ell)}(q_1, p_1, \dots, q_{k+1}, p_{k+1},\tq_1, \tp_1, \dots, \tq_\ell, \tp_\ell),\\
		(\widetilde{\mathcal{R}}_{22,k,\ell})_j:=&\frac{1}{(2\pi)^{3}}  \int(\dw\du)^{\otimes k}(\dtw\dtu)^{\otimes\ell}\int\dty\,\left[\int_{0}^{1}\dd s\,\nabla V_{22}(s\tu_{j}+(1-s)\tw_{j}-\ty)\right]\\
		& \left(f_{q,p}^{\hbar}(w)\overline{f_{q,p}^{\hbar}(u)}\right)^{\otimes k}  \left(f_{\tq,\tp}^{\hbar}(\tw)\overline{f_{\tq,\tp}^{\hbar}(\tu)}\right)^{\otimes\ell}\iint\dd \tq_{\ell+1}\dd \tp_{\ell+1}\ f_{\tq_{\ell+1},\tp_{\ell+1}}^{\hbar}(\ty)\int\dd \tv\ \overline{f_{\tq_{\ell+1},\tp_{\ell+1}}^{\hbar}(\tv)}\\
		& \left<b_{\tw_{\ell}}\dots b_{\tw_{1}}b_{\ty}a_{w_{k}}\cdots a_{w_{1}}\Psi_{N_1, N_2, t},b_{\tu_{\ell}}\dots b_{\tu_{1}}b_{\tv}a_{u_{k}}\cdots a_{u_{1}}\Psi_{N_1, N_2, t}\right>\\
		& -\frac{1}{(2\pi)^{3}} \iint\dd \tq_{\ell+1}\dd \tp_{\ell+1}\nabla V_{22}(\tq_{j}-\tq_{\ell+1})m_{N_{1},N_{2},t}^{(k,\ell+1)}(q_1, p_1, \dots, q_k, p_k,\tq_1, \tp_1, \dots, \tq_{\ell+1}, \tp_{\ell+1}).
	\end{align*}
	
	\medskip
	
	\begin{Proposition}\label{prop:tildeR11kell}
		Under Assumption \ref{ass:main}, for arbitrary small $\delta>0$, we have: for $1 \leq k \leq N_1$, $0\leq\ell\leq N_2$, and $1\leq j\leq k$
		\begin{align}\label{eq:prop_tR_11klj}
			\bigg|\int\dq^{\otimes k} \ddp^{\otimes k} \dtq^{\otimes \ell} \dtp^{\otimes \ell} \, \varphi^{\otimes k} \phi^{\otimes k} \tilde\varphi^{\otimes \ell} \tilde\phi^{\otimes \ell} \nabla_{p_j}\cdot (\widetilde{\mathcal{R}}
			_{11,k,\ell})_{j}(q_{1},p_{1},\dots,q_{k},p_{k},\tq_{1},\tp_{1},\dots,\tq_{\ell},\tp_{\ell}) \bigg| \leq C \hbar^{\frac{1}{2}-\delta}
		\end{align}
		and for $0 \leq k \leq N_1$, $1\leq\ell\leq N_2$, and $1\leq j\leq \ell$
		\begin{align}\label{eq:prop_tR_22klj}
			\bigg|\int\dq^{\otimes k} \ddp^{\otimes k} \dtq^{\otimes \ell} \dtp^{\otimes \ell} \, \varphi^{\otimes k} \phi^{\otimes k} \tilde\varphi^{\otimes \ell} \tilde\phi^{\otimes \ell} \nabla_{\tp_j}\cdot (\widetilde{\mathcal{R}}
			_{22,k,\ell})_{j}(q_{1},p_{1},\dots,q_{k},p_{k},\tq_{1},\tp_{1},\dots,\tq_{\ell},\tp_{\ell}) \bigg| \leq C \hbar^{\frac{1}{2}-\delta},
		\end{align}
		where  $C$ depends on $\varphi, \phi, \tilde\varphi, \tilde\phi, f, k, \ell$.
	\end{Proposition}
	
	\begin{proof}
		For the sake of convenience, let us denote the left-hand side of the inequality \eqref{eq:prop_tR_11klj} we aim to control as $\widetilde{P}_{11,k,\ell,j}$. Subsequently, we can express this as:
		\begin{align*}
			&\widetilde{P}_{11,k,\ell,j}\\
			&=\bigg|\int\dq^{\otimes k} \ddp^{\otimes k} \dtq^{\otimes \ell} \dtp^{\otimes \ell} \, \varphi^{\otimes k} \phi^{\otimes k} \tilde\varphi^{\otimes \ell} \tilde\phi^{\otimes \ell} \nabla_{p_j}\cdot\\
			&\quad\bigg(\frac{1}{(2\pi)^3}\int(\dw\du)^{\otimes k}(\dtw\dtu)^{\otimes\ell}\int\dy\left[\int_{0}^{1}\dd s\,\nabla V_{11}(su_{j}+(1-s)w_{j}-y)\right]\\
			&\quad \left(f_{q,p}^{\hbar}(w)\overline{f_{q,p}^{\hbar}(u)}\right)^{\otimes k}\left(f_{\tq,\tp}^{\hbar}(\tw)\overline{f_{\tq,\tp}^{\hbar}(\tu)}\right)^{\otimes\ell}
			\iint\dq_{k+1} \dd  p_{k+1} \ f_{q_{k+1},p_{k+1}}^{\hbar}(y)\int\dd v\ \overline{f_{{q_{k+1}},{p_{k+1}}}^{\hbar}(v)}\\
			&\quad\left<b_{\tw_{\ell}}\dots b_{\tw_{1}}a_{w_{k}}\cdots a_{w_{1}}a_{y}\Psi_{N_1, N_2, t}, b_{\tu_{\ell}}\dots b_{\tu_{1}}a_{u_{k}}\cdots a_{u_{1}}a_{v}\Psi_{N_1, N_2, t}\right>\\
			& \quad-\frac{1}{(2\pi)^3}\iint\dd q_{k+1}\dd p_{k+1}\,\nabla V_{11}(q_{j}-q_{k+1})m_{N_{1},N_{2},t}^{(k+1,\ell)}(q_1, p_1, \dots, q_{k+1}, p_{k+1},\tq_1, \tp_1, \dots, \tq_\ell, \tp_\ell)\bigg)\bigg|\\
			&=\frac{1}{(2\pi)^3}\bigg|\int\dq^{\otimes k} \ddp^{\otimes k} \dtq^{\otimes \ell} \dtp^{\otimes \ell} \, \varphi^{\otimes k} \nabla_{p_j}\phi^{\otimes k} \tilde\varphi^{\otimes \ell} \tilde\phi^{\otimes \ell} \cdot\int(\dw\du)^{\otimes k}(\dtw\dtu)^{\otimes\ell}\iint\dy\dd v \iint\dq_{k+1} \dd  p_{k+1} \\
			&\quad \left(f_{q,p}^{\hbar}(w)\overline{f_{q,p}^{\hbar}(u)}\right)^{\otimes k}\left(f_{\tq,\tp}^{\hbar}(\tw)\overline{f_{\tq,\tp}^{\hbar}(\tu)}\right)^{\otimes\ell}
			f_{q_{k+1},p_{k+1}}^{\hbar}(y)\ \overline{f_{{q_{k+1}},{p_{k+1}}}^{\hbar}(v)}\\
			&\quad\int_{0}^{1}\dd s\left[\nabla V_{11}(su_{j}+(1-s)w_{j}-y)-\nabla V_{11}(q_{j}-q_{k+1})\right]\\
			&\quad\left<b_{\tw_{\ell}}\dots b_{\tw_{1}}a_{w_{k}}\cdots a_{w_{1}}a_{y}\Psi_{N_1, N_2, t}, b_{\tu_{\ell}}\dots b_{\tu_{1}}a_{u_{k}}\cdots a_{u_{1}}a_{v}\Psi_{N_1, N_2, t}\right>\bigg|\\
			&=\frac{1}{(2\pi)^{3}\hbar^\frac{3}{2}}\bigg|\int\dq^{\otimes k} \ddp^{\otimes k} \dtq^{\otimes \ell} \dtp^{\otimes \ell} \, \varphi^{\otimes k} \nabla_{p_j}\phi^{\otimes k} \tilde\varphi^{\otimes \ell} \tilde\phi^{\otimes \ell} \cdot\int(\dw\du)^{\otimes k}(\dtw\dtu)^{\otimes\ell}\iint\dy\dd v \iint\dq_{k+1} \dd  p_{k+1} \\
			&\quad \left(f_{q,p}^{\hbar}(w)\overline{f_{q,p}^{\hbar}(u)}\right)^{\otimes k}\left(f_{\tq,\tp}^{\hbar}(\tw)\overline{f_{\tq,\tp}^{\hbar}(\tu)}\right)^{\otimes\ell}
			e^{\frac{\mathrm{i}}{\hbar} p_{k+1} \cdot (y - v) } f \left(\frac{y-q_{k+1}}{\sqrt{\hbar}} \right) f \left(\frac{v-q_{k+1}}{\sqrt{\hbar}} \right) \\
			&\quad\int_{0}^{1}\dd s\left[\nabla V_{11}(su_{j}+(1-s)w_{j}-y)-\nabla V_{11}(q_{j}-q_{k+1})\right]\\
			&\quad\left<b_{\tw_{\ell}}\dots b_{\tw_{1}}a_{w_{k}}\cdots a_{w_{1}}a_{y}\Psi_{N_1, N_2, t}, b_{\tu_{\ell}}\dots b_{\tu_{1}}a_{u_{k}}\cdots a_{u_{1}}a_{v}\Psi_{N_1, N_2, t}\right>\bigg|\\
			&=\hbar^\frac{3}{2}\bigg|\int\dq^{\otimes k} \ddp^{\otimes k} \dtq^{\otimes \ell} \dtp^{\otimes \ell} \, \varphi^{\otimes k} \nabla_{p_j}\phi^{\otimes k} \tilde\varphi^{\otimes \ell} \tilde\phi^{\otimes \ell} \cdot\int(\dw\du)^{\otimes k}(\dtw\dtu)^{\otimes\ell}\int\dy \int\dq_{k+1} \\
			&\quad \left(f_{q,p}^{\hbar}(w)\overline{f_{q,p}^{\hbar}(u)}\right)^{\otimes k}\left(f_{\tq,\tp}^{\hbar}(\tw)\overline{f_{\tq,\tp}^{\hbar}(\tu)}\right)^{\otimes\ell}
			\left| f \left(\frac{y-q_{k+1}}{\sqrt{\hbar}} \right) \right|^2 \\
			&\quad\int_{0}^{1}\dd s\left[\nabla V_{11}(su_{j}+(1-s)w_{j}-y)-\nabla V_{11}(q_{j}-q_{k+1})\right]\\
			&\quad\left<b_{\tw_{\ell}}\dots b_{\tw_{1}}a_{w_{k}}\cdots a_{w_{1}}a_{y}\Psi_{N_1, N_2, t}, b_{\tu_{\ell}}\dots b_{\tu_{1}}a_{u_{k}}\cdots a_{u_{1}}a_{y}\Psi_{N_1, N_2, t}\right>\bigg|,
		\end{align*}
		where we have applied the $\hbar$-weighted Dirac-delta function $\delta_y (v) =  \frac{1}{(2\pi \hbar)^3}\int\dd p_{k+1}\  e^{\frac{\mathrm{i}}{\hbar}p_{k+1}\cdot(y-v)}$ in the last equality. To control the difference in the brackets, we use Assumption \ref{ass:main}.\ref{item:Interaction} that $\nabla V_{11}$ is Lipschitz continuous and get 
		\begin{equation}\label{seper_V}
			\begin{aligned}
				&\big|\nabla V_{11}(su_{j}+(1-s)w_{j}-y)-\nabla V_{11}(q_{j}-q_{k+1})\big|\leq C \big|(su_{j}+(1-s)w_{j}-y)-(q_{j}-q_{k+1})\big|\\
				\leq & C(|u_j-q_j|+|w_j-q_j|) +C |y-q_{k+1}|.
			\end{aligned}
		\end{equation}
		In accordance with Assumption \ref{ass:main}, the support of $f\left(\frac{y-q_{k+1}}{\sqrt{\hbar}} \right)$ is contained within $B_{R_1}$. Consequently, we establish the inequality $|y-q_{k+1}| \leq \sqrt{\hbar} R_1$. Similarly, employing the definitions of $f_{q_j,p_j}^{\hbar}(w_j)$ and $f_{q_j,p_j}^{\hbar}(u_j)$ as provided in \eqref{def_coherent}, we deduce that $|u_j-q_j|+|w_j-q_j| \leq \sqrt{\hbar} R_1$, following the same rationale. So we have
		\begin{align*}
			\widetilde{P}_{11,k,\ell,j}&\leq C\hbar^{2}\int\dq^{\otimes k} \dtq^{\otimes \ell}  \int(\dw\du)^{\otimes k}(\dtw\dtu)^{\otimes\ell}\int\dy \int\dq_{k+1}\\
			&\quad\bigg|\int\ddp^{\otimes k}\dtp^{\otimes \ell} \, \varphi^{\otimes k} \nabla_{p_j}\phi^{\otimes k} \tilde\varphi^{\otimes \ell} \tilde\phi^{\otimes \ell} \left(f_{q,p}^{\hbar}(w)\overline{f_{q,p}^{\hbar}(u)}\right)^{\otimes k}\left(f_{\tq,\tp}^{\hbar}(\tw)\overline{f_{\tq,\tp}^{\hbar}(\tu)}\right)^{\otimes\ell}
			\left| f \left(\frac{y-q_{k+1}}{\sqrt{\hbar}} \right) \right|^2 \\
			&\quad\left<b_{\tw_{\ell}}\dots b_{\tw_{1}}a_{w_{k}}\cdots a_{w_{1}}a_{y}\Psi_{N_1, N_2, t}, b_{\tu_{\ell}}\dots b_{\tu_{1}}a_{u_{k}}\cdots a_{u_{1}}a_{y}\Psi_{N_1, N_2, t}\right>\bigg|.
		\end{align*}
		Then we split the space by 
		\begin{align}\label{space_split}
			1 = \prod_{n=1}^k \left(\rchi_{(w_n-u_n)\in ({\Omega_\hbar^\alpha})^c}+ \rchi_{(w_n-u_n)\in {\Omega_\hbar^\alpha}} \right)
			\prod_{\kappa=1}^\ell \left(\rchi_{(\tw_\kappa-\tu_\kappa)\in ({\Omega_\hbar^\alpha})^c}+ \rchi_{(\tw_\kappa-\tu_\kappa)\in {\Omega_\hbar^\alpha}} \right)
		\end{align}
		where
		\begin{equation*}
			{\Omega_\hbar^\alpha} := \{x \in \bR^3:\ \max_{1\leq j \leq 3} |x_j|\leq \hbar^\alpha \}.
		\end{equation*}
		Then for all $\alpha \in (\frac{1}{2},1)$, we have
		\begin{align*}
			&\widetilde{P}_{11,k,\ell,j}\\
			&\leq C\hbar^{2}\int\dq^{\otimes k} \dtq^{\otimes \ell}  \int(\dw\du)^{\otimes k}(\dtw\dtu)^{\otimes\ell}\int\dy \int\dq_{k+1}\bigg|\int\ddp^{\otimes k}\dtp^{\otimes \ell} \, \varphi^{\otimes k} \nabla_{p_j}\phi^{\otimes k} \tilde\varphi^{\otimes \ell} \tilde\phi^{\otimes \ell}\\
			&\quad \prod_{n=1}^k \left(\rchi_{(w_n-u_n)\in ({\Omega_\hbar^\alpha})^c}+ \rchi_{(w_n-u_n)\in {\Omega_\hbar^\alpha}} \right)
			\prod_{\kappa=1}^\ell \left(\rchi_{(\tw_\kappa-\tu_\kappa)\in ({\Omega_\hbar^\alpha})^c}+ \rchi_{(\tw_\kappa-\tu_\kappa)\in {\Omega_\hbar^\alpha}} \right)\\
			&\quad \left(f_{q,p}^{\hbar}(w)\overline{f_{q,p}^{\hbar}(u)}\right)^{\otimes k}\left(f_{\tq,\tp}^{\hbar}(\tw)\overline{f_{\tq,\tp}^{\hbar}(\tu)}\right)^{\otimes\ell}
			\left| f \left(\frac{y-q_{k+1}}{\sqrt{\hbar}} \right) \right|^2 \\
			&\quad \left<b_{\tw_{\ell}}\dots b_{\tw_{1}}a_{w_{k}}\cdots a_{w_{1}}a_{y}\Psi_{N_1, N_2, t}, b_{\tu_{\ell}}\dots b_{\tu_{1}}a_{u_{k}}\cdots a_{u_{1}}a_{y}\Psi_{N_1, N_2, t}\right>\bigg|\\
			& \leq C\hbar^{2 - \frac{3}{2}(k+\ell)} \max_{\substack{0\leq \nu \leq k \\ 0\leq\zeta\leq\ell}}\int\dq^{\otimes k} \dtq^{\otimes \ell}  \int(\dw\du)^{\otimes k}(\dtw\dtu)^{\otimes\ell}\int\dy \int\dq_{k+1}\bigg|\int\ddp^{\otimes k}\dtp^{\otimes \ell} \varphi^{\otimes k} \nabla_{p_j}\phi^{\otimes k} \tilde\varphi^{\otimes \ell} \tilde\phi^{\otimes \ell} \\
			&\quad \left[ \rchi_{(w_1-u_1)\in {\Omega_\hbar^\alpha}} \cdots  \rchi_{(w_\nu-u_\nu)\in {\Omega_\hbar^\alpha}} \rchi_{(w_{\nu+1}-u_{\nu+1})\in ({\Omega_\hbar^\alpha})^c}\cdots\rchi_{(w_{k}-u_{k})\in ({\Omega_\hbar^\alpha})^c} \right]\\
			&\quad \left[ \rchi_{(\tw_1-\tu_1)\in {\Omega_\hbar^\alpha}} \cdots  \rchi_{(\tw_\zeta-\tu_\zeta)\in {\Omega_\hbar^\alpha}} \rchi_{(\tw_{\zeta+1}-\tu_{\zeta+1})\in ({\Omega_\hbar^\alpha})^c}\cdots\rchi_{(\tw_\ell-\tu_\ell)\in ({\Omega_\hbar^\alpha})^c} \right]\\
			&\quad  e^{\frac{i}{\hbar}\sum_{m=1}^k p_m \cdot (w_m -u_m)+\frac{i}{\hbar}\sum_{\sigma=1}^\ell \tp_\sigma \cdot (\tw_\sigma -\tu_\sigma)} \prod_{n=1}^k \left(f \left(\frac{w_n-q_n}{\sqrt{\hbar}}\right)f \left(\frac{u_n-q_n}{\sqrt{\hbar}}\right)\right)
			\left| f \left(\frac{y-q_{k+1}}{\sqrt{\hbar}} \right) \right|^2 \\
			&\quad  \prod_{\kappa=1}^\ell \left(f \left(\frac{\tw_\kappa-\tq_\kappa}{\sqrt{\hbar}}\right)f \left(\frac{\tu_\kappa-\tq_\kappa}{\sqrt{\hbar}}\right)\right)
			\left<b_{\tw_{\ell}}\dots b_{\tw_{1}}a_{w_{k}}\cdots a_{w_{1}}a_{y}\Psi_{N_1, N_2, t}, b_{\tu_{\ell}}\dots b_{\tu_{1}}a_{u_{k}}\cdots a_{u_{1}}a_{y}\Psi_{N_1, N_2, t}\right>\bigg|.
		\end{align*}
		In the last inequality, we expand the product into a sum of $2^{k+\ell}$ terms and select the one that maximizes the corresponding component. To illustrate this point, we choose a specific term enclosed within the brackets without loss of generality.
		
		We aim to simplify our computation by focusing on a specific component of the integral. In particular, we will examine the integral concerning $p^{\otimes k}$ and $\tp^{\otimes \ell}$. According to \cite[Lemma 2.5]{Chen2021} and the boundness of $\phi$, $\tilde\phi$, $\varphi$, and $\tilde\varphi$, we have the following relation:
		\begin{equation}\label{out_of_Omega}
			\begin{aligned}
				&\bigg|\int \ddp^{\otimes k} \dtp^{\otimes \ell} \nabla_{p_j}\phi^{\otimes k} \tilde\phi^{\otimes \ell}\varphi^{\otimes k}\tilde\varphi^{\otimes \ell}\cdot\rchi_{(w_{\nu+1}-u_{\nu+1})\in ({\Omega_\hbar^\alpha})^c}\rchi_{(w_{k}-u_{k})\in ({\Omega_\hbar^\alpha})^c} \\
				&\quad\rchi_{(\tw_{\zeta+1}-\tu_{\zeta+1})\in ({\Omega_\hbar^\alpha})^c}\cdots\rchi_{(\tw_\ell-\tu_\ell)\in ({\Omega_\hbar^\alpha})^c}
				e^{\frac{i}{\hbar}\sum_{m=1}^k p_m \cdot (w_m -u_m)+\frac{i}{\hbar}\sum_{\sigma=1}^\ell \tp_\sigma \cdot (\tw_\sigma -\tu_\sigma)}\bigg|\\
				&=\Bigg|\prod_{\substack{m = \nu + 1 \\ m \neq j}}^k \left(\int \ddp_m \rchi_{(w_m-u_m)\in ({\Omega_\hbar^\alpha})^c}  e^{\frac{i}{\hbar} p_m \cdot (w_m -u_m)} \phi_m(p_m)\right)\int \ddp_j \rchi_{(w_j-u_j)\in ({\Omega_\hbar^\alpha})^c}\nabla_{p_j}\phi_j(p_j)\\
				&\quad   e^{\frac{i}{\hbar} p_j \cdot (w_j -u_j)}\prod_{\sigma = \zeta + 1}^\ell \left(\int \ddp_\sigma \rchi_{(\tw_\sigma-\tu_\sigma)\in ({\Omega_\hbar^\alpha})^c}  e^{\frac{i}{\hbar} \tp_\sigma \cdot (\tw_\sigma- \tu_\sigma)}  \tilde\phi_\sigma(\tp_\sigma)\right)\\
				&\quad \int\ddp^{\otimes \nu} \dtp^{\otimes \zeta} e^{\frac{i}{\hbar}\sum\limits_{m=1}^\nu p_m \cdot (w_m -u_m) +\frac{i}{\hbar}\sum\limits_{\sigma=1}^\zeta \tp_\sigma \cdot (\tw_\sigma -\tu_\sigma)}\phi^{\otimes \nu}  \tilde\phi^{\otimes \zeta}\varphi^{\otimes k}\tilde\varphi^{\otimes \ell}\Bigg|\\    
				&\leq C_s \hbar^{(1-\alpha)(k+\ell-\nu-\zeta)s}
			\end{aligned}
		\end{equation}
		where the in the last line, $s$ is any fixed $s\in \mathbb{N}$ and $C_s$ is a constant depending on $s$.
		
		On the right-hand side of the equality above, we assume that $j > \nu$, as otherwise $\phi_j$ is part of $\phi^{\otimes \nu}$, and it doesn't introduce any additional complexity. Upon substituting this estimate back into $\widetilde{P}_{11,k,\ell,j}$, we obtain:
		\begin{align}
			\widetilde{P}_{11,k,\ell,j} &\leq {C_s} \max_{\substack{0\leq \nu \leq k \\ 0\leq\zeta\leq\ell}}\hbar^{2 - \frac{3}{2}(k+\ell) + (1-\alpha)(k+\ell-\nu-\zeta)s}\int\dq^{\otimes k} \dtq^{\otimes \ell}   \int(\dw\du)^{\otimes k}(\dtw\dtu)^{\otimes\ell} \nonumber\\
			&\quad \int\dy \int\dq_{k+1} \rchi_{(w_1-u_1)\in {\Omega_\hbar^\alpha}} \cdots  \rchi_{(w_\nu-u_\nu)\in {\Omega_\hbar^\alpha}}
			\rchi_{(\tw_1-\tu_1)\in {\Omega_\hbar^\alpha}} \cdots  \rchi_{(\tw_\zeta-\tu_\zeta)\in {\Omega_\hbar^\alpha}} \nonumber\\
			&\quad   \prod_{n=1}^k \left|f \left(\frac{w_n-q_n}{\sqrt{\hbar}}\right)f \left(\frac{u_n-q_n}{\sqrt{\hbar}}\right)\right|
			\left| f \left(\frac{y-q_{k+1}}{\sqrt{\hbar}} \right) \right|^2 \prod_{\upsilon=1}^\ell \left|f \left(\frac{\tw_\upsilon-\tq_\upsilon}{\sqrt{\hbar}}\right)f \left(\frac{\tu_\upsilon-\tq_\upsilon}{\sqrt{\hbar}}\right)\right|\nonumber\\
			&\quad \bigg|\left<b_{\tw_{\ell}}\dots b_{\tw_{1}}a_{w_{k}}\cdots a_{w_{1}}a_{y}\Psi_{N_1, N_2, t}, b_{\tu_{\ell}}\dots b_{\tu_{1}}a_{u_{k}}\cdots a_{u_{1}}a_{y}\Psi_{N_1, N_2, t}\right>\bigg|\nonumber\\    
			&\leq {C_s} \max_{\substack{0\leq \nu \leq k \nonumber\\ 0\leq\zeta\leq\ell}}\hbar^{2 - \frac{3}{2}(k+\ell) + (1-\alpha)(k+\ell-\nu-\zeta)s}\int\dq^{\otimes k} \dtq^{\otimes \ell}  \int(\dw\du)^{\otimes k}(\dtw\dtu)^{\otimes\ell}\nonumber\\  
			&\quad \int\dy \int\dq_{k+1} \rchi_{(w_1-u_1)\in {\Omega_\hbar^\alpha}} \cdots  \rchi_{(w_\nu-u_\nu)\in {\Omega_\hbar^\alpha}}
			\rchi_{(\tw_1-\tu_1)\in {\Omega_\hbar^\alpha}} \cdots  \rchi_{(\tw_\zeta-\tu_\zeta)\in {\Omega_\hbar^\alpha}} \nonumber\\
			&\quad   \prod_{n=1}^k \left|f \left(\frac{w_n-q_n}{\sqrt{\hbar}}\right)f \left(\frac{u_n-q_n}{\sqrt{\hbar}}\right)\right|
			\left| f \left(\frac{y-q_{k+1}}{\sqrt{\hbar}} \right) \right|^2 \prod_{\upsilon=1}^\ell \left|f \left(\frac{\tw_\upsilon-\tq_\upsilon}{\sqrt{\hbar}}\right)f \left(\frac{\tu_\upsilon-\tq_\upsilon}{\sqrt{\hbar}}\right)\right|\nonumber\\
			&\quad \|b_{\tw_{\ell}}\dots b_{\tw_{1}}a_{w_{k}}\cdots a_{w_{1}}a_{y}\Psi_{N_1, N_2, t}\|^2,\label{bound_P_11klj}  
		\end{align}
		where we utilize the Schwarz inequality, focusing solely on the $w$ part due to its symmetry with the $u$ part. Concerning the integrals with respect to $u_i$ for $1\leq i \leq \nu$, we ascertain that
		\begin{align}\label{int_ui}
			\int \du_i \rchi_{(w_i-u_i)\in {\Omega_\hbar^\alpha}} \left| f \left(\frac{u_i-q_i}{\sqrt{\hbar}}\right)\right| \leq {C_s}  \int \du_i \rchi_{(w_i-u_i)\in {\Omega_\hbar^\alpha}}  \leq {C_s} \hbar^{3\alpha}.
		\end{align}
		This same analysis extends to the second set of variables, $\tu_i$ for $1 \leq i \leq \zeta$. For the remaining $u_i$ for $\nu + 1 \leq i \leq k$, we perform a variable substitution, replacing $u_i$ with $\sqrt{\hbar}u_i + q_i$, which results in
		\begin{align}\label{int_ui_rest}
			\int \du_i  \left| f \left(\frac{u_i-q_i}{\sqrt{\hbar}}\right)\right| = \hbar^{\frac32} \int \du_i  \left| f \left(u_i\right)\right|=C \hbar^{\frac32}.
		\end{align}
		This substitution also applies to $\tu_i$ for $\zeta + 1 \leq i \leq \ell$. For the integrals involving $q_j$ for $1 \leq j \leq k$, we employ a similar variable substitution, replacing $q_j$ with $\sqrt{\hbar}q_j + w_j$, which leads to the following:
		\begin{align}\label{int_qj}
			\int \dq_j  \left| f \left(\frac{w_j-q_j}{\sqrt{\hbar}}\right)\right| = \hbar^{\frac32} \int \dq_j  \left| f \left(q_j\right)\right|=C \hbar^{\frac32}.
		\end{align}
		This computation is equally applicable to $q_{k+1}$ and $\tq_j$ for $1 \leq j \leq \ell$. Subsequently, we substitute \eqref{int_ui}, \eqref{int_ui_rest}, and \eqref{int_qj} into \eqref{bound_P_11klj}. Consequently, we obtain
		\begin{align*}
			\widetilde{P}_{11,k,\ell,j} &\leq C\max_{\substack{0\leq \nu \leq k \\ 0\leq\zeta\leq\ell}}\hbar^{2 - \frac{3}{2}(k+\ell) + (1-\alpha)(k+\ell-\nu-\zeta)s+3\alpha(\nu+\zeta)+\frac32(k-\nu+\ell-\zeta)+\frac32(k+1+\ell)} \\
			& \int\dw^{\otimes k}\dtw^{\otimes\ell}\int\dy \left<b_{\tw_{\ell}}\dots b_{\tw_{1}}a_{w_{k}}\cdots a_{w_{1}}a_{y}\Psi_{N_1, N_2, t}, b_{\tw_{\ell}}\dots b_{\tw_{1}}a_{w_{k}}\cdots a_{w_{1}}a_{y}\Psi_{N_1, N_2, t}\right>\\ 
			=& C\max_{\substack{0\leq \nu \leq k \\ 0\leq\zeta\leq\ell}} \hbar^{\frac72+\left((1-\alpha)s+\frac32\right)(k+\ell)+\left(-(1-\alpha)s+3\alpha-\frac32\right)(\nu+\zeta)}\\
			&\left<\Psi_{N_1, N_2, t},\cN_1(\cN_1 - 1)\cdots(\cN_1 - k) \cN_2(\cN_2 - 1)\cdots(\cN_2 - \ell + 1)\Psi_{N_1, N_2, t}\right>\\
			&\leq  C \max_{\substack{0\leq \nu \leq k \\ 0\leq\zeta\leq\ell}}\hbar^{\frac72+\left((1-\alpha)s+\frac32\right)(k+\ell)+\left(-(1-\alpha)s+3\alpha-\frac32\right)(\nu+\zeta)} N_1^{k+1} N_2^\ell\\
			&\leq  C \max_{\substack{0\leq \nu \leq k \\ 0\leq\zeta\leq\ell}}\hbar^{\frac72+\left((1-\alpha)s+\frac32\right)(k+\ell)+\left(-(1-\alpha)s+3\alpha-\frac32\right)(\nu+\zeta)} N^{k+\ell+1}\\
			=&  C \max_{\substack{0\leq \nu \leq k \\ 0\leq\zeta\leq\ell}}\hbar^{\frac72+\left((1-\alpha)s+\frac32\right)(k+\ell)+\left(-(1-\alpha)s+3\alpha-\frac32\right)(\nu+\zeta) - 3(k+\ell+1)}\\
			=& C\max_{\substack{0\leq \nu \leq k \\ 0\leq\zeta\leq\ell}} \hbar^{\frac12 + \left((1-\alpha)s-\frac32\right)(k+\ell-\nu-\zeta)+(3\alpha-3)(\nu+\zeta)}.
		\end{align*}
		
		Now, we choose
		\begin{align*}
			s = \left\lceil \frac{3(2\alpha-1)}{2(1-\alpha)} \right\rceil.
		\end{align*}
		Thus, we can select an appropriate $\alpha \in (\frac12, 1)$ to achieve the desired result:
		\begin{align*}
			\widetilde{P}_{11,k,\ell,j} &\leq C \hbar^{\frac12 + (3\alpha-3)(k+\ell-\nu-\zeta)+(3\alpha-3)(\nu+\zeta)} = C \hbar^{\frac12 + 3(\alpha-1)(k+\ell)} \leq C \hbar^{\frac12-\delta}.
		\end{align*}
		
		By reversing the roles of the two species, we can derive the corresponding inequality \eqref{eq:prop_tR_22klj}.
	\end{proof}
	
	\begin{Proposition}\label{prop:tildeR12kell}
		Under Assumption \ref{ass:main}, for arbitrary small $\delta>0$, we have: for $1 \leq k \leq N_1$, $0\leq\ell\leq N_2$, and $1\leq j\leq k$
		\begin{align}\label{eq:prop_tR_121klj}
			\bigg|\int\dq^{\otimes k} \ddp^{\otimes k} \dtq^{\otimes \ell} \dtp^{\otimes \ell} \, \varphi^{\otimes k} \phi^{\otimes k} \tilde\varphi^{\otimes \ell} \tilde\phi^{\otimes \ell} \nabla_{p_j}\cdot (\widetilde{\mathcal{R}}
			_{12,1,k,\ell})_{j}(q_{1},p_{1},\dots,q_{k},p_{k},\tq_{1},\tp_{1},\dots,\tq_{\ell},\tp_{\ell}) \bigg| \leq C \hbar^{\frac{1}{2}-\delta}
		\end{align}
		and for $0 \leq k \leq N_1$, $1\leq\ell\leq N_2$, and $1\leq j\leq \ell$
		\begin{align}\label{eq:prop_tR_122klj}
			\bigg|\int\dq^{\otimes k} \ddp^{\otimes k} \dtq^{\otimes \ell} \dtp^{\otimes \ell} \, \varphi^{\otimes k} \phi^{\otimes k} \tilde\varphi^{\otimes \ell} \tilde\phi^{\otimes \ell} \nabla_{\tp_j}\cdot (\widetilde{\mathcal{R}}
			_{12,2,k,\ell})_{j}(q_{1},p_{1},\dots,q_{k},p_{k},\tq_{1},\tp_{1},\dots,\tq_{\ell},\tp_{\ell}) \bigg| \leq C \hbar^{\frac{1}{2}-\delta},
		\end{align}
		where  $C$ depends on $\varphi, \phi, \tilde\varphi, \tilde\phi, f, k, \ell$.
	\end{Proposition}
	\begin{proof}
		Controlling $ \nabla_{p_j}\cdot (\widetilde{\mathcal{R}}_{12,1,k,\ell})_{j}$, which involves interactions between two species, closely resembles the single-species case. Therefore, we primarily follow the proof outlined in Proposition \ref{prop:tildeR11kell}, avoiding the redundancy of details that are largely consistent with the previous analysis. Instead, we focus on highlighting the distinctions.
		
		For ease of reference, we denote the left-hand side of the inequality \eqref{eq:prop_tR_121klj}, which we seek to manage, as $\widetilde{P}_{12,1,k,\ell,j}$. Subsequently, we can express it as follows:
		\begin{align*}
			\widetilde{P}_{12,1,k,\ell,j}
			&=\bigg|\int\dq^{\otimes k} \ddp^{\otimes k} \dtq^{\otimes \ell} \dtp^{\otimes \ell} \, \varphi^{\otimes k} \phi^{\otimes k} \tilde\varphi^{\otimes \ell} \tilde\phi^{\otimes \ell} \nabla_{p_j}\cdot\\
			&\quad\bigg(\frac{1}{(2\pi)^3}  \int (\dw\du)^{\otimes k}(\dtw\dtu)^{\otimes \ell}  \int \dty\,  \left[\int_{0}^{1}\dd s\,\nabla V_{12}(su_{j}+(1-s)w_{j}-\ty)\right]\\
			&\left( f^\hbar_{q,p}(w) \overline{f^\hbar_{q,p}(u)} \right)^{\otimes k}
			\left( f^\hbar_{\tq,\tp}(\tw) \overline{f^\hbar_{\tq,\tp}(\tu)} \right)^{\otimes \ell} \iint\dd \tq_{\ell+1}\dd \tp_{\ell+1}\ f_{\tq_{\ell+1},\tp_{\ell+1}}^{\hbar}(\ty)\int\dd \tv\ \overline{f_{\tq_{\ell+1},\tp_{\ell+1}}^{\hbar}(\tv)}\\
			&\langle b_{\tw_\ell}\cdots b_{\tw_1}  b_{\ty} a_{w_k}\cdots  a_{w_1}\Psi_{N_1, N_2, t}, b_{\tu_\ell}\cdots b_{\tu_1} b_{\tv} a_{u_k}\cdots  a_{u_1} \Psi_{N_1, N_2, t}\rangle\\
			& -\frac{1}{(2\pi)^{3}} \iint\dd \tq_{\ell+1}\dd \tp_{\ell+1}\nabla V_{12}(q_{j}-\tq_{\ell+1})m_{N_{1},N_{2},t}^{(k,\ell+1)}(q_1, p_1, \dots, q_k, p_k,\tq_1, \tp_1, \dots, \tq_{\ell+1}, \tp_{\ell+1})\bigg)\bigg|\\
			&=\hbar^\frac{3}{2}\bigg|\int\dq^{\otimes k} \ddp^{\otimes k} \dtq^{\otimes \ell} \dtp^{\otimes \ell} \, \varphi^{\otimes k} \phi^{\otimes k} \tilde\varphi^{\otimes \ell} \tilde\phi^{\otimes \ell} \nabla_{p_j}\cdot\int(\dw\du)^{\otimes k}(\dtw\dtu)^{\otimes\ell}\int\dty \int\dtq_{\ell+1} \\
			&\quad \left(f_{q,p}^{\hbar}(w)\overline{f_{q,p}^{\hbar}(u)}\right)^{\otimes k}\left(f_{\tq,\tp}^{\hbar}(\tw)\overline{f_{\tq,\tp}^{\hbar}(\tu)}\right)^{\otimes\ell}
			\left| f \left(\frac{\ty-\tq_{\ell+1}}{\sqrt{\hbar}} \right) \right|^2 \\
			&\quad\int_{0}^{1}\dd s\left[\nabla V_{12}(su_{j}+(1-s)w_{j}-\ty)-\nabla V_{12}(q_{j}-\tq_{\ell+1})\right]\\
			&\quad\left<b_{\tw_{\ell}}\dots b_{\tw_{1}}b_{\ty}a_{w_{k}}\cdots a_{w_{1}}\Psi_{N_1, N_2, t}, b_{\tu_{\ell}}\dots b_{\tu_{1}}b_{\ty}a_{u_{k}}\cdots a_{u_{1}}\Psi_{N_1, N_2, t}\right>\bigg|.
		\end{align*}
		By leveraging Assumption \ref{ass:main}.\ref{item:Interaction} that $\nabla V_{12}$ is Lipschitz continuous and considering the support of $f$, we arrive at the following result:
		\begin{equation}
			\begin{aligned}
				&\big|\nabla V_{12}(su_{j}+(1-s)w_{j}-\ty)-\nabla V_{12}(q_{j}-\tq_{\ell+1})\big|\leq C(|u_j-q_j|+|w_j-q_j|) +C |\ty-\tq_{\ell+1}|\leq C \sqrt{\hbar}.
			\end{aligned}
		\end{equation}
		
		Continuing further, we meticulously proceed through the proof of Proposition \ref{prop:tildeR11kell}, methodically applying the space split \eqref{space_split}, and drawing upon \cite[Lemma 2.5]{Chen2021}, for any positive integer $s$, which leads to:
		\begin{align*}
			&\widetilde{P}_{12,1,k,\ell,j}\\
			&\leq C\hbar^{2 - \frac{3}{2}(k+\ell)} \max_{\substack{0\leq \nu \leq k \\ 0\leq\zeta\leq\ell}}\int\dq^{\otimes k} \dtq^{\otimes \ell}  \int(\dw\du)^{\otimes k}(\dtw\dtu)^{\otimes\ell}\int\dty \int\dtq_{\ell+1}\bigg|\int\ddp^{\otimes k}\dtp^{\otimes \ell}  \varphi^{\otimes k} \nabla_{p_j} \phi^{\otimes k} \tilde\varphi^{\otimes \ell} \tilde\phi^{\otimes \ell} \\
			&\quad  \left[ \rchi_{(w_1-u_1)\in {\Omega_\hbar^\alpha}} \cdots  \rchi_{(w_\nu-u_\nu)\in {\Omega_\hbar^\alpha}} \rchi_{(w_{\nu+1}-u_{\nu+1})\in ({\Omega_\hbar^\alpha})^c}\cdots\rchi_{(w_{k}-u_{k})\in ({\Omega_\hbar^\alpha})^c} \right]\\
			&\quad \left[ \rchi_{(\tw_1-\tu_1)\in {\Omega_\hbar^\alpha}} \cdots  \rchi_{(\tw_\zeta-\tu_\zeta)\in {\Omega_\hbar^\alpha}} \rchi_{(\tw_{\zeta+1}-\tu_{\zeta+1})\in ({\Omega_\hbar^\alpha})^c}\cdots\rchi_{(\tw_\ell-\tu_\ell)\in ({\Omega_\hbar^\alpha})^c} \right]\\
			&\quad  e^{\frac{i}{\hbar}\sum_{m=1}^k p_m \cdot (w_m -u_m)+\frac{i}{\hbar}\sum_{\sigma=1}^\ell \tp_\sigma \cdot (\tw_\sigma -\tu_\sigma)} \prod_{n=1}^k \left(f \left(\frac{w_n-q_n}{\sqrt{\hbar}}\right)f \left(\frac{u_n-q_n}{\sqrt{\hbar}}\right)\right)
			\left| f \left(\frac{\ty-\tq_{\ell+1}}{\sqrt{\hbar}} \right) \right|^2 \\
			&\quad  \prod_{\kappa=1}^\ell \left(f \left(\frac{\tw_\kappa-\tq_\kappa}{\sqrt{\hbar}}\right)f \left(\frac{\tu_\kappa-\tq_\kappa}{\sqrt{\hbar}}\right)\right)
			\left<b_{\tw_{\ell}}\dots b_{\tw_{1}}b_{\ty}a_{w_{k}}\cdots a_{w_{1}}\Psi_{N_1, N_2, t}, b_{\tu_{\ell}}\dots b_{\tu_{1}}b_{\ty}a_{u_{k}}\cdots a_{u_{1}}\Psi_{N_1, N_2, t}\right>\bigg|\\    
			&\leq C \max_{\substack{0\leq \nu \leq k \\ 0\leq\zeta\leq\ell}}\hbar^{2 - \frac{3}{2}(k+\ell) + (1-\alpha)(k+\ell-\nu-\zeta)s} \int\dq^{\otimes k} \dtq^{\otimes \ell}  \int(\dw\du)^{\otimes k}(\dtw\dtu)^{\otimes\ell}  \\  
			&\quad \int\dty \int\dtq_{\ell+1} \rchi_{(w_1-u_1)\in {\Omega_\hbar^\alpha}} \cdots  \rchi_{(w_\nu-u_\nu)\in {\Omega_\hbar^\alpha}}
			\rchi_{(\tw_1-\tu_1)\in {\Omega_\hbar^\alpha}} \cdots  \rchi_{(\tw_\zeta-\tu_\zeta)\in {\Omega_\hbar^\alpha}} \\
			&\quad   \prod_{n=1}^k \left|f \left(\frac{w_n-q_n}{\sqrt{\hbar}}\right)f \left(\frac{u_n-q_n}{\sqrt{\hbar}}\right)\right|
			\left| f \left(\frac{\ty-\tq_{\ell+1}}{\sqrt{\hbar}} \right) \right|^2 \prod_{\upsilon=1}^\ell \left|f \left(\frac{\tw_\upsilon-\tq_\upsilon}{\sqrt{\hbar}}\right)f \left(\frac{\tu_\upsilon-\tq_\upsilon}{\sqrt{\hbar}}\right)\right|\\
			&\quad  \|b_{\tw_{\ell}}\dots b_{\tw_{1}}b_{\ty}a_{w_{k}}\cdots a_{w_{1}}\Psi_{N_1, N_2, t}\|^2\\
			&\leq C  \max_{\substack{0\leq \nu \leq k \\ 0\leq\zeta\leq\ell}}  \hbar^{\frac72+\left((1-\alpha)s+\frac32\right)(k+\ell)+\left(-(1-\alpha)s+3\alpha-\frac32\right)(\nu+\zeta)} \\
			&\quad  \int\dw^{\otimes k}\dtw^{\otimes\ell}\int\dty \left<b_{\tw_{\ell}}\dots b_{\tw_{1}}b_{\ty}a_{w_{k}}\cdots a_{w_{1}}\Psi_{N_1, N_2, t}, b_{\tw_{\ell}}\dots b_{\tw_{1}}b_{\ty}a_{w_{k}}\cdots a_{w_{1}}\Psi_{N_1, N_2, t}\right>\\ 
			&= C \max_{\substack{0\leq \nu \leq k \\ 0\leq\zeta\leq\ell}} \hbar^{\frac72+\left((1-\alpha)s+\frac32\right)(k+\ell)+\left(-(1-\alpha)s+3\alpha-\frac32\right)(\nu+\zeta)}\\
			&\quad \left<\Psi_{N_1, N_2, t},\cN_1(\cN_1 - 1)\cdots(\cN_1 - k + 1) \cN_2(\cN_2 - 1)\cdots(\cN_2 - \ell)\Psi_{N_1, N_2, t}\right>\\
			&\leq C  \max_{\substack{0\leq \nu \leq k \\ 0\leq\zeta\leq\ell}} \hbar^{\frac12 + \left((1-\alpha)s-\frac32\right)(k+\ell-\nu-\zeta)+(3\alpha-3)(\nu+\zeta)}
			\leq C \hbar^{\frac12-\delta},
		\end{align*}
		where we choose the values of $s$ and $\alpha$ in accordance with our earlier selection.
		
		The bound in \eqref{eq:prop_tR_122klj} can be easily obtained by interchanging the roles of the two species.
	\end{proof}
	
	\subsection{$\widehat{\cR}_{11,k,\ell}$, $\widehat{\cR}_{12,k,\ell}$, and $\widehat{\cR}_{22,k,\ell}$}
	Recall the remainder terms in Proposition \ref{prop:vla_k>1}, i.e., 
	\begin{equation}
		\begin{aligned}
			\widehat{\mathcal{R}}_{11,k,\ell}:= & \frac{\mathrm{i}\hbar^{2}}{2}\int(\dw\du)^{\otimes k}(\dtw\dtu)^{\otimes\ell}\sum_{j\neq i}^{k}\bigg[V_{11}(u_{j}-u_{i})-V_{11}(w_{j}-w_{i})\bigg]
			\left(f_{q,p}^{\hbar}(w)\overline{f_{q,p}^{\hbar}(u)}\right)^{\otimes k}\\
			& \left(f_{\tq,\tp}^{\hbar}(\tw)\overline{f_{\tq,\tp}^{\hbar}(\tu)}\right)^{\otimes\ell}
			\left<b_{\tw_{\ell}}\dots b_{\tw_{1}}a_{w_{k}}\cdots a_{w_{1}}\Psi_{N_1, N_2, t},b_{\tu_{\ell}}\dots b_{\tu_{1}}a_{u_{k}}\cdots a_{u_{1}}\Psi_{N_1, N_2, t}\right>,\\
			\widehat{\mathcal{R}}_{12,k,\ell}:= & \mathrm{i}\hbar^{2}\int(\dw\du)^{\otimes k}(\dtw\dtu)^{\otimes\ell}\sum_{j=1}^{k}\sum_{i=1}^{\ell}\bigg[V_{12}(u_{j}-\tu_{i})-V_{12}(w_{j}-\tw_{i})\bigg]\left(f_{q,p}^{\hbar}(w)\overline{f_{q,p}^{\hbar}(u)}\right)^{\otimes k}\\
			& \left(f_{\tq,\tp}^{\hbar}(\tw)\overline{f_{\tq,\tp}^{\hbar}(\tu)}\right)^{\otimes\ell}
			\left<b_{\tw_{\ell}}\dots b_{\tw_{1}}a_{w_{k}}\cdots a_{w_{1}}\Psi_{N_1, N_2, t},b_{\tu_{\ell}}\dots b_{\tu_{1}}a_{u_{k}}\cdots a_{u_{1}}\Psi_{N_1, N_2, t}\right>,\\
			\widehat{\mathcal{R}}_{22,k,\ell} := &  \frac{\mathrm{i}\hbar^2}{2}\int(\dw\du)^{\otimes k}(\dtw\dtu)^{\otimes\ell}\sum_{j\neq i}^{\ell}\bigg[V_{22}(\tu_{j}-\tu_{i})-V_{22}(\tw_{j}-\tw_{i})\bigg]
			\left(f_{q,p}^{\hbar}(w)\overline{f_{q,p}^{\hbar}(u)}\right)^{\otimes k}\\
			& \left(f_{\tq,\tp}^{\hbar}(\tw)\overline{f_{\tq,\tp}^{\hbar}(\tu)}\right)^{\otimes\ell} \left<b_{\tw_{\ell}}\dots b_{\tw_{1}}a_{w_{k}}\cdots a_{w_{1}}\Psi_{N_1, N_2, t},b_{\tu_{\ell}}\dots b_{\tu_{1}}a_{u_{k}}\cdots a_{u_{1}}\Psi_{N_1, N_2, t}\right>.
		\end{aligned}
	\end{equation}
	
	\begin{Proposition}\label{prop:hatR11kell}
		Under Assumption \ref{ass:main}, for arbitrary small $\delta>0$, we have: for $1 < k \leq N_1$, $0\leq\ell\leq N_2$
		\begin{align}\label{eq:prop_hR_11klj}
			\bigg|\int\dq^{\otimes k} \ddp^{\otimes k} \dtq^{\otimes \ell} \dtp^{\otimes \ell} \, \varphi^{\otimes k} \phi^{\otimes k} \tilde\varphi^{\otimes \ell} \tilde\phi^{\otimes \ell}\widehat{\mathcal{R}}_{11,k,\ell}(q_{1},p_{1},\dots,q_{k},p_{k},\tq_{1},\tp_{1},\dots,\tq_{\ell},\tp_{\ell}) \bigg| \leq C \hbar^{\frac52-\delta}
		\end{align}
		and for $0 \leq k \leq N_1$, $1<\ell\leq N_2$
		\begin{align}\label{eq:prop_hR_22klj}
			\bigg|\int\dq^{\otimes k} \ddp^{\otimes k} \dtq^{\otimes \ell} \dtp^{\otimes \ell} \, \varphi^{\otimes k} \phi^{\otimes k} \tilde\varphi^{\otimes \ell} \tilde\phi^{\otimes \ell} \widehat{\mathcal{R}}_{22,k,\ell}(q_{1},p_{1},\dots,q_{k},p_{k},\tq_{1},\tp_{1},\dots,\tq_{\ell},\tp_{\ell}) \bigg| \leq C \hbar^{\frac52-\delta},
		\end{align}
		where  $C$ depends on $\varphi, \phi, \tilde\varphi, \tilde\phi, f, k, \ell$.
	\end{Proposition}
	\begin{proof}
		We designate the left-hand side of inequality \eqref{eq:prop_hR_11klj} as $\widehat{P}_{11,k,\ell}$. Certain aspects of its structure resemble those of $\widetilde{P}_{11,k,\ell,j}$. Consequently, we can incorporate elements from the proof of Proposition \ref{prop:tildeR11kell}. With the Assumption \ref{ass:main}.\ref{item:Interaction} that $\nabla V_{11}$ is bounded, we get
		\begin{equation}\nonumber
			\begin{aligned}
				&\left|V_{11}(u_{j}-u_{i})-V_{11}(w_{j}-w_{i})\right| \leq C |(u_{j}-u_{i})-(w_{j}-w_{i})| \\
				=& C |((u_j-q_j)-(w_j-q_j))-((u_i-q_i)-(w_i-q_i))| \leq C(|u_j-q_j|+|w_j-q_j|+|u_i-q_i|+|w_i-q_i|).
			\end{aligned}
		\end{equation}  
		Since $f$ is compactly supported, following the same analysis as in the proof of Proposition \ref{prop:tildeR11kell}, we have 
		$$\left|V_{11}(u_{j}-u_{i})-V_{11}(w_{j}-w_{i})\right|\leq C\sqrt{\hbar}.$$
		This leads to the following result:
		\begin{align}
			&\widehat{P}_{11,k,\ell}\nonumber\\
			&=\bigg|\int\dq^{\otimes k} \ddp^{\otimes k} \dtq^{\otimes \ell} \dtp^{\otimes \ell} \, \varphi^{\otimes k} \phi^{\otimes k} \tilde\varphi^{\otimes \ell} \tilde\phi^{\otimes \ell}\nonumber\\
			&\quad \frac{\mathrm{i}\hbar^{2}}{2}\int(\dw\du)^{\otimes k}(\dtw\dtu)^{\otimes\ell}\sum_{j\neq i}^{k}\bigg[V_{11}(u_{j}-u_{i})-V_{11}(w_{j}-w_{i})\bigg]
			\left(f_{q,p}^{\hbar}(w)\overline{f_{q,p}^{\hbar}(u)}\right)^{\otimes k}\nonumber\\
			&\quad  \left(f_{\tq,\tp}^{\hbar}(\tw)\overline{f_{\tq,\tp}^{\hbar}(\tu)}\right)^{\otimes\ell}
			\left<b_{\tw_{\ell}}\dots b_{\tw_{1}}a_{w_{k}}\cdots a_{w_{1}}\Psi_{N_1, N_2, t},b_{\tu_{\ell}}\dots b_{\tu_{1}}a_{u_{k}}\cdots a_{u_{1}}\Psi_{N_1, N_2, t}\right>\bigg|\nonumber\\
			&\leq C\hbar^{\frac52}\int(\dw\du)^{\otimes k}\bigg|\int\dq^{\otimes k} \ddp^{\otimes k} \dtq^{\otimes \ell} \dtp^{\otimes \ell} \, \varphi^{\otimes k} \phi^{\otimes k} \tilde\varphi^{\otimes \ell} \tilde\phi^{\otimes \ell}\int(\dtw\dtu)^{\otimes\ell}
			\left(f_{q,p}^{\hbar}(w)\overline{f_{q,p}^{\hbar}(u)}\right)^{\otimes k}\nonumber\\
			&\quad  \left(f_{\tq,\tp}^{\hbar}(\tw)\overline{f_{\tq,\tp}^{\hbar}(\tu)}\right)^{\otimes\ell}
			\left<b_{\tw_{\ell}}\dots b_{\tw_{1}}a_{w_{k}}\cdots a_{w_{1}}\Psi_{N_1, N_2, t},b_{\tu_{\ell}}\dots b_{\tu_{1}}a_{u_{k}}\cdots a_{u_{1}}\Psi_{N_1, N_2, t}\right>\bigg|\nonumber\\ 
			&=C\hbar^{\frac52}\int(\dw\du)^{\otimes k}\bigg|\int\dq^{\otimes k} \ddp^{\otimes k} \dtq^{\otimes \ell} \dtp^{\otimes \ell} \, \varphi^{\otimes k} \phi^{\otimes k} \tilde\varphi^{\otimes \ell} \tilde\phi^{\otimes \ell}\int(\dtw\dtu)^{\otimes\ell}
			\nonumber\\
			&\quad \prod_{n=1}^k \left(\rchi_{(w_n-u_n)\in ({\Omega_\hbar^\alpha})^c}+ \rchi_{(w_n-u_n)\in {\Omega_\hbar^\alpha}} \right)
			\prod_{\kappa=1}^\ell \left(\rchi_{(\tw_\kappa-\tu_\kappa)\in ({\Omega_\hbar^\alpha})^c}+ \rchi_{(\tw_\kappa-\tu_\kappa)\in {\Omega_\hbar^\alpha}} \right)\nonumber\\
			&\quad  \left(f_{q,p}^{\hbar}(w)\overline{f_{q,p}^{\hbar}(u)}\right)^{\otimes k}\left(f_{\tq,\tp}^{\hbar}(\tw)\overline{f_{\tq,\tp}^{\hbar}(\tu)}\right)^{\otimes\ell}
			\left<b_{\tw_{\ell}}\dots b_{\tw_{1}}a_{w_{k}}\cdots a_{w_{1}}\Psi_{N_1, N_2, t},b_{\tu_{\ell}}\dots b_{\tu_{1}}a_{u_{k}}\cdots a_{u_{1}}\Psi_{N_1, N_2, t}\right>\bigg|\nonumber\\ 
			&\leq {C\hbar^{\frac52 - \frac{3}{2}(k+\ell)}} \max_{\substack{0\leq \nu \leq k \nonumber\\ 
					0\leq\zeta\leq\ell}}\int(\dw\du)^{\otimes k}\bigg|\int\dq^{\otimes k} \ddp^{\otimes k} \dtq^{\otimes \ell} \dtp^{\otimes \ell} \, \varphi^{\otimes k} \phi^{\otimes k} \tilde\varphi^{\otimes \ell} \tilde\phi^{\otimes \ell}\int(\dtw\dtu)^{\otimes\ell}\nonumber\\
			&\quad \left[ \rchi_{(w_1-u_1)\in {\Omega_\hbar^\alpha}} \cdots  \rchi_{(w_\nu-u_\nu)\in {\Omega_\hbar^\alpha}} \rchi_{(w_{\nu+1}-u_{\nu+1})\in ({\Omega_\hbar^\alpha})^c}\cdots\rchi_{(w_{k}-u_{k})\in ({\Omega_\hbar^\alpha})^c} \right]\nonumber\\
			&\quad \left[ \rchi_{(\tw_1-\tu_1)\in {\Omega_\hbar^\alpha}} \cdots  \rchi_{(\tw_\zeta-\tu_\zeta)\in {\Omega_\hbar^\alpha}} \rchi_{(\tw_{\zeta+1}-\tu_{\zeta+1})\in ({\Omega_\hbar^\alpha})^c}\cdots\rchi_{(\tw_\ell-\tu_\ell)\in ({\Omega_\hbar^\alpha})^c} \right]\nonumber\\
			&\quad  e^{\frac{i}{\hbar}\sum_{m=1}^k p_m \cdot (w_m -u_m)+\frac{i}{\hbar}\sum_{\sigma=1}^\ell \tp_\sigma \cdot (\tw_\sigma -\tu_\sigma)} \prod_{n=1}^k \left(f \left(\frac{w_n-q_n}{\sqrt{\hbar}}\right)f \left(\frac{u_n-q_n}{\sqrt{\hbar}}\right)\right)      \nonumber\\
			&\quad  \prod_{\kappa=1}^\ell \left(f \left(\frac{\tw_\kappa-\tq_\kappa}{\sqrt{\hbar}}\right)f \left(\frac{\tu_\kappa-\tq_\kappa}{\sqrt{\hbar}}\right)\right)
			\left<b_{\tw_{\ell}}\dots b_{\tw_{1}}a_{w_{k}}\cdots a_{w_{1}}\Psi_{N_1, N_2, t}, b_{\tu_{\ell}}\dots b_{\tu_{1}}a_{u_{k}}\cdots a_{u_{1}}\Psi_{N_1, N_2, t}\right>\bigg|,\label{eq:P_11kl}
		\end{align}
		where we employ the space split \eqref{space_split} in the last equality. Next, we apply the control described in \eqref{out_of_Omega}. Given the absence of $\nabla_{p_j}$ in the current case, the analysis is further simplified, jinyeopc{for any positive integer $s$,} yielding:
		\begin{align*}
			&\widehat{P}_{11,k,\ell}\\
			&\leq C \max_{\substack{0\leq \nu \leq k \\ 0\leq\zeta\leq\ell}}
			{\hbar^{\frac52 - \frac{3}{2}(k+\ell) + (1-\alpha)(k+\ell-\nu-\zeta)s}}\int (\dw\du)^{\otimes k}\int\dq^{\otimes k} \dtq^{\otimes \ell}   \int(\dtw\dtu)^{\otimes\ell} \\
			&\quad \rchi_{(w_1-u_1)\in {\Omega_\hbar^\alpha}} \cdots  \rchi_{(w_\nu-u_\nu)\in {\Omega_\hbar^\alpha}}
			\rchi_{(\tw_1-\tu_1)\in {\Omega_\hbar^\alpha}} \cdots  \rchi_{(\tw_\zeta-\tu_\zeta)\in {\Omega_\hbar^\alpha}} \\
			&\quad   \prod_{n=1}^k \left|f \left(\frac{w_n-q_n}{\sqrt{\hbar}}\right)f \left(\frac{u_n-q_n}{\sqrt{\hbar}}\right)\right|
			\prod_{\upsilon=1}^\ell \left|f \left(\frac{\tw_\upsilon-\tq_\upsilon}{\sqrt{\hbar}}\right)f \left(\frac{\tu_\upsilon-\tq_\upsilon}{\sqrt{\hbar}}\right)\right|\|b_{\tw_{\ell}}\dots b_{\tw_{1}}a_{w_{k}}\cdots a_{w_{1}}\Psi_{N_1, N_2, t}\|^2\\
			&\leq C\max_{\substack{0\leq \nu \leq k \\ 0\leq\zeta\leq\ell}}\hbar^{\frac52 - \frac{3}{2}(k+\ell) + (1-\alpha)(k+\ell-\nu-\zeta)s+3\alpha(\nu+\zeta)+\frac32(k-\nu+\ell-\zeta)+\frac32(k+\ell)}\\
			&\quad \left<\Psi_{N_1, N_2, t},\cN_1(\cN_1 - 1)\cdots(\cN_1 - k + 1) \cN_2(\cN_2 - 1)\cdots(\cN_2 - \ell + 1)\Psi_{N_1, N_2, t}\right>\\
			&\leq  C \max_{\substack{0\leq \nu \leq k \\ 0\leq\zeta\leq\ell}}\hbar^{\frac52+\left((1-\alpha)s+\frac32\right)(k+\ell)+\left(-(1-\alpha)s+3\alpha-\frac32\right)(\nu+\zeta)}N^{k+\ell}\\
			&= C\max_{\substack{0\leq \nu \leq k \\ 0\leq\zeta\leq\ell}} \hbar^{\frac52 + \left((1-\alpha)s-\frac32\right)(k+\ell-\nu-\zeta)+(3\alpha-3)(\nu+\zeta)}\leq C \hbar^{\frac52-\delta}.
		\end{align*}
		In the preceding computation, we make use of \eqref{int_ui}, \eqref{int_ui_rest}, and \eqref{int_qj}, while ensuring that our choice of values for $s$ and $\alpha$ aligns with our previous selection.
		
		By exchanging the roles of the two species, we can establish the corresponding inequality \eqref{eq:prop_hR_22klj}.
	\end{proof}
	
	\begin{Proposition}\label{prop:hatR12kell}
		Under Assumption \ref{ass:main}, for arbitrary small $\delta>0$, we have: for $1 \leq k \leq N_1$, $1\leq\ell\leq N_2$
		\begin{align}\label{eq:prop_hR_12klj}
			\bigg|\int\dq^{\otimes k} \ddp^{\otimes k} \dtq^{\otimes \ell} \dtp^{\otimes \ell} \, \varphi^{\otimes k} \phi^{\otimes k} \tilde\varphi^{\otimes \ell} \tilde\phi^{\otimes \ell}\widehat{\mathcal{R}}_{12,k,\ell}(q_{1},p_{1},\dots,q_{k},p_{k},\tq_{1},\tp_{1},\dots,\tq_{\ell},\tp_{\ell}) \bigg| \leq C \hbar^{\frac52-\delta}
		\end{align}
		where  $C$ depends on $\varphi, \phi, \tilde\varphi, \tilde\phi, f, k, \ell$.
	\end{Proposition}
	\begin{proof}
		The proof strategy for inequality \eqref{eq:prop_hR_12klj} is largely similar to that of Proposition \ref{prop:hatR11kell}, with our main focus on the distinguishing aspects. We designate the left-hand side of inequality \eqref{eq:prop_hR_12klj} as $\widehat{P}_{12,k,\ell}$. With the Assumption \ref{ass:main}.\ref{item:Interaction} that $\nabla V_{12}$ is bounded, we get
		\begin{align*}
			&\left|V_{12}(u_{j}-\tu_{i})-V_{12}(w_{j}-\tw_{i})\right| \leq C |(u_{j}-\tu_{i})-(w_{j}-\tw_{i})| \\
			& \leq  C(|u_j-q_j|+|w_j-q_j|+|\tu_i-\tq_i|+|\tw_i-\tq_i|)\leq C\sqrt{\hbar}.
		\end{align*}
		This leads to the following result:
		\begin{align*}
			\widehat{P}_{12,k,\ell}
			&=\bigg|\int\dq^{\otimes k} \ddp^{\otimes k} \dtq^{\otimes \ell} \dtp^{\otimes \ell} \, \varphi^{\otimes k} \phi^{\otimes k} \tilde\varphi^{\otimes \ell} \tilde\phi^{\otimes \ell}\\
			&\qquad \mathrm{i}\hbar^{2}\int(\dw\du)^{\otimes k}(\dtw\dtu)^{\otimes\ell}\sum_{j=1}^{k}\sum_{i=1}^{\ell}\bigg[V_{12}(u_{j}-\tu_{i})-V_{12}(w_{j}-\tw_{i})\bigg]\left(f_{q,p}^{\hbar}(w)\overline{f_{q,p}^{\hbar}(u)}\right)^{\otimes k}\\
			&\qquad \left(f_{\tq,\tp}^{\hbar}(\tw)\overline{f_{\tq,\tp}^{\hbar}(\tu)}\right)^{\otimes\ell}
			\left<b_{\tw_{\ell}}\dots b_{\tw_{1}}a_{w_{k}}\cdots a_{w_{1}}\Psi_{N_1, N_2, t},b_{\tu_{\ell}}\dots b_{\tu_{1}}a_{u_{k}}\cdots a_{u_{1}}\Psi_{N_1, N_2, t}\right>\bigg|\\
			&\leq C\hbar^{\frac52}\int(\dw\du)^{\otimes k}\bigg|\int\dq^{\otimes k} \ddp^{\otimes k} \dtq^{\otimes \ell} \dtp^{\otimes \ell} \, \varphi^{\otimes k} \phi^{\otimes k} \tilde\varphi^{\otimes \ell} \tilde\phi^{\otimes \ell}\int(\dtw\dtu)^{\otimes\ell}
			\left(f_{q,p}^{\hbar}(w)\overline{f_{q,p}^{\hbar}(u)}\right)^{\otimes k}\\
			&\qquad \left(f_{\tq,\tp}^{\hbar}(\tw)\overline{f_{\tq,\tp}^{\hbar}(\tu)}\right)^{\otimes\ell}
			\left<b_{\tw_{\ell}}\dots b_{\tw_{1}}a_{w_{k}}\cdots a_{w_{1}}\Psi_{N_1, N_2, t},b_{\tu_{\ell}}\dots b_{\tu_{1}}a_{u_{k}}\cdots a_{u_{1}}\Psi_{N_1, N_2, t}\right>\bigg|.
		\end{align*}
		The right-hand side of the inequality is identical to the right-hand side of the first inequality in \eqref{eq:P_11kl}, which is used to bound $\widehat{P}_{11,k,\ell}$ and is bounded by $C \hbar^{\frac52-\delta}$. Consequently, we immediately obtain $\widehat{P}_{12,k,\ell}\leq C \hbar^{\frac52-\delta}$.
	\end{proof}
	
	\section{Proof of the Main Results}
	\label{sec:Uniqueness Vlasov hierarchy}
	\begin{proof}[Proof of Theorem \ref{thm:main}]
		We apply the Dunford-Pettis theorem to derive the weak convergence of $m_{N_1,N_2,t}^{(k,\ell)}$ as $N\rightarrow\infty$, as demonstrated in the proof of Proposition 2.7, \cite{Chen2021}. More precisely, Lemma \ref{lem:prop_kHusimi} ensures the uniform $L^1$-boundedness of $m_{N_1,N_2,t}^{(k,\ell)}$. 
		For arbitrary $\epsilon>0$,  by taking $\delta=\epsilon$, 
		then for all $E\in \mathbb{R}^{6k+6\ell}$ with $\operatorname{Vol}(E)\le \delta$,
		it holds the uniform absolute continuity:
		
		\begin{equation*}
			\int \cdots \int_E m_{N_1,N_2,t}^{(k,\ell)} \leqslant\|m_{N_1,N_2,t}^{(k,\ell)}\|_{\infty} \operatorname{Vol}(E) \leqslant \varepsilon.
		\end{equation*}
		Furthermore, note that the first moment of $m_{N_1,N_2,t}^{(k,\ell)}$ is uniformly bounded by
		Lemma \ref{prop:finite-moment}. For any $\epsilon>0$, by takeing $r=\varepsilon^{-1}(2 \pi)^{3 k+3\ell} C(t)$, we can deduce the tightness:
		\begin{equation*}
			\begin{aligned}
				&\frac{1}{(2 \pi)^{3 k+3\ell}} \int \cdots \int_{\left|\mathbf{q}_{k}\right|+\left|\mathbf{p}_{k}\right| \geqslant r}
				\int_{\left|\tilde{\mathbf{q}}_{\ell}\right|+\left|\tilde{\mathbf{p}}_{\ell}\right| \geqslant r}
				(\dd q\dd p)^{\otimes k}(\dd \tq\dd \tp)^{\otimes \ell} m_{N_1,N_2,t}^{(k,\ell)} \\
				&\leqslant 
				\frac{1}{r} \frac{1}{(2 \pi)^{3 k+3\ell}} \int \ldots \int
				(\dd q\dd p)^{\otimes k}(\dd \tq\dd \tp)^{\otimes \ell}
				\big(|\mathbf{q}_{k}|+|\mathbf{p}_{k}| +|\tilde{\mathbf{q}}_{\ell}|+|\tilde{\mathbf{p}}_{\ell}| \big) m_{N_1,N_2,t}^{(k,\ell)} \leqslant \varepsilon .
			\end{aligned}
		\end{equation*}
		
		Together, these conditions are equivalent to the uniform integrability, and allow us to invoke the Dunford-Pettis theorem. As a result, it ensures the existence of a convergent subsequence of $m_{N_1,N_2,t}^{(k,\ell)}$, as $N\rightarrow\infty$. 
		
		Additionally, Cantor’s diagonal procedure shows that we can take the same convergent subsequence of $m_{N_1,N_2,t}^{(k,\ell)}$ for all $k,\ell > 1$. Then by the error estimates obtained in Proposition \ref{prop:R1kell}-\ref{prop:hatR12kell},  we can obtain that the limit satisfies the infinite Vlasov hierarchy \eqref{eq:BBGKY_limit} in the sense of distribution, by directly taking the limit in the weak formulation \eqref{BBGKY_k=1} and \eqref{BBGKY_k>1}. The uniqueness of this solution will be established later in Proposition \ref{prop:uniqueness_of_solution}, which implies that the sequence $m_{N_1,N_2,t}^{(k,\ell)}$ itself converges weakly to the solution of the
		infinite hierarchy.
	\end{proof}
	
	\begin{proof}[Proof of Theorem \ref{thm:main2}]
		Assuming the initial data is factorized as $m_{1}^{\otimes k}\otimes m_{2}^{\otimes \ell}$, we denote the solutions of the coupled Vlasov equations with initial data $m_{1}$ and $m_{2}$ as $m_{1,t}$ and $m_{2,t}$. These solutions satisfy
		\begin{align*}
			\partial_{t}m_{1,t}(q,p)+p\cdot\nabla_{q}m_{1,t}(q,p) & =\nabla_{q}\left((V_{11}*\rho_{1, t})(q)+(V_{12}*\rho_{2, t})(q)\right)\cdot\nabla_{p}m_{1,t}(q,p)\\
			\partial_{t}m_{2,t}(\tq,\tp)+\tp\cdot\nabla_{\tq}m_{2,t}(\tq,\tp) & =\nabla_{\tq}\left((V_{22}*\rho_{2, t})(\tq)+(V_{21}*\rho_{1, t})(\tq)\right)\cdot\nabla_{\tp}m_{2,t}(\tq,\tp).
		\end{align*}
		This implies that
		\begin{align*}
			& \partial_{t}\left(m_{1,t}^{\otimes k}m_{2,t}^{\otimes\ell}\right)\\
			& =\sum_{j=1}^{k}\left(m_{1,t}^{\otimes(k-1)}m_{2,t}^{\otimes\ell}\right)\partial_{t}m_{1,t}(q_{j},p_{j})+\sum_{i=1}^{\ell}\left(m_{1,t}^{\otimes k}m_{2,t}^{\otimes(\ell-1)}\right)\partial_{t}m_{2,t}(\tq_{i},\tp_{i})\\
			& =\sum_{j=1}^{k}\left(m_{1,t}^{\otimes(k-1)}m_{2,t}^{\otimes\ell}\right)\left(-p_{j}\cdot\nabla_{q_{j}}m_{1,t}(q_{j},p_{j})+\nabla_{q}\left((V_{11}*\rho_{1, t})(q_{j})+(V_{12}*\rho_{2, t})(q_{j})\right)\cdot\nabla_{p}m_{1,t}(q_{j},p_{j})\right)\\
			& \quad+\sum_{i=1}^{\ell}\left(m_{1,t}^{\otimes k}m_{2,t}^{\otimes(\ell-1)}\right)\left(-\tp_{j}\cdot\nabla_{\tq_{j}}m_{2,t}(\tq_{j},\tp_{j})+\nabla_{\tq_{j}}\left((V_{22}*\rho_{2, t})(\tq_{j})+(V_{21}*\rho_{1, t})(\tq_{j})\right)\cdot\nabla_{\tp_{j}}m_{2,t}(\tq_{j},\tp_{j})\right).
		\end{align*}
		By the definition of $\rho_{\beta,t}$, we have that $(V_{\alpha\beta}*\rho_{\beta,t})(q_j)=\frac{1}{(2\pi)^3} \int \dq_i\ddp_i V_{\alpha\beta}(q_j - q_i) m_{\alpha,t}(q_i,p_i)$.
		Then we have
		\begin{align*}
			& \partial_{t}\left(m_{1,t}^{\otimes k}m_{2,t}^{\otimes\ell}\right)+\left(\mathbf{p}_{k}\cdot\nabla_{\mathbf{q}_{k}}+\tilde{\mathbf{p}}_{\ell}\cdot\nabla_{\tilde{\mathbf{q}}_{\ell}}\right)\left(m_{1,t}^{\otimes k}m_{2,t}^{\otimes\ell}\right)\\
			& = \frac{1}{(2\pi)^3} \sum_{j=1}^{k} \nabla_{\mathbf{p}_j} \cdot \iint \dd q_{k+1}\dd p_{k+1} \nabla V_{11}(q_j - q_{k+1})\left(m_{1,t}^{\otimes (k + 1)}m_{2,t}^{\otimes\ell}\right)\\
			&+\frac{1}{(2\pi)^3} \sum_{j=1}^{k} \nabla_{\mathbf{p}_j} \cdot \iint \dd \tilde q_{\ell+1}\dd \tilde p_{\ell+1} \nabla V_{12}(q_j - \tilde q_{\ell+1})\left(m_{1,t}^{\otimes k}m_{2,t}^{\otimes(\ell + 1)}\right)\\
			&+ \frac{1}{(2\pi)^3} \sum_{j=1}^{\ell}  \nabla_{\mathbf{\tp}_j} \cdot \iint \dd q_{k+1}\dd p_{k+1} \nabla V_{21}(q_{k+1} - \tq_{j})\left(m_{1,t}^{\otimes (k + 1)}m_{2,t}^{\otimes\ell}\right)\\
			&+\frac{1}{(2\pi)^3} \sum_{j=1}^{\ell} \nabla_{\mathbf{\tp}_j} \cdot \iint \dd \tilde q_{\ell+1}\dd \tilde p_{\ell+1}\nabla V_{22}(\tq_j - \tilde q_{\ell+1})\left(m_{1,t}^{\otimes k}m_{2,t}^{\otimes(\ell + 1)}\right).
		\end{align*}
		This is a specific form of the infinite hierarchy \eqref{eq:BBGKY_limit}. Since the limit $m_{\infty,\infty,t}^{(k,\ell)}$ is the unique solution of the hierarchy, it implies that $m_{\infty,\infty,t}^{(k,\ell)} = m_{1,t}^{\otimes k}m_{2,t}^{\otimes\ell}$ for all $(k,\ell) \in \bN \times \bN \setminus \{(0,0)\}$. Note that $ \|m_{1,t}^{\otimes k}m_{2,t}^{\otimes\ell}\|_{L^1(\mathbb{R}^3k\times\mathbb{R}^\ell)} = n_1^k n_2^\ell$. Combining this result with \cite[Theorem 7.12]{Villani2003}, we deduce that $m_{N_1,N_2,t}^{(1,0)}$ and $m_{N_1,N_2,t}^{(0,1)}$ converge to the solutions $m_{1,t}$ and $m_{2,t}$ in the 1-Wasserstein distance.
	\end{proof}
	
	To ascertain the uniqueness of the infinite hierarchy, we adopt the approach outlined in \cite{Narnhofer1981}. Let $\mu_{\infty,\infty,t}^{(k,\ell)}$ be the characteristic function for the probability measure $(2\pi)^{-3(k+\ell)} n_1^{-k}n_2^{-\ell} m_{\infty,\infty,t}^{(k,\ell)}$ on $\bR^{6k}\times\bR^{6\ell}$ as
	\begin{equation}
		\begin{aligned}
			&\mu_{\infty,\infty,t}^{(k,\ell)}(\xi_{1},\eta_{1},\dots,\xi_{k},\eta_{k},\txi_{1},\teta_{1},\dots,\txi_{\ell},\teta_{\ell})\\
			&:=\frac{1}{(2\pi)^{3(k+\ell)} n_1^k n_2^\ell} \int \dd  m_{\infty,\infty,t}^{(k,\ell)} (q_1, p_1, \dots, q_k, p_k,\tq_1, \tp_1, \dots, \tq_\ell, \tp_\ell) e^{i\left(\sum_{j=1}^{k} (\xi_j p_j + \eta_j q_j)+\sum_{j=1}^{\ell} (\txi_j \tp_j + \teta_j \tq_j)\right)}
		\end{aligned}
	\end{equation}
	$\forall \xi_j, \eta_j, \txi_j, \teta_j \in \bR^3$. We denote $\widehat{V}_{\alpha\beta}$ to be the Fourier transform of $V_{\alpha\beta}$ for each $\alpha,\beta\in\{1,2\}$ according to the convention that $V_{\alpha\beta}(x)=\int \ddp \widehat{V}_{\alpha\beta}(p) e^{ipx}$. Then we have the following proposition:
	
	\begin{Proposition}\label{prop:uniqueness_of_solution}
		
		Assume $\widehat{V}_{\alpha\beta}\in C_0$ for all $\alpha,\beta\in\{1,2\}$. Then the two species Vlasov hierarchy \eqref{eq:BBGKY_limit} has a unique solution.    
	\end{Proposition}
	\begin{proof}
		We reformulate the Vlasov hierarchy in interaction representation. Thus, we start by defining
		\[
		\bar{\mu}_t:=\left\{ \bar{\mu}_t^{(k,\ell)} \;\bigg|\; (k,\ell) \in \bN \times \bN \setminus \{(0,0)\} \right\}
		\] 
		to be the sequence of characteristic functions, related to $\mu_{\infty,\infty,t}$ by the formula
		\[
		\bar{\mu}_t^{(k,\ell)} := \mu_{\infty,\infty,t}^{(k,\ell)}(\xi_1-\eta_1 t, \eta_1, \dots, \xi_k - \eta_k t, \eta_k, \txi_1-\teta_1 t, \teta_1, \dots, \txi_\ell - \teta_\ell t, \teta_\ell );
		\]
		and then note that the Vlasov hierarchy \eqref{eq:BBGKY_limit} is equivalent to the equation
		\begin{equation}\label{eq: equivalent to the equation for Vlasov hierarchy}
			\bar{\mu}_t = \bar{\mu}_0 + \int_{0}^{t} \dt_1 K(t_1) \bar{\mu}_{t_1},
		\end{equation}
		where
		\begin{align*}
			&\left(K(t)\,\bar{\mu}_{t}\right)^{(k,\ell)}(\xi_{1},\eta_{1},\dots,\xi_{k},\eta_{k},\txi_{1},\teta_{1},\dots,\txi_{\ell},\teta_{\ell})\\
			&:={n_1}\sum_{j=1}^{k}\int\dd\eta_{k+1}\,\widehat{V}_{11}(\eta_{k+1})\,\eta_{k+1}(\xi_{j}-\eta_{j}t)\\
			&\qquad\times\bar{\mu}_{t}^{(k+1,\ell)}(\xi_{1},\eta_{1},\dots,\xi_{j}+\eta_{k+1}t,\eta_{j}+\eta_{k+1},\dots,\xi_{k},\eta_{k},-\eta_{k+1}t,-\eta_{k+1},\txi_{1},\teta_{1},\dots,\txi_{\ell},\teta_{\ell})\\
			&\quad+{n_2}\sum_{j=1}^{k}\int\dd\teta_{\ell+1}\,\widehat{V}_{12}(\teta_{\ell+1})\,\teta_{\ell+1}(\xi_{j}-\eta_{j}t)\\
			&\quad\qquad\times\bar{\mu}_{t}^{(k,\ell+1)}(\xi_{1},\eta_{1},\dots,\xi_{j}+\teta_{\ell+1}t,\eta_{j}+\teta_{\ell+1},\dots,\xi_{k},\eta_{k},\txi_{1},\teta_{1},\dots,\txi_{\ell},\teta_{\ell},-\teta_{\ell+1}t,-\teta_{\ell+1})\\
			&\quad+{n_1}\sum_{j=1}^{\ell}\int\dd\eta_{k+1}\,\widehat{V}_{21}(\eta_{k+1})\,\eta_{k+1}(\txi_{j}-\teta_{j}t)\\
			&\quad\qquad\times\bar{\mu}_{t}^{(k+1,\ell)}(\xi_{1},\eta_{1},\dots,\xi_{k},\eta_{k},-\eta_{k+1}t,-\eta_{k+1},\txi_{1},\teta_{1},\dots,\txi_{j}+\eta_{k+1}t,\teta_{j}+\eta_{k+1},\dots\txi_{\ell},\teta_{\ell})\\
			&\quad+{n_2}\sum_{j=1}^{\ell}\int\dd\teta_{\ell+1}\,\widehat{V}_{22}(\teta_{\ell+1})\,\teta_{\ell+1}(\txi_{j}-\teta_{j}t)\\
			&\quad\qquad\times\bar{\mu}_{t}^{(k,\ell+1)}(\xi_{1},\eta_{1},\dots,\xi_{k},\eta_{k},\txi_{1},\teta_{1},\dots,,\txi_{j}+\teta_{\ell+1}t,\teta_{j}+\teta_{\ell+1},\dots\txi_{\ell},\teta_{\ell},-\teta_{\ell+1}t,-\teta_{\ell+1}).
		\end{align*}
		For any given $L>0$, a straightforward iteration of  \eqref{eq: equivalent to the equation for Vlasov hierarchy} yields the formula
		\begin{equation}\label{eq: iteration of mt}
			\bar{\mu}_t = \bar{\mu}_0 + \sum_{j=1}^{L-1} \int_{0}^{t} \dt_1 \dots \int_{0}^{t_j - 1} \dt_j \, K(t_1) \dots K(t_j) \, \bar{\mu}_{0} + \Delta_L \bar{\mu}_t
		\end{equation}
		where
		\begin{equation}\label{eq: def of Delta_L bar m_t}
			\Delta_L \bar{\mu}_t := \int_{0}^{t} \dt_1 \dots \int_{0}^{t_{L - 1}} \dt_L  K(t_1) \dots K(t_L) \bar{\mu}_{t_L}.
		\end{equation}
		
		Now we can write \eqref{eq: def of Delta_L bar m_t} differently as follows:
		First, we define operator $\theta_{1,j}(\xi,\eta)$ on the functions on $(\bR^3)^{2k}\times(\bR^3)^{2\ell}$ with $(\xi,\eta)\in (\bR^3)^2$ and $k\geq j$ such that
		\begin{equation}\label{eq:theta_1}
			\begin{aligned}
				&\left(\theta_{1,j}(\xi,\eta)f\right)(\xi_1,\eta_1,\dots,\xi_k,\eta_k,\txi_1,\teta_1,\dots,\txi_\ell,\teta_\ell)
				\\
				&:=
				f(\xi_1,\eta_1,\dots,\xi_j-\xi,\eta_j-\eta,\dots,\xi_k,\eta_k,\txi_1,\teta_1,\dots,\txi_\ell,\teta_\ell)
			\end{aligned}
		\end{equation}
		and operator $\theta_{2,j}(\txi,\teta)$ on the functions on $(\bR^3)^{2k}\times(\bR^3)^{2\ell}$ with $(\txi,\teta)\in (\bR^3)^2$ and $\ell\geq j$ such that
		\begin{equation}\label{eq:theta_2}
			\begin{aligned}
				&\left(\theta_{2,j}(\txi,\teta)f\right)(\xi_1,\eta_1,\dots,\xi_k,\eta_k,\txi_1,\teta_1,\dots,\txi_\ell,\teta_\ell)\\
				&:= f(\xi_1,\eta_1,\dots,\xi_k,\eta_k,\txi_1,\teta_1,\dots,\txi_j - \txi, \teta_j - \teta, \dots, \txi_\ell,\teta_\ell).
			\end{aligned}
		\end{equation}
		
		The integral in Equation \eqref{eq: equivalent to the equation for Vlasov hierarchy}, depending on the specific values of $\alpha,\beta$ in $V_{\alpha\beta}$, can be categorized into four branches. For the sake of clarity, we aim to represent these four branches uniformly. To do so, we introduce some notations and functions: $\xi_{1,j}:=\xi_j,\ \xi_{2,j}:=\txi_j,\ \fks_1(k,\ell):=k,\ \fks_2(k,\ell):=\ell,\ \fkk_\alpha:=2-\alpha,\ \fkl_\alpha := \alpha -1,\ \fkK_m:=\sum_{d=1}^{m}\fkk_{\beta_d},\ \fkL_m:=\sum_{d=1}^{m}\fkl_{\beta_d}$ for $\alpha,\beta_d\in\{1,2\}$. With these definitions, we can obtain:
		\begin{align*}
			&\left(K(t)\,\bar{\mu}_{t}\right)^{(k,\ell)}(\xi_{1},\eta_{1},\dots,\xi_{k},\eta_{k},\txi_{1},\teta_{1},\dots,\txi_{\ell},\teta_{\ell})\\
			&=-{n_1}\sum_{j=1}^{\fks_1(k,\ell)}\int\dd\eta_{1,\fks_1(k,\ell)+1}\dd\xi_{1,\fks_1(k,\ell)+1}\delta(\xi_{1,\fks_1(k,\ell)+1}-\eta_{1,\fks_1(k,\ell)+1}t)\,\widehat{V}_{11}(-\eta_{1,\fks_1(k,\ell)+1})\,\eta_{1,\fks_1(k,\ell)+1}(\xi_{1,j}-\eta_{1,j}t)\\
			&\qquad\times\theta_{1,j}(\xi_{1,\fks_1(k,\ell)+1},\eta_{1,\fks_1(k,\ell)+1})\bar{\mu}_{t}^{(k+\fkk_1,\ell+\fkl_1)}(\xi_{1},\eta_{1},\dots,\xi_{k+\fkk_1},\eta_{k+\fkk_1},\txi_{1},\teta_{1},\dots,\txi_{\ell+\fkl_1},\teta_{\ell+\fkl_1})\\
			&\quad-{n_2}\sum_{j=1}^{\fks_1(k,\ell)}\int\dd\eta_{2,\fks_2(k,\ell)+1}\dd\xi_{2,\fks_2(k,\ell)+1}\delta(\xi_{2,\fks_2(k,\ell)+1}-\eta_{2,\fks_2(k,\ell)+1}t)\,\widehat{V}_{12}(-\eta_{2,\fks_2(k,\ell)+1})\,\eta_{2,\fks_2(k,\ell)+1}(\xi_{1,j}-\eta_{1,j}t)\\
			&\quad\qquad\times\theta_{1,j}(\xi_{2,\fks_2(k,\ell)+1},\eta_{2,\fks_2(k,\ell)+1})\bar{\mu}_{t}^{(k+\fkk_2,\ell+\fkl_2)}(\xi_{1},\eta_{1},\dots,\xi_{k+\fkk_2},\eta_{k+\fkk_2},\txi_{1},\teta_{1},\dots,\txi_{\ell+\fkl_2},\teta_{\ell+\fkl_2})\\
			&\quad-{n_1}\sum_{j=1}^{\fks_2(k,\ell)}\int\dd\eta_{1,\fks_1(k,\ell)+1}\dd\xi_{1,\fks_1(k,\ell)+1}\delta(\xi_{1,\fks_1(k,\ell)+1}-\eta_{1,\fks_1(k,\ell)+1}t)\,\widehat{V}_{21}(-\eta_{1,\fks_1(k,\ell)+1})\,\eta_{1,\fks_1(k,\ell)+1}(\xi_{2,j}-\eta_{2,j}t)\\
			&\quad\qquad\times\theta_{2,j}(\xi_{1,\fks_1(k,\ell)+1},\eta_{1,\fks_1(k,\ell)+1})\bar{\mu}_{t}^{(k+\fkk_1,\ell+\fkl_1)}(\xi_{1},\eta_{1},\dots,\xi_{k+\fkk_1},\eta_{k+\fkk_1},\txi_{1},\teta_{1},\dots,\txi_{\ell+\fkl_1},\teta_{\ell+\fkl_1})\\
			&\quad-{n_2}\sum_{j=1}^{\fks_2(k,\ell)}\int\dd\eta_{2,\fks_2(k,\ell)+1}\dd\xi_{2,\fks_2(k,\ell)+1}\delta(\xi_{2,\fks_2(k,\ell)+1}-\eta_{2,\fks_2(k,\ell)+1}t)\,\widehat{V}_{22}(-\eta_{2,\fks_2(k,\ell)+1})\,\eta_{2,\fks_2(k,\ell)+1}(\xi_{2,j}-\teta_{2,j}t)\\
			&\quad\qquad\times\theta_{2,j}(\xi_{2,\fks_2(k,\ell)+1},\eta_{2,\fks_2(k,\ell)+1})\bar{\mu}_{t}^{(k+\fkk_2,\ell+\fkl_2)}(\xi_{1},\eta_{1},\dots,\xi_{k+\fkk_2},\eta_{k+\fkk_2},\txi_{1},\teta_{1},\dots,\txi_{\ell+\fkl_2},\teta_{\ell+\fkl_2})\\
			&=-\sum_{\alpha,\beta\in\{1,2\}}{n_\beta}\sum_{j=1}^{\fks_\alpha(k,\ell)}\int\dd\eta_{\beta,\fks_\beta(k,\ell)+1}\dd\xi_{\beta,\fks_\beta(k,\ell)+1}\delta(\xi_{\beta,\fks_\beta(k,\ell)+1}-\eta_{\beta,\fks_\beta(k,\ell)+1}t)\\
			&\quad\qquad\widehat{V}_{\alpha\beta}(-\eta_{\beta,\fks_\beta(k,\ell)+1})\,\eta_{\beta,\fks_\beta(k,\ell)+1}(\xi_{\alpha,j}-\eta_{\alpha,j}t)\\
			&\quad\qquad\times\theta_{\alpha,j}(\xi_{\beta,\fks_\beta(k,\ell)+1},\eta_{\beta,\fks_\beta(k,\ell)+1})\bar{\mu}_{t}^{(k+\fkk_\beta,\ell+\fkl_\beta)}(\xi_{1},\eta_{1},\dots,\xi_{k+\fkk_\beta},\eta_{k+\fkk_\beta},\txi_{1},\teta_{1},\dots,\txi_{\ell+\fkl_\beta},\teta_{\ell+\fkl_\beta}).
		\end{align*}
		With this result, we have that:
		\begin{equation}\label{eq:delta_Lkl}
			\begin{aligned}
				& (\Delta_{L}\bar{\mu}_{t})^{(k,\ell)}(\xi_{1},\eta_{1},\dots,\xi_{k},\eta_{k},\txi_{1},\teta_{1},\dots,\txi_{\ell},\teta_{\ell})\\
				& =(-1)^{L}\sum_{\alpha_1,\beta_1\in\{1,2\}}\dots\sum_{\alpha_L,\beta_L\in\{1,2\}}\sum_{j_1=1}^{\fks_{\alpha_1}(k,\ell)}\cdots\sum_{j_L=1}^{\fks_{\alpha_L}(k+\fkK_{L-1},\ell+\fkL_{L-1})}\int\limits_{0}^{t}\dd t_{1}\dots\int\limits_{0}^{t_{L}-1}\dd t_{L}\\
				&\quad \prod_{m=1}^L \dd\eta_{{\beta_m},\fks_{\beta_m}(k+\fkK_{m},\ell+\fkL_{m})}\dd\xi_{{\beta_m},\fks_{\beta_m}(k+\fkK_{m},\ell+\fkL_{m})}\delta(\xi_{{\beta_m},\fks_{\beta_m}(k+\fkK_{m},\ell+\fkL_{m})}-\eta_{{\beta_m},\fks_{\beta_m}(k+\fkK_{m},\ell+\fkL_{m})}t_m)\\
				&\quad\Big\{\widehat{V}_{{\alpha_1}{\beta_1}}(-\eta_{{\beta_1},\fks_{\beta_1}(k+\fkK_1,\ell+\fkL_1)})\,\eta_{{\beta_1},\fks_{\beta_1}(k+\fkK_1,\ell+\fkL_1)}(\xi_{{\alpha_1},{j_1}}-\eta_{{\alpha_1},{j_1}}t_1)\theta_{{\alpha_1},{j_1}}(\xi_{{\beta_1},\fks_{\beta_1}(k+\fkK_1,\ell+\fkL_1)},\eta_{{\beta_1},\fks_{\beta_1}(k+\fkK_1,\ell+\fkL_1)})\\
				&\cdots\widehat{V}_{{\alpha_L}{\beta_L}}(-\eta_{{\beta_L},\fks_{\beta_L}(k+\fkK_L,\ell+\fkL_L)})\,\eta_{{\beta_L},\fks_{\beta_L}(k+\fkK_L,\ell+\fkL_L)}(\xi_{{\alpha_L},{j_L}}-\eta_{{\alpha_L},{j_L}}t_L)\theta_{{\alpha_L},{j_L}}(\xi_{{\beta_L},\fks_{\beta_L}(k+\fkK_L,\ell+\fkL_L)},\eta_{{\beta_L},\fks_{\beta_L}(k+\fkK_L,\ell+\fkL_L)})\\
				&\quad\bar{\mu}_{t}^{(k+\fkK_L,\ell+\fkL_L)}(\xi_{1},\eta_{1},\dots,\xi_{k+\fkK_L},\eta_{k+\fkK_L},\txi_{1},\teta_{1},\dots,\txi_{\ell+\fkL_L},\teta_{\ell+\fkL_L})\Big\}.
			\end{aligned}
		\end{equation}
		Due to the complexity of the symbols, for a better understanding of the above expression, we will select the first branch, $\alpha_m=\beta_m=1$, at each step to obtain a sub-branch, which will be demonstrated as a special case as follows:
		\begin{equation}
			\begin{aligned}
				&\left(-1\right)^{L}\sum_{j_1=1}^{k}\cdots\sum_{j_L=1}^{k+L-1}\int\limits_{0}^{t}\dd t_{1}\dots\int\limits_{0}^{t_{L}-1}\dd t_{L}\prod_{m=1}^L \dd\eta_{k+m}\dd\xi_{k+m}\delta(\xi_{k+m}-\eta_{k+m}t_m)\\
				&\quad\Big\{\widehat{V}_{11}(-\eta_{k+1})\,\eta_{k+1}(\xi_{{j_1}}-\eta_{{j_1}}t_1)\theta_{1,{j_1}}(\xi_{k+1},\eta_{k+1})
				\cdots\widehat{V}_{11}(-\eta_{k+L})\,\eta_{k+L}(\xi_{{j_L}}-\eta_{{j_L}}t_L)\theta_{1,{j_L}}(\xi_{k+L},\eta_{k+L})\\
				&\quad\bar{\mu}_{t}^{(k+L,\ell)}(\xi_{1},\eta_{1},\dots,\xi_{k+L},\eta_{k+L},\txi_{1},\teta_{1},\dots,\txi_{\ell},\teta_{\ell})\Big\}.
			\end{aligned}
		\end{equation}
		This term is the same as the results obtained for the case with only one species of particles \cite{Narnhofer1981}. Using equations \eqref{eq:theta_1} and \eqref{eq:theta_2}, it can be obtained that the quantity inside the braces in equation \eqref{eq:delta_Lkl} is equal to
		\begin{align}
			&\widehat{V}_{{\alpha_1}{\beta_1}}(-\eta_{{\beta_1},\fks_{\beta_1}(k+\fkK_1,\ell+\fkL_1)})\,\eta_{{\beta_1},\fks_{\beta_1}(k+\fkK_1,\ell+\fkL_1)}(\xi_{{\alpha_1},{j_1}}-\eta_{{\alpha_1},{j_1}}t_1)\nonumber\\
			&\widehat{V}_{{\alpha_2}{\beta_2}}(-\eta_{{\beta_2},\fks_{\beta_2}(k+\fkK_2,\ell+\fkL_2)})\,\eta_{{\beta_2},\fks_{\beta_2}(k+\fkK_2,\ell+\fkL_2)}(\xi_{{\alpha_2},{j_2}}-\eta_{{\alpha_2},{j_2}}t_2-\delta_{\alpha_1\alpha_2}\delta_{j_1 j_2}(\xi_{{\beta_1},\fks_{\beta_1}(k+\fkK_1,\ell+\fkL_1)}-\eta_{{\beta_1},\fks_{\beta_1}(k+\fkK_1,\ell+\fkL_1)}t_2))\nonumber\\
			&\cdots\widehat{V}_{{\alpha_L}{\beta_L}}(-\eta_{{\beta_L},\fks_{\beta_L}(k+\fkK_L,\ell+\fkL_L)})\,\eta_{{\beta_L},\fks_{\beta_L}(k+\fkK_L,\ell+\fkL_L)}\nonumber\\
			&\qquad(\xi_{{\alpha_L},{j_L}}-\eta_{{\alpha_L},{j_L}}t_L-\sum_{m=1}^{L-1}\delta_{\alpha_m\alpha_L}\delta_{j_m j_L}(\xi_{{\beta_m},\fks_{\beta_m}(k+\fkK_m,\ell+\fkL_m)}-\eta_{{\beta_m},\fks_{\beta_m}(k+\fkK_m,\ell+\fkL_m)}t_L))\nonumber\\
			&\big[\theta_{{\alpha_1},{j_1}}(\xi_{{\beta_1},\fks_{\beta_1}(k+\fkK_1,\ell+\fkL_1)},\eta_{{\beta_1},\fks_{\beta_1}(k+\fkK_1,\ell+\fkL_1)})
			\cdots
			\theta_{{\alpha_L},{j_L}}(\xi_{{\beta_L},\fks_{\beta_L}(k+\fkK_L,\ell+\fkL_L)},\eta_{{\beta_L},\fks_{\beta_L}(k+\fkK_L,\ell+\fkL_L)})\nonumber\\
			&\bar{\mu}_{t}^{(k+\fkK_L,\ell+\fkL_L)}(\xi_{1},\eta_{1},\dots,\xi_{k+\fkK_L},\eta_{k+\fkK_L},\txi_{1},\teta_{1},\dots,\txi_{\ell+\fkL_L},\teta_{\ell+\fkL_L})\big].\label{eq:term_in_braces}
		\end{align}
		Since $\bar{\mu}_{t}^{(k+\fkK_L,\ell+\fkL_L)}$ is a characteristic function, the absolute value of the expression inside the brackets in \eqref{eq:term_in_braces} cannot exceed 1. Additionally, the presence of the Dirac delta functions in \eqref{eq:delta_Lkl} enables us to substitute $\xi_{{\beta_m},\fks_{\beta_m}(k+\fkK_{m},\ell+\fkL_{m})}$ with $\eta_{{\beta_m},\fks_{\beta_m}(k+\fkK_{m},\ell+\fkL_{m})}t_m$, for $m\leq L$, within \eqref{eq:term_in_braces}. As a result, the modulus of the expression in \eqref{eq:term_in_braces} is restricted by
		\begin{equation}\label{eq:term_in_braces_bdd}
			\begin{aligned}
				& \Big|\widehat{V}_{{\alpha_1}{\beta_1}}(-\eta_{{\beta_1},\fks_{\beta_1}(k+\fkK_1,\ell+\fkL_1)})\,\eta_{{\beta_1},\fks_{\beta_1}(k+\fkK_1,\ell+\fkL_1)}(\xi_{{\alpha_1},{j_1}}-\eta_{{\alpha_1},{j_1}}t_1)\Big|\\
				&\Big|\widehat{V}_{{\alpha_2}{\beta_2}}(-\eta_{{\beta_2},\fks_{\beta_2}(k+\fkK_2,\ell+\fkL_2)})\,\eta_{{\beta_2},\fks_{\beta_2}(k+\fkK_2,\ell+\fkL_2)}(\xi_{{\alpha_2},{j_2}}-\eta_{{\alpha_2},{j_2}}t_2+\delta_{\alpha_1\alpha_2}\delta_{j_1 j_2}\eta_{{\beta_1},\fks_{\beta_1}(k+\fkK_1,\ell+\fkL_1)}(t_2-t_1))\Big|\\
				&\cdots\Big|\widehat{V}_{{\alpha_L}{\beta_L}}(-\eta_{{\beta_L},\fks_{\beta_L}(k+\fkK_L,\ell+\fkL_L)})\,\eta_{{\beta_L},\fks_{\beta_L}(k+\fkK_L,\ell+\fkL_L)}\\
				&\qquad(\xi_{{\alpha_L},{j_L}}-\eta_{{\alpha_L},{j_L}}t_L+\sum_{m=1}^{L-1}\delta_{\alpha_m\alpha_L}\delta_{j_m j_L}\eta_{{\beta_m},\fks_{\beta_m}(k+\fkK_m,\ell+\fkL_m)}(t_L-t_m))\Big|.
			\end{aligned}
		\end{equation}
		Given the assumptions on $\widehat{V}_{\alpha\beta}$, we can define finite constants as $A:=\sup|\eta\widehat{V}_{\alpha\beta}(\eta)|$, $B:=\sup\{|\eta||\eta\in\supp\widehat{V}_{\alpha\beta}\}$, $C:=\max\Big\{|\xi_1|,\dots,|\xi_k|,|\txi_1|,\dots,|\txi_\ell|\Big\}$, and $D:=\max\Big\{|\eta_1|,\dots,|\eta_k|,|\teta_1|,\dots,|\teta_\ell|\Big\}$. Considering that the volume of $\supp\widehat{V}$ is constrained by $\frac43\pi B^3$ and employing \eqref{eq:term_in_braces_bdd}, following some straightforward computations, we obtain:
		\begin{align*}
			&\left|\Delta_{L}\bar{\mu}_{t}^{(k,\ell)}(\xi_{1},\eta_{1},\dots,\xi_{k},\eta_{k},\txi_{1},\teta_{1},\dots,\txi_{\ell},\teta_{\ell})\right|\\
			&\leq \left(4\cdot\frac43\pi B^3 A\right)^{L}((C+D|t|)(k+\ell))((C+D|t|)(k+\ell)+(k+\ell+1)B|t|)\\
			&\qquad\cdots((C+D|t|)(k+\ell)+(k+\ell+L-1)B|t|)\frac{|t|^{L}}{L!}.
		\end{align*}
		
		Therefore, we can establish that as $L$ tends to infinity, $\Delta_{L}\bar{\mu}_{t}^{(k,\ell)}(\xi_{1},\eta_{1},\dots,\xi_{k},\eta_{k},\txi_{1},\teta_{1},\dots,\txi_{\ell},\teta_{\ell})$ converges to zero. This holds for all $(\xi_{1},\eta_{1},\dots,\xi_{k},\eta_{k},\txi_{1},\teta_{1},\dots,\txi_{\ell},\teta_{\ell})\in(\bR^3)^{2k}\times(\bR^3)^{2\ell}$ and $|t|\leq\tau$, where $\tau$ is chosen to be $\left(\frac{2^5 }{3}\pi B^4 A\right)^{-1/2}$. As a result, for $|t|\leq \tau$, the equation \eqref{eq: equivalent to the equation for Vlasov hierarchy} admits a unique solution, represented by:
		
		\begin{equation}
			\bar{\mu}_t = \bar{\mu}_0 + \sum_{j=1}^{\infty} \int_{0}^{t} \dt_1 \dots \int_{0}^{t_j - 1} \dt_j \, K(t_1) \dots K(t_j) \, \bar{\mu}_{0}.
		\end{equation}
		
		Since the value of $\tau$ does not depend on $\bar{\mu}_0$, we can similarly establish the uniqueness of $\bar{\mu}_t$, and consequently, of $\mu_{\infty,\infty,t}$ (and $m_{\infty,\infty,t}$ simultaneously) over successive time intervals $[n\tau, (n + 1)\tau]$ and $[- (n + 1)\tau, -n\tau]$ for all $n\in \bN$. This, in turn, allows us to assert that the hierarchy \eqref{eq:BBGKY_limit} possesses a unique global solution.
	\end{proof}
	
	\section*{Acknowledgement}
	L.C. and J.L was partially funded by Center for Advanced Studies (CAS) at Ludwig-Maximilians-Universität (LMU) München. 
	J.L. was supported by the European Research Council (ERC CoG RAMBAS, Project Nr. 101044249), the Deutsche Forschungsgemeinschaft (DFG, German Research Foundation) – TRR 352 – Project-ID 470903074, 
	the Swiss National Science Foundation through the NCCR SwissMAP, the SNSF Eccellenza project PCEFP2\_181153, 
	by the Swiss State Secretariat for Research and Innovation through the project P.530.1016 (AEQUA), and
	Basic Science Research Program through the National Research Foundation of Korea(NRF) funded by the Ministry of Education (RS-2024-00411072).
	
	\section*{Declarations}
	\begin{itemize}
		\item Conflict of interest: The Authors have no conflicts of interest to declare that are relevant to the content of this article.
		\item Data availability: Data sharing is not applicable to this article as no datasets were generated or analysed during the current study.
	\end{itemize}
	
	\bibliographystyle{abbrv}
	\bibliography{refs}
	
\end{document}